\documentclass[dvipsnames, 11pt]{article}
\usepackage[english]{babel}
\usepackage[utf8x]{inputenc}

\RequirePackage{amsthm,amsmath}
\RequirePackage{natbib}
\RequirePackage[colorlinks,citecolor=blue,urlcolor=blue, linkcolor = blue ]{hyperref}
\usepackage{tikz,pgf}
\usepackage{multirow, multicol}
\usepackage{float}
\usepackage{amsfonts}
\usepackage{setspace}
\usepackage{amsmath,graphicx,centernot}
\newcommand{\bigCI}{\mathrel{\text{\scalebox{1.07}{$\perp\mkern-10mu\perp$}}}}
\newcommand{\nbigCI}{\centernot{\bigCI}}
\usepackage{tkz-euclide}
\usepackage[font=singlespacing]{caption}
\newtheorem{theorem}{Theorem}[section]

\newtheorem{lemma}[theorem]{Lemma}
\theoremstyle{definition}
\newtheorem{definition}{Definition}[section]
\theoremstyle{remark}

\usepackage{verbatim}
\usepackage{subcaption}
\usepackage{soul}
\usepackage{epstopdf}
\usepackage{colortbl}

\usepackage{bm}
\usepackage[plain,noend]{algorithm2e}

\makeatletter
\renewcommand{\algocf@captiontext}[2]{#1\algocf@typo. \AlCapFnt{}#2} 
\def\@algocf@capt@plain{top}
\renewcommand{\algocf@makecaption}[2]{%
  \addtolength{\hsize}{\algomargin}%
  \sbox\@tempboxa{\algocf@captiontext{#1}{#2}}%
  \ifdim\wd\@tempboxa >\hsize
    \hskip .5\algomargin%
    \parbox[t]{\hsize}{\algocf@captiontext{#1}{#2}}
  \else%
    \global\@minipagefalse%
    \hbox to\hsize{\box\@tempboxa}
  \fi%
  \addtolength{\hsize}{-\algomargin}%
}
\makeatother

\def\AR{\textsc{AR}}
\def\CLR{\textsc{CLR}}

\def\TSLS{\textsc{TSLS}}

\def\TSLS{\textsc{TSLS}}
\def\SAR{\textsc{SAR}}

\title{Two Robust Tools for Inference about Causal Effects with Invalid Instruments}
\author{Hyunseung Kang*, Youjin Lee*, T. Tony Cai, Dylan S. Small*}
\date{}

\usepackage[margin=1.0in]{geometry}

\begin{document}
\maketitle

\begin{abstract}
Instrumental variables have been widely used to estimate the causal effect of a treatment on an outcome. Existing confidence intervals for causal effects based on instrumental variables assume that all of the putative instrumental variables are valid; a valid instrumental variable is a variable that affects the outcome only by affecting the treatment and is not related to unmeasured confounders. However, in practice, some of the putative instrumental variables are likely to be invalid. This paper presents two tools to conduct valid inference and tests in the presence of invalid instruments. First, we propose a simple and general approach to construct confidence intervals based on taking unions of well-known confidence intervals. Second, we propose a novel test for the null causal effect  based on a collider bias. Our two proposals, especially when fused together, outperform traditional instrumental variable confidence intervals when invalid instruments are present, and can also be used as a sensitivity analysis when there is concern that instrumental variables assumptions are violated. The new approach is applied to a Mendelian randomization study on the causal effect of low-density lipoprotein on the incidence of cardiovascular diseases.
\end{abstract}

Keywords: Anderson-Rubin test; Confidence interval; Invalid instrument; Instrumental variable; Likelihood ratio test; Weak instrument.

*: Denotes equal contribution.

\newpage
\setcounter{page}{1}
\setstretch{1.25}

\section{Introduction} \label{sec:intro}

Instrumental variables have been a popular method to estimate the causal effect of a treatment, exposure, or policy on an outcome when unmeasured confounding is present \citep{angrist_identification_1996, hernan2006instruments, baiocchi_instrumental_2014}. Informally speaking, the method relies on having instruments that are (A1) related to the exposure, (A2) only affect the outcome by affecting the exposure (no direct effect), and (A3) are not related to unmeasured confounders that affect the exposure and the outcome (see Section \ref{sec:model} for details). Unfortunately, in many applications, practitioners are unsure if all of the candidate instruments satisfy these assumptions. For example, in Mendelian randomization \citep{davey_smith_mendelian_2003, davey_smith_mendelian_2004}, a subfield of genetic epidemiology, the candidate instruments are genetic variants which are associated with the exposure. But, these instruments may have direct effects on the outcome, an effect known as pleiotropy, and violate (A2) \citep{solovieff_pleiotropy_2013}. Or, the genetic instruments may be in linkage disequilibrium and violate (A3)~\citep{ lawlor_mendelian_2008,burgess_use_2012}. A similar problem arises in economics where some instruments may not be exogenous, where being exogenous is a combination of assumptions (A2) and (A3) \citep{murray_avoiding_2006,conley_plausibly_2012}.


Violation of (A1), known as the weak instrument problem, has been studied in detail; see \citet{stock_survey_2002} for a survey. In contrast, the literature on violations of (A2) and (A3), known as the invalid instrument problem \citep{murray_avoiding_2006} is  limited. \citet{andrews_consistent_1999} and \citet{andrews_consistent_2001} considered selecting valid instruments within context of the generalized method of moments \citep{hansen_large_1982}, but not inferring a treatment effect after selection of valid instruments. \citet{liao_adaptive_2013} and \citet{cheng_select_2015} considered estimating a treatment effect when there is, a priori, a known, specified set of valid instruments and another set of instruments which may not be valid. 
\citet{conley_plausibly_2012} proposed different approaches to assess violations of (A2) and (A3). \citet{small_sensitivity_2007} considered a sensitivity analysis to assess violations of (A2) and (A3). Our work is most closely related to the recent works by \citet{kolesar_identification_2013}, \citet{kang_instrumental_2015}, \citet{bowden_mendelian_2015}, \citet{guo2018confidence}, \citet{windmeijer2018confidence}, \citet{zhao2018statistical} and \citet{windmeijer2019use}. \citet{kolesar_identification_2013} and \citet{bowden_mendelian_2015} considered the case when the instruments violate (A2) and proposed an orthogonality condition where the instruments' effect on the exposure are orthogonal to their effects on the outcome. \citet{kang_instrumental_2015} considered violations of (A2) and (A3) based on imposing an upper bound on the number of invalid instruments among candidate instruments, without knowing which instruments are invalid a priori or without imposing assumptions about instruments' effect like \citet{kolesar_identification_2013} and \citet{bowden_mendelian_2015}. \citet{windmeijer2019use}, under similar settings, discussed consistent selection of the invalid instruments and proposed a median-Lasso estimator that is consistent when less than 50\% candidate instruments are invalid. In addition, \citet{guo2018confidence} proposed sequential hard thresholding to select strong and valid instruments and provided valid confidence intervals for the treatment effect. \citet{windmeijer2018confidence} used multiple confidence intervals to select a set of valid instruments and to construct a valid confidence interval for the treatment effect. 

Instead of first selecting a set of valid or invalid instruments and subsequently testing the treatment effect, our paper directly focuses on testing the effect with invalid instruments by proposing two methods. First, we propose a simple and general confidence interval procedure based on taking unions of well-known confidence intervals and show that it achieves correct coverage rates in the presence of invalid instruments. Second, we propose a novel test for the null hypothesis of no treatment effect in the presence of multiple invalid instruments by leveraging a collider that arises with invalid instruments; we call this test the collider bias test. The null distribution of the collider bias test only depends on the number of valid instruments. Our two methods can also be interpreted as sensitivity analysis.  The usual instrumental variable analysis makes the assumption that all instrumental variables are valid.  Our methods allow one to relax this assumption and see how sensitive the results are by varying the parameter $\bar{s}$ that indicates the number of invalid instruments and observing how the proposed inferential quantities change from $\bar{s}=1$ (i.e. no invalid instruments) to $\bar{s} = L$ (i.e. at most $L-1$ instruments are valid). We also demonstrate that combining the two methods can produce a more powerful inferential procedure than if each method were used alone. We conclude by demonstrating our methods in both synthetic and real datasets.

\section{Setup}

\subsection{Notation}

We use the potential outcomes notation \citep{neyman1923application, rubin_estimating_1974} for instruments laid out in \citet{holland_causal_1988}. Specifically, let there be $L$ potential candidate instruments and $n$ individuals in the sample. Let $Y_i^{(d,\mathbf{z})}$ be the potential outcome if individual $i$ had exposure $d$, a scalar value, and instruments $\mathbf{z}$, an $L$ dimensional vector. Let $D_i^{(\mathbf{z})}$ be the potential exposure if individual $i$ had instruments $\mathbf{z}$. For each individual, we observe the outcome $Y_i$, the exposure, $D_i$, and instruments $\mathbf{Z}_{i\cdot}$. In total, we have $n$ observations of $(Y_i, D_i,\mathbf{Z}_{i\cdot})$. We denote $\mathbf{Y}_{n} = (Y_1,\ldots, Y_n)$ and $\mathbf{D}_{n} = (D_1,\ldots,D_n)$ to be vectors of $n$ observations. Finally, we denote $\mathbf{Z}_{n}$ to be an $n$ by $L$ matrix where row $i$ consists of $\mathbf{Z}_{i\cdot}^T$ and $\mathbf{Z}_{n}$ is assumed to have full rank. 

For a subset $A \subseteq \{1,\ldots,L\}$, denote its cardinality $c(A)$ and $A^C$ its complement. Let $\mathbf{Z}_{A}$ be an $n$ by $c(A)$ matrix of instruments where the columns of $\mathbf{Z}_A$ are from the set $A$. Similarly, for any $L$ dimensional vector $\boldsymbol{\pi}$, let $\boldsymbol{\pi}_A$ only consist of elements of the vector $\boldsymbol{\pi}$ determined by the set $A \subseteq \{1,\ldots,L\}$.

\subsection{Model and Definition of Valid Instruments} \label{sec:model}

For two possible values of the exposure $d',d$ and instruments $\mathbf{z}',\mathbf{z}$, we assume the following potential outcomes model
\begin{equation} \label{eq:model1}
Y_{i}^{(d',\mathbf{z}')} - Y_{i}^{(d,\mathbf{z})} =  (\mathbf{z}' - \mathbf{z})^T \boldsymbol{\phi}^* + (d' - d) \beta^*,\quad E\{Y_i^{(0,0)} \mid \mathbf{Z}_{i \cdot}\} =  \mathbf{Z}_{i \cdot}^T \boldsymbol{\psi}^* 
\end{equation}
where $\boldsymbol{\phi}^*, \boldsymbol{\psi}^*$, and $\beta^*$ are unknown parameters. The parameter $\beta^*$ represents the causal parameter of interest, the causal effect (divided by $d' - d$) of changing the exposure from $d'$ to $d$ on the outcome. The parameter $\boldsymbol{\phi}^*$ represents violation of (A2), the direct effect of the instruments on the outcome. If (A2) holds, then $\boldsymbol{\phi}^* = 0$. The parameter $\boldsymbol{\psi}^*$ represents violation of (A3), the presence of unmeasured confounding between the instruments and the outcome. If (A3) holds, then $\boldsymbol{\psi}^* = 0$. 

Let $\boldsymbol{\pi}^* = \boldsymbol{\phi}^* + \boldsymbol{\psi}^*$ and $\epsilon_i = Y_{i}^{(0,0)} - E\{Y_{i}^{(0,0)} \mid \mathbf{Z}_{i\cdot}\}$. When we combine equations \eqref{eq:model1} along with the definition of $\epsilon_i$, we arrive at the observed data model
\begin{equation} \label{eq:model2}
Y_i =  \mathbf{Z}_{i\cdot}^T \boldsymbol{\pi}^* +  D_i \beta^* + \epsilon_i,\quad{} E(\epsilon_i \mid \mathbf{Z}_{i\cdot}) = 0
\end{equation}
The observed model is also known as an under-identified single-equation linear model in econometrics (page 83 of \citet{wooldridge_econometrics_2010}). 
The observed model can have exogenous covariates, say $\mathbf{X}_{i\cdot}$, including an intercept term, and the Frisch-Waugh-Lovell Theorem allows us to reduce the model with covariates to model \eqref{eq:model2} \citep{davidson_estimation_1993}. The parameter $\boldsymbol{\pi}^*$ 
in the observed data model \eqref{eq:model2} combines both violation of (A2), represented by $\boldsymbol{\phi}^*$, and violation of (A3), represented by $\boldsymbol{\psi}^*$. If both (A2) and (A3) are satisfied, then $\boldsymbol{\phi}^* = \boldsymbol{\psi}^* = \mathbf{0}$ and $\boldsymbol{\pi}^* = \mathbf{0}$. Hence, the value of $\boldsymbol{\pi}^*$ captures whether instruments are valid or invalid. Definition \ref{def:validIV} formalizes this idea.

\begin{definition} \label{def:validIV}
	Suppose we have $L$ candidate instruments along with models \eqref{eq:model1}--\eqref{eq:model2}. We say that instrument $j = 1,\ldots,L$ is valid (i.e. satisfy (A2) and (A3)) if $\pi_j^* = 0$ and invalid if $\pi_j^* \neq 0$.
\end{definition}

When there is only one instrument, $L = 1$, Definition \ref{def:validIV} of a valid instrument is identical to the definition of a valid instrument in \citet{holland_causal_1988}. Specifically, assumption (A2), the exclusion restriction, which implies $Y_i^{(d,\mathbf{z})} = Y_i^{(d,\mathbf{z}')}$ for all $d,\mathbf{z},\mathbf{z}'$, is equivalent to $\boldsymbol{\phi}^* = 0$ and assumption (A3), no unmeasured confounding, which means $Y_i^{(d,\mathbf{z})}$ and $D_{i}^{(\mathbf{z})}$ are independent of $\mathbf{Z}_{i \cdot}$ for all $d$ and $\mathbf{z}$, is equivalent to $\boldsymbol{\psi}^* = \mathbf{0}$, implying $\boldsymbol{\pi}^* = \boldsymbol{\phi}^* + \boldsymbol{\psi}^* = \mathbf{0}$. Definition \ref{def:validIV} is also a special case of the definition of a valid instrument in \citet{angrist_identification_1996} where our setup assumes a model with additive, linear, and constant treatment effect $\beta^{*}$. Hence, when multiple instruments, $L > 1$, are present, our models \eqref{eq:model1}--\eqref{eq:model2} and Definition \ref{def:validIV} can be viewed as a generalization of the definition of valid instruments in \citet{holland_causal_1988}. The supplementary materials contain additional discussions of \eqref{eq:model1}--\eqref{eq:model2} and Definition \ref{def:validIV}.

Given the models above, our paper will focus on testing $\beta^*$ in the presence of invalid instruments, or formally
\begin{equation} \label{eq:hyp}
H_0: \beta^* = \beta_0
\end{equation}
for some value of $\beta_0$ when $\bm{\pi}^*$ may not equal to $0$. Specifically, let $s^{*} = 0,\ldots,L - 1$ be the number of invalid instruments and let $\bar{s}$ be an upper bound on $s^{*}$ plus 1, $s^* \leq \bar{s} +1$ or $s^* < \bar{s}$; the number of invalid instruments is assumed to be less than $\bar{s}$.
Let $v^{*} = L - s^{*}$ be the number of valid instruments. We assume that there is at least one valid IV, even if we don't know which among the $L$ instruments are valid. 
This setup was considered in \citet{kang_instrumental_2015} as a relaxation to traditional instrumental variables setups where one knows exactly which instruments are valid. Also, in Mendelian randomization where instruments are genetic, the setup represents a way for a genetic epidemiologist to impose prior beliefs about instruments' validity. For example, based on the investigator's expertise and prior genome wide association studies, the investigator may provide an upper bound $\bar{s}$, with a small $\bar{s}$ representing an investigator's strong confidence that most of the $L$ candidate instruments are valid and a large $\bar{s}$ representing an investigator's weak confidence in the candidate instruments' validity.

In the absence of prior belief about $\bar{s}$, the setup using the additional parameter $\bar{s}$ can also be viewed as a sensitivity analysis common in causal inference. In particular, similar to the sensitivity analysis presented in \citet{rosenbaum_observational_2002}, we can treat $\bar{s}$ as the sensitivity parameter and vary from $\bar{s} = 1$ to $\bar{s} = L$ where $\bar{s} = 1$ represents the traditional case where all instruments satisfy (A2) and (A3) and $\bar{s} = L$ represents the worst case where at most $L - 1$ instruments may violate (A2) and (A3). For each $\bar{s}$, we can construct confidence intervals from our two proposed methods below and observe how increasing violations of instrumental variables assumptions through increasing $\bar{s}$ impact the resulting conclusions about  $\beta^*$. Also, similar to a typical sensitivity analysis, we can find the smallest $\bar{s}$ that retains the null hypothesis of no causal effect.
If at $\bar{s} = L$, our methods still reject the null, then the conclusion about the causal effect $\beta^*$ is insensitive to violations of assumptions (A2) and (A3).

\section{Method 1: Union Confidence Interval With Invalid Instruments}
\label{sec:CI}

\subsection{Procedure}
\label{sec:genprocedure}

Let $B^* \subset \{1,\ldots,L\}$ be the true set of invalid instruments.
In the instrumental variables literature, there are many test statistics $T(\beta_0,B^*)$ of the null hypothesis in \eqref{eq:hyp} if $B^*$ is known. Some examples include the test statistic based on two-stage least squares, the Anderson-Rubin test \citep{anderson_estimation_1949}, and the conditional likelihood ratio test \citep{moreira_conditional_2003}; see the supplementary materials for details of these test statistics and additional test statistics in the literature. By the duality of hypothesis testing and confidence intervals, inverting any of the aforementioned test statistic $T(\beta_0,B^*)$ under size $\alpha$ provides a $1 - \alpha$ confidence interval for $\beta^*$, which we denote as $C_{1 - \alpha}(\mathbf{Y}_{n},\mathbf{D}_{n},\mathbf{Z}_{n},B^*)$
\begin{equation} \label{eq:confIntInversion}
C_{1 - \alpha}(\mathbf{Y}_{n},\mathbf{D}_{n},\mathbf{Z}_{n},B^*) = \{\beta_0 \mid T(\beta_0,B^*) \leq q_{1 - \alpha}\}
\end{equation}
Here, $q_{1 - \alpha}$ is the $1 - \alpha$ quantile of the null distribution of $T(\beta_0,B^*)$. 

Unfortunately, in our setup, we do not know the true set $B^*$ of invalid instruments, so we cannot directly use \eqref{eq:confIntInversion} to estimate confidence intervals of $\beta^*$. However, from Section \ref{sec:model}, we have a constraint on the number of invalid instruments, $s^* < \bar{s}$. We can use this constraint to take unions of $C_{1 - \alpha}(\mathbf{Y}_{n},\mathbf{D}_{n},\mathbf{Z}_{n},B)$ over all possible subsets of instruments $B \subset \{1,\ldots,L\}$ where $c(B) < \bar{s}$. The confidence interval with the true set of invalid instruments $C(\mathbf{Y}_{n},\mathbf{D}_{n},\mathbf{Z}_{n},B^*)$ will be in this union since $s^{*} = c(B^*) < \bar{s}$. Our proposal, which we call the union method, is exactly this except we restrict the subsets $B$ in the union to be only of size $c(B) = \bar{s} - 1$.
\begin{equation} \label{eq:confIntProposal}
C_{1 -\alpha}(\mathbf{Y}_{n},\mathbf{D}_{n},\mathbf{Z}_{n}) = \cup_{B,c(B) = \bar{s} - 1} \{C_{1 - \alpha}(\mathbf{Y}_{n},\mathbf{D}_{n},\mathbf{Z}_{n},B) \} 
\end{equation}
Theorem \ref{theorem1} shows that the confidence interval in \eqref{eq:confIntProposal} has the proper coverage in the presence of  invalid instruments.
\begin{theorem}
	\label{theorem1}
	Suppose model \eqref{eq:model2} holds and $s^* < \bar{s}$. Given $\alpha \in (0,1)$, consider any test statistic $T(\beta_0,B$) with the property that for any $B^* \subseteq B$, $T(\beta_0,B)$ has size at most $\alpha$ under the null hypothesis in \eqref{eq:hyp}. Then, $C_{1 - \alpha}(\mathbf{Y}_{n}, \mathbf{D}_{n},\mathbf{Z}_{n})$ in \eqref{eq:confIntProposal} has at least $1 - \alpha$ coverage of $\beta^*$.
\end{theorem}
The proof is in the appendix. The proposed confidence interval $C_{1 - \alpha}(\mathbf{Y}_{n},\mathbf{D}_{n},\mathbf{Z}_{n})$ is not only robust to the presence of invalid instruments, but also simple and general. Specifically, for any test statistic $T(\beta_0,B)$ with a valid size, such as those mentioned above, one simply takes unions of confidence intervals of $T(\beta_0,B)$ over subsets of instruments $B$ where $c(B) = \bar{s} - 1$. In addition, a key feature of our procedure is that we do not have to iterate through all subsets of instruments where $c(B) < \bar{s}$; we only have to examine the largest possible set of invalid instruments, $c(B) = \bar{s} - 1$, to guarantee at least $1 - \alpha$ coverage. 

A caveat to our procedure is computational feasibility. Even though we restrict the union to subsets of exactly size $c(B) = \bar{s} - 1$, if there are many candidate instruments $L$ and $\bar{s}$ is moderately large, $C_{1-\alpha}(\mathbf{Y}_{n}, \mathbf{D}_{n},\mathbf{Z}_{n})$ becomes computationally burdensome. However, in many instrumental variables studies, it is difficult to find good candidate instruments that are both strong and plausibly valid. In economic applications, the number of good instruments rarely exceeds $L= 20$. In some, but not all, Mendelian randomization studies, after linkage disequilibrium clumping and p-value thresholding, $L$ remains small. In these two cases, our procedure in \eqref{eq:confIntProposal} is computationally tractable.


\subsection{A Shorter Interval With Pretesting} \label{sec:pretest}
As shown in Theorem \ref{theorem1}, our proposed interval $C_{1-\alpha}(\mathbf{Y}_{n},\mathbf{D}_{n},\mathbf{Z}_{n})$ achieves the desired coverage level by taking unions of confidence intervals $C_{1-\alpha}(\mathbf{Y}_{n},\mathbf{D}_{n},\mathbf{Z}_{n},B)$ over all $c(B) = \bar{s} - 1$. Some of these subsets $B$ have every invalid instrument, leading to unbiased confidence intervals (i.e. contain $\beta^*$ with probability greater than or equal to $1 -\alpha$). But, other subsets may not have every invalid instrument, leading to biased confidence intervals. Then, taking the union of both types of confidence intervals may elongate $C_{1-\alpha}(\mathbf{Y}_{n},\mathbf{D}_{n},\mathbf{Z}_{n})$ since we only need one unbiased confidence interval to have the desired coverage level; in other words, including biased intervals will make $C_{1 -\alpha}(\mathbf{Y}_{n},\mathbf{D}_{n},\mathbf{Z}_{n})$ unnecessarily conservative. In this section, we propose a way to shorten $C_{1-\alpha}(\mathbf{Y}_{n},\mathbf{D}_{n},\mathbf{Z}_{n})$ by pretesting whether each subset $B$ contain invalid instruments.

Formally, for a $1 - \alpha$ confidence interval of $\beta^*$, consider the null hypothesis that $B^C$, for $c(B^C) \geq 2$, contains only valid instruments, $H_0: \boldsymbol{\pi}_{B_C}^* = 0$. Suppose $S(B)$ is a test statistic for this null with level $\alpha_{s} < \alpha$ and $q_{1 - \alpha_{s}}$ is the $1 - \alpha_{s}$ quantile of the null distribution of $S(B)$. Let $\alpha_t = \alpha - \alpha_s$ be the confidence level for $C_{1-\alpha_t}(\mathbf{Y}_{n},\mathbf{D}_{n},\mathbf{Z}_{n},B)$. Then, a $1 -\alpha$ confidence interval for $\beta^*$ that incorporates the pretest $S(B)$ is
\begin{equation} \label{eq:confIntProposalScreen}
C_{1 - \alpha}'(\mathbf{Y}_{n},\mathbf{D}_{b},\mathbf{Z}_{n}) = \cup_{B} \{C_{1-\alpha_{t}}(\mathbf{Y}_{n},\mathbf{D}_{b},\mathbf{Z}_{n},B) \mid c(B) = \bar{s} - 1, S(B) \leq q_{1-\alpha_{s}} \}
\end{equation}
 For example, if the desired confidence level for $\beta^*$ is 95\% so that $\alpha =  0.05$, we can run the pretest $S(B)$ at $\alpha_s = 0.01$ level and compute $C_{1-\alpha_{t}}(\mathbf{Y}_{n},\mathbf{D}_{n},\mathbf{Z}_{n},B)$ at the $\alpha_t = 0.04$ level. Theorem \ref{theorem2} shows that $C_{1 - \alpha}'(\mathbf{Y}_{n},\mathbf{D}_{n},\mathbf{Z}_{n})$ achieves the desired $1 - \alpha$ coverage of $\beta^*$ in the presence of invalid instruments.
\begin{theorem} \label{theorem2} Suppose the assumptions in Theorem \ref{theorem1} hold. For any pretest $S(B)$ where $c(B^C) \geq 2$ and $S(B)$ has the correct size under the null hypothesis $H_0: \bm{\pi}_{B^C}^* = 0$, $C_{1 - \alpha}'(\mathbf{Y}_{n},\mathbf{D}_{n},\mathbf{Z}_{n})$ has at least $1 - \alpha$. 
\end{theorem}
Similar to Theorem \ref{theorem1}, procedure \eqref{eq:confIntProposalScreen} is general in the sense that any pretest $S(B)$ with the correct size can be used to construct $C'_{1-\alpha}(\mathbf{Y}_{n},\mathbf{D}_{n},\mathbf{Z}_{n})$. For example, the Sargan test \citep{sargan_estimation_1958} can act as a pre-test for \eqref{eq:confIntProposalScreen}; see the supplementary materials for details of the Sargan test. 

Finally, while many tests satisfy the requirements for Theorems \ref{theorem1} and \ref{theorem2}, some tests will be better than others where ``better'' can be defined in terms of statistical power or length of the confidence interval. In the supplementary materials, we characterize the power of common tests in instrumental variables literature when invalid instruments are present and we show that under additional assumptions, the Anderson-Rubin test tends to have better power than the test based on two-stage least squares when invalid instruments are present.

\section{Method 2: A Collider Bias Test With Invalid Instruments}
\subsection{Motivation With Two Instruments}
\label{sec:conditional_test}


In this section, we introduce a new test statistic to test the null hypothesis of no treatment effect when invalid instruments are possibly present, i.e. when $\beta_0 = 0$ in equation \eqref{eq:hyp}. Broadly speaking, the new test is based on recognizing a collider bias in a directed acyclic graph when the null hypothesis of no effect does not hold and there is at least one valid instrument among a candidate set of $L$ instruments. To better illustrate the bias, we start with two, independent candidate instruments $L=2$  where at least one instrument is valid and generalize the idea to $L > 2$.

Suppose $H_0: \beta^* = 0$ holds and consider Figure~\ref{fig:diagram} which illustrates a directed acyclic graph with two mutually independent instruments. Each node indicates a variable and a directed edge connecting two nodes indicates a non-zero direct causal effect. For illustrative purposes, dotted directed edges represent \textit{possibly} non-zero causal effects. The variable $U$ represents an unmeasured confounder between $D$ and $Y$.
\begin{figure}[ht]
	\begin{centering}
		\begin{tikzpicture}[>=stealth, node distance=0.5cm]
		\tikzstyle{format} = [thin, circle, minimum size=5.0mm,
		inner sep=0.1pt]
		\tikzstyle{square} = [thin, rectangle, draw]
		\begin{scope}[xshift=0.0cm, yshift = 0cm]
		\path[->, thick]
		node[format, circle, line width=0.3mm] (z1) {$Z_{1}$}
		node[format, below of=z1, yshift = -2cm, circle, line width=0.3mm] (z2) {$Z_{2}$}
		node[format, right of=z1, xshift = 1cm, yshift = -1cm, circle, line width=0.3mm] (d) {$D$}
		node[format, right of=d, xshift = 1.5cm, circle, line width=0.3mm] (y) {$Y$}
		node[format, right of=d, xshift = 0.7cm, yshift =  -1cm, circle, line width=0.3mm] (u) {$U$}
		node[format, below of=d, yshift =  -1.5cm, circle, line width=0.3mm] (a) {(a)}
		
		
		(z1) edge[thick, black] node[below]{$\gamma^{*}_{1}$} (d)
		(z2) edge[thick, black] node[above]{$\gamma^{*}_{2}$} (d)
		(u) edge[thick, densely dotted]  (d)
		(u) edge[thick, densely dotted] (y)
		(z1) edge[thick, densely dotted,  color = Gray] node[above]{$\pi^{*}_{1}  = 0$} (y)
		(z2) edge[thick, densely dotted,  color = Gray] node[below]{$\pi^{*}_{2} = 0$} (y)
		;
		\end{scope}
		
		\begin{scope}[xshift=4.5cm, yshift = 0cm]
		\path[->, thick]
		node[format, circle, line width=0.3mm] (z1) {$Z_{1}$}
		node[format, below of=z1, yshift = -2cm, circle, line width=0.3mm] (z2) {$Z_{2}$}
		node[format, right of=z1, xshift = 1cm, yshift = -1cm, circle, line width=0.3mm] (d) {$D$}
		node[format, right of=d, xshift = 1.5cm, circle, line width=0.3mm] (y) {$Y$}
		node[format, right of=d, xshift = 0.7cm, yshift =  -1cm, circle, line width=0.3mm] (u) {$U$}
		node[format, below of=d, yshift =  -1.5cm, circle, line width=0.3mm] (b) {(b)}
		
		(z1) edge[very thick, color = RoyalBlue] node[below]{$\gamma^{*}_{1}$} (d)
		(z2) edge[very thick, color = RoyalBlue] node[above]{$\gamma^{*}_{2}$} (d)
		(u) edge[thick, densely dotted]  (d)
		(u) edge[thick, densely dotted] (y)
		(z1) edge[very thick, color = RoyalBlue] node[above]{$\pi^{*}_{1}$} (y)
		(z2) edge[thin, densely dotted, color = Gray] node[below]{$\pi^{*}_{2} = 0$} (y)
		;
		\end{scope}
		
		\begin{scope}[xshift=9.0cm, yshift = 0cm]
		\path[->, thick]
		node[format, circle, line width=0.3mm] (z1) {$Z_{1}$}
		node[format, below of=z1, yshift = -2cm, circle, line width=0.3mm] (z2) {$Z_{2}$}
		node[format, right of=z1, xshift = 1cm, yshift = -1cm, circle, line width=0.3mm] (d) {$D$}
		node[format, right of=d, xshift = 1.5cm, circle, line width=0.3mm] (y) {$Y$}
		node[format, right of=d, xshift = 0.7cm, yshift =  -1cm, circle, line width=0.3mm] (u) {$U$}
		node[format, below of=d, yshift =  -1.5cm, circle, line width=0.3mm] (c) {(c)}
		
		(z1) edge[very thick, color = RoyalBlue] node[below]{$\gamma^{*}_{1}$} (d)
		(z2) edge[very thick, color = RoyalBlue] node[above]{$\gamma^{*}_{2}$} (d)
		(u) edge[thick, densely dotted]  (d)
		(u) edge[thick, densely dotted] (y)
		(z1) edge[thin, densely dotted, color = Gray] node[above]{$\pi^{*}_{1} = 0$} (y)
		(z2) edge[very thick, color = RoyalBlue] node[below]{$\pi^{*}_{2}$} (y)
		;
		\end{scope}
		\end{tikzpicture} 
		\par\end{centering}
	\caption{ \label{fig:diagram} Causal directed acyclic graph with two candidate instruments $Z_{1}$ and $Z_{2}$ when  $H_0: \beta^* = 0$ holds. Solid lines indicate a non-zero causal path and dotted lines indicate a \textit{possibly non-zero} causal path.  A variable $U$ indicates an unmeasured confounder between $D$ and $Y$. We use $\gamma^*$ and $\pi^*$ to label each edge. Our setup supposes that at least one instrument is valid, i.e. $\pi^{*}_{1} \pi^{*}_2 = 0$, without knowing which $\pi^*$ is zero.}
\end{figure}
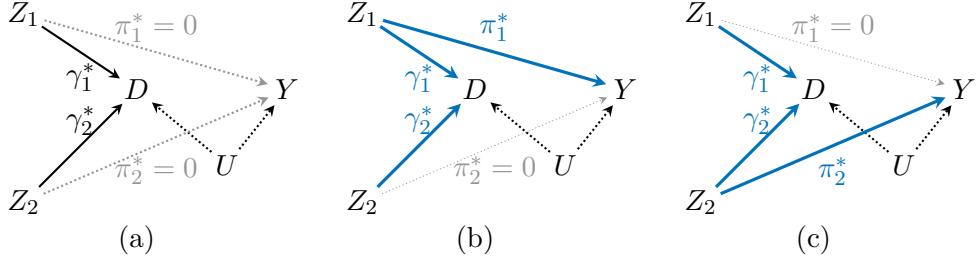
In all three graphs of Figure~\ref{fig:diagram}, $D$ is a collider, but $Y$ is not, thanks to a lack of edge between $D$ and $Y$ under the null hypothesis of no effect. 
It is well known that conditioning on a collider like $D$ induces correlation between two marginally independent variables, in this case $Z_{1}$ and $Z_{2}$;  see~\citet{cole2009illustrating} for one explanation.
But, so long as one instrument is valid so that there is no edge between one of the $Z$'s and $Y$, $Y$ remains a non-collider and $Z_{1}$ and $Z_{2}$ must be conditionally independent on $Y$, $Z_{1} \bigCI Z_{2} | Y$ under $H_0: \beta^* = 0$. Critically, the conditional independence does not require us knowing which instrument is invalid or valid a priori. For example, in Figure~\ref{fig:diagram} (a) where both instruments are valid or in Figures ~\ref{fig:diagram} (b) and (c) where one of the two instruments is invalid, $Y$ is still not a collider and the conditional independence of $Z_{1} \bigCI Z_{2} | Y$ remains true if $H_0: \beta^* = 0$.

The intuition above generalizes to more than two instruments. Formally, let $\{ (\mathbf{Z}, D, Y) : \mathbf{Z}^{T} = (Z_{1}, \ldots, Z_{L}) \in \mathbb{R}^{L} \}$ be a set of random variables containing $L$ instruments, the exposure, and the outcome. Let $\boldsymbol{\Sigma} = (\sigma_{jk}) \in \mathbb{R}^{(L+1) \times (L+1)}$ be the covariance matrix of the instrument-outcome pair $(\mathbf{Z}^{T}, Y)$. Similar to the case with two instruments, if the $L$ instruments $Z_{1}, Z_{2}, \ldots, Z_{L}$ are independent with each other, conditioning on $Y$ does not induce a collider bias between a valid instrument $Z_{j}$ and any other $L-1$ candidate instruments under the null causal effect $H_0: \beta^* = 0$, regardless of whether the $L-1$ candidate instruments are valid or not. Additionally, by using theory of conditional independence in graphs, for each valid instrument $j$, the following equivalences can be formally stated; see \citet{drton2006algebraic} for one example.
\begin{eqnarray}
\label{eq:equivalence}
& Z_{j} \bigCI Z_{k} | Y; k=1,2,\ldots, L,~k \neq j \nonumber \\
\Longleftrightarrow  &  \{Z_{j} \bigCI Z_{k}~\text{and}~Z_{j} \bigCI Y;~k = 1,2,\ldots,L,~k \neq j \} \text{ or } \{ Z_{k}  \bigCI Y; k=1,2,\ldots, L,k\neq j \} 
\end{eqnarray}
Define $H_{0j}$ to be the first condition in~\eqref{eq:equivalence} that involves valid instrument $j$, $H_{0j} = \{Z_{j} \bigCI Z_{k}~\text{and}~Z_{j} \bigCI Y;~k = 1,2,\ldots,L,~k \neq j \}$.  Then, we have
\begin{eqnarray*}
\text{For each valid $Z_j$: } & H_{0j} : Z_{j} \bigCI Z_{k}~\text{and}~Z_{j} \bigCI Y;~k = 1,2,\ldots,L,~k \neq j   \\
\Longleftrightarrow  &  \sigma_{jk} = 0;~ k=1,2,\ldots, L+1,~k \neq j  \\ 
\Longleftrightarrow & ~\sigma_{j,L+1} = 0 
\end{eqnarray*}
The theorem below translates the null hypothesis of no treatment effect into the collection of (conditional) independence tests denoted by $H_{0j}$'s across all $j=1,\ldots,L$ instruments. In particular, it shows that if at least one instrument among $L$ is valid, we only need to test the product of $\sigma_{j,L+1}$'s being zero.
\begin{theorem} \label{thm:hyp_imp} Suppose we have at least one valid instrument among $L$ candidate instruments and all the instruments are independent of each other. Then, the null of no treatment effect $H_0: \beta^* = 0$ is equivalent to the null of
\begin{equation} \label{eq:hyp_joint}
H_0: \prod_{j=1,2,\ldots,L} \sigma_{j,L+1} = 0.
\end{equation}
\end{theorem}
The proof of Theorem \ref{thm:hyp_imp} is in the Appendix. 

While the results above formally rely on independence between instruments,  it is possible to dependent instruments at the expense of having a more complex null hypothesis in \eqref{eq:hyp_joint} that varies depending on the exact nature of dependencies between instruments. Also, from an Mendelian randomization standpoint, we can enforce instruments to be independent of each other by choosing SNPs that are far apart from each other in genetic distance. Finally, our result concerning the collider bias with invalid instruments differs from a recent work by \citet{marden2018implementation} who also proposed to use collider bias, but to test the presence of selection bias using a single instrument. 

The next section discusses a test statistic to test the null hypothesis~\eqref{eq:hyp_joint}. 



\subsection{A Likelihood Ratio Test for Collider Bias}
\label{sec:LRT}
There are a myriad of statistical tests for the null in \eqref{eq:hyp_joint} concerning covariances. In this paper, we adapt the test statistic proposed by \citet{drton2006algebraic} and \citet{drton2009likelihood} which is based on a likelihood ratio test for Gaussian graphical models that allow for some singularity constraints. Specifically, consider the following model
\begin{eqnarray} \label{eq:powerSetup}
	\mathbf{Z}_{i \cdot} & \sim & N(\mathbf{0}, \Sigma_{Z} ), \quad{} \Sigma_{Z} = {\rm diag}( \upsilon^{2}_{1} ,\upsilon^{2}_{2}, \cdots, \upsilon^{2}_{L}) \nonumber \\
	D_i &= & \mathbf{Z}_{i\cdot}^T \boldsymbol{\gamma}^* + \xi_i \\
	Y_i &= & \mathbf{Z}_{i \cdot}^T \boldsymbol{\pi}^* + D_i \beta^* + \epsilon_i, \quad{} E(\epsilon_i, \xi_i | \mathbf{Z}_{i\cdot}) =0 \nonumber \\
	\begin{pmatrix} \epsilon_i \\ \xi_i \end{pmatrix} &\sim & N\left[\begin{pmatrix} 0 \\ 0 \end{pmatrix},\begin{pmatrix} \sigma_2^2  & \rho \sigma_1 \sigma_2 \\ \rho \sigma_1 \sigma_2  & \sigma_1^2 \end{pmatrix}\right] \nonumber
\end{eqnarray}
The setup in \eqref{eq:powerSetup} is a special case of model \eqref{eq:model2} with the additional assumptions that (i) $D_i$ is linearly associated to $\mathbf{Z}_{i\cdot}$, (ii) the error terms are bivariate i.i.d. Normal with an unknown covariance matrix, and (iii) $\Sigma_{Z}$ is diagonal. These additional assumptions are used to derive the asymptotic null distribution of the proposed test in \eqref{eq:CT}, which we call the collider bias test; they are not needed to establish the relationships between the null hypotheses in Theorem \ref{thm:hyp_imp}. Also, in the supplementary materials, we present empirical results when the Normality assumption is relaxed. In particular, we assess the performance of the collider bias test when the instruments and/or the outcome are binary and show that the test's size and power are largely insensitive to violations of the distribution assumptions. 


Let $\mathbf{S}^{(L)} = (s_{jk}) \in \mathbb{R}^{(L+1) \times (L+1)}$ be the sample covariance of $(\mathbf{Z}_{n}, \mathbf{Y}_{n})$. We propose to test \eqref{eq:hyp_joint} by computing the smallest determinant of sub-matrices of the estimated covariance matrix $\mathbf{S}^{(L)}$, i.e.
\begin{eqnarray}
\label{eq:CT}
\lambda_{n} & = & \min\limits_{j=1,2,\ldots, L} \left( n \log \left( \frac{s_{jj} \text{det} (\mathbf{S}^{(L)}_{-j, -j}) }{ \text{det}(\mathbf{S}^{(L)}) } \right) \right)
\end{eqnarray}
We call $\lambda_n$ the collider bias test and Theorem \ref{thm:CT} shows the limiting null distribution of $\lambda_n$.
\begin{theorem}
	\label{thm:CT} Let $\mathbf{W} = (W_{jk})$ be a $L \times L$ symmetric matrix where each entry is an independent $\chi^{2}_{1}$ random variable. Let $V^{*} \subset \{1,2,\ldots, L\}$ bet a set of valid instruments among $L$ candidate instruments and $v^{*} = c(V^{*})$. Under model \eqref{eq:powerSetup} and the null hypothesis of no effect, $H_0: \beta^{*} = 0$, or equivalently under \eqref{eq:hyp_joint}, the collider bias test $\lambda_{n}$ in \eqref{eq:CT} converges to the minimum of  $\chi^{2}_{L}$-distributed random variables, which we denote as $\underline{\chi}^{2}_{L, v^{*}}$
	\begin{equation}
	\lambda_{n} \overset{n \rightarrow \infty}{\longrightarrow} \min\limits_{j \in V^{*}} \left( \sum\limits_{k = 1}^{L} W_{jk} \right)
	:= \underline{\chi}^{2}_{L,v^{*}}.
	\end{equation}
\end{theorem}
For any size $\alpha \in (0,1)$, we can use the asymptotic null distribution of $\lambda_n$ in Theorem \ref{thm:CT} to obtain a critical value $\underline{\chi}^{2}_{L, v^{*}, 1-\alpha}$, which is the $1-\alpha$ quantile of the $\underline{\chi}^{2}_{L, v^{*}}$ distribution. We would reject the null of no effect if the observed $\lambda_n$ exceeds the critical value. Theorem~\ref{thm:CT} also shows that the asymptotic null distribution of the collider bias test $\lambda_{n}$ does not depend on the exact set of valid instruments $V^{*}$; it only depends on the number of valid instruments $v^*$.
Finally, for a fixed $\alpha$, as the number of valid instruments $v^{*}$ increases, the critical value becomes smaller. In other words, by allowing a greater number of invalid instruments into our test statistic, we push the critical value farther away from zero and make it harder to reject the null hypothesis of no effect.

In comparison to the union method in Section \ref{sec:genprocedure}, a disadvantage of the collider bias test is that it does not directly produce confidence intervals for $\beta^{*}$; it only produces statistical evidence in the form of a p-value for or against the null hypothesis of no effect. But, the collider bias test does not depend on a pre-specified $\bar{s}$ like the method in Section \ref{sec:genprocedure} and consequently, is computationally efficient. 
Also, both the method in Section \ref{sec:genprocedure} and the collider bias test $\lambda_{n}$ can handle a very small proportion of valid instruments and maintain the correct size; they only require one valid instrument while other methods in the literature require more valid instruments.

\subsection{Combining Method 1 and Method 2} \label{sec:comb_test}
Given the advantages and disadvantages of each method, we propose a combined procedure to test the hypothesis of no causal effect in the presence of invalid instruments. The combined testing procedure, which is described formally in Theorem \ref{prop:combine}, essentially splits the Type-I error between the two methods introduced in Sections \ref{sec:CI} and \ref{sec:conditional_test} and rejects the null hypothesis of no effect if either test rejects it.
\begin{theorem} \label{prop:combine} 
    For any $\alpha \in (0,1)$, pick $\alpha_1, \alpha_2 \in (0,1)$ so that $\alpha = \alpha_1 + \alpha_2$. Consider a combined testing procedure where we reject the null hypothesis of no causal effect if $C_{1-\alpha_1}(\mathbf{Y}_{n}, \mathbf{D}_{n},\mathbf{Z}_{n})$ contains $0$ or if the collider bias test rejects the null with $\alpha_2$. Then, the Type-I error of this combined test is less than or equal to $\alpha$. 
\end{theorem}
We remark that the combined test reduces to the test based on $C_{1-\alpha}(\mathbf{Y}_{n}, \mathbf{D}_{n},\mathbf{Z}_{n})$ if $\alpha_1 = \alpha$ and $\alpha_2 = 0$. Similarly, the combined test reduces to the collider bias test if $\alpha_1 = 0$ and $\alpha_2 = \alpha$. If $\alpha_{1} = \alpha/2$ and $\alpha_{2} = \alpha/2$, we are using both procedures to test the null of no effect, but each test is conducted at more stringent Type-I error rates than if they are used alone. While this may seem statistically costly, numerical studies below demonstrate that the cost is minimal in comparison to the gain in power across different values of the alternative.


\section{Simulation Study With Invalid Instruments}
\label{sec:simGeneral}

We conduct a simulation study to evaluate the performance of our two methods when invalid instruments are present. The simulation setup follows equation \eqref{eq:powerSetup} with $n=1000$ individuals, $L = 10$ candidate instruments, and each instrument is independent each other. For each simulation setting, we generate $1000$ independent replicates.  
We test the null causal effect $H_{0} : \beta^{*} = 0$ and vary $\beta^{*}$. We change $\bm{\pi}^*$'s support from $0$ to $1$ and vary the number of invalid instruments (i.e. number of $1$s in $\bm{\pi}^*$) by changing the number of non-zero $\bm{\pi}^{*}$'s. We set $\sigma_1 = \sigma_2 = 2$ and $\rho = 0.8$. We consider two different values for $\boldsymbol{\gamma}^*$ that correspond to concentration parameters $100$ and $5$. The concentration parameter is the expected value of the F statistic for the coefficients $\mathbf{Z}_{V^{*}}$ in the regression of $D$ and $\mathbf{Z}$ and is a measure of instrument strength \citep{stock_survey_2002}. Here, a concentration parameter of $100$ represents strong instruments and a concentration parameter of $5$ represents weak instruments. 

\subsection{Choice of Test Statistics for Method 1}
In the first simulations study, we compare different test statistics that can be used in \eqref{eq:confIntProposal} and \eqref{eq:confIntProposalScreen}. We also include ``naive'' and ``oracle'' methods as two baseline procedures where ``naive'' methods assume all candidate instruments are valid, typical in practice, and ``oracle'' methods assume one knows exactly which instruments are valid and invalid, i.e. $V^*$ is known, and use \eqref{eq:confIntInversion}. Note that the oracle methods are not practical because an investigator rarely has complete knowledge about which instruments are invalid versus valid. We use five different types of test statistics in the union method and examining the length and coverage of confidence intervals. 

Table~\ref{tab:strong20_U5} shows the case where we set $\bar{s} = 5$, the instruments are strong, and $s^{*}$ varies from $0$ to $4$; this is the case where at most 50\% of instruments are invalid. When there are no invalid instruments, $s^{*}=0$, the naive and oracle procedures have the desired 95\% coverage. Our methods have higher than 95\% coverage because they do not assume that all $10$ candidate instruments are valid. As the number of invalid instruments, $s^{*}$, increases, the naive methods fail to have any coverage. Our methods, in contrast, have the desired level of coverage, with coverage levels reaching nominal levels when $s^{*}$ is at the boundary of $s^{*} < \bar{s}$, i.e., $s^{*} = 4$; our method does this without knowing which instruments are valid or invalid a priori. The oracle methods always have approximately the desired coverage at every $s^{*}$ since they know which instruments are valid and invalid. 
\begin{table}[H]
	\centering
	\resizebox{0.6\textwidth}{!}{\begin{tabular}{rllllll}
			\hline
			Method & Test Statistic & $s^{*}$=0 & $s^{*}$=1 & $s^{*}$=2 & $s^{*}$=3 & $s^{*}$=4 \\ 
			\hline
			Naive &  TSLS & 94.3 & 0.0 & 0.0 & 0.0 & 0.0 \\ 
			& AR & 93.0 & 0.0 & 0.0 & 0.0 & 0.0 \\ 
			&
			CLR & 94.7 & 0.0 & 0.0 & 0.0 & 0.0 \\ 
			Union &  TSLS & 100.0 & 100.0 & 100.0 & 99.3 & 94.2 \\ 
			& AR & 100.0 & 100.0 & 100.0 & 99.5 & 95.0 \\ 
			& CLR & 100.0 & 100.0 & 99.9 & 99.1 & 94.5 \\ 
			& SarganTSLS & 100.0 & 100.0 & 100.0 & 99.2 & 93.9 \\ 
			& SarganCLR & 100.0 & 100.0 & 100.0 & 99.5 & 94.3 \\ 
			Oracle & TSLS & 94.3 & 94.4 & 93.7 & 94.0 & 94.2 \\ 
			& AR & 93.0 & 94.5 & 93.0 & 94.3 & 95.0 \\
			& CLR & 94.7 & 94.8 & 95.2 & 94.5 & 94.5 \\ 
			\hline
	\end{tabular}}
	\caption{\label{tab:strong20_U5} TSLS, two-stage least squares; AR, Anderson-Rubin test; CLR, conditional likelihood ratio test; Sargan--, Sargan test used as a pretest. Comparison of coverage between 95\% confidence intervals with strong instruments when we set the upper bound to $\bar{s}=5$.}
\end{table}
Table~\ref{tab:strong20_U5_length} examines the median length of the 95\% confidence intervals simulated in Table ~\ref{tab:strong20_U5}. We only compare between our methods and the oracles since they always have at least 95\% coverage. The table shows that our method and the oracles become similar in terms of length as the number of invalid instruments $s^{*}$ grows, with the Anderson-Rubin test and methods with pretesting achieving oracle performance at $s^{*}=4$ while two-stage least squares and the conditional likelihood ratio test, both without pretesting, not reaching oracle performance at $s^{*}=4$. The improved performance using pretesting is expected since pretesting removes unnecessary unions of intervals in~\eqref{eq:confIntProposalScreen}. 
\begin{table}[H]
	\centering
	\resizebox{0.7\textwidth}{!}{\begin{tabular}{rlllllll}
			\hline
			Method & Test Statistic & $s^{*}$=0 & $s^{*}$=1 & $s^{*}$=2 & $s^{*}$=3 & $s^{*}$=4 \\ 
			\hline
			Union & TSLS &  0.238 & 0.500 & 0.912 & 1.390 & 1.878 \\ 
			& AR & 0.337 & 0.318 & 0.290 & 0.254 & 0.202 \\   
			& CLR &  0.243 & 9.694 & 68.175 & 117.631 & 160.678 \\ 
			& SarganTSLS & 0.258 & 0.242 & 0.222 & 0.194 & 0.155 \\ 
			& SarganCLR & 0.264 & 0.247 & 0.227 & 0.198 & 0.157 \\ 
			Oracle & TSLS & 0.105 & 0.111 & 0.117 & 0.126 & 0.136 \\ 
			& AR &  0.168 & 0.176 & 0.181 & 0.190 & 0.202 \\ 
			& CLR & 0.106 & 0.113 & 0.119 & 0.128 & 0.138 \\  
			\hline
	\end{tabular}}
	\caption{\label{tab:strong20_U5_length} TSLS, two-stage least squares; AR, Anderson-Rubin test; CLR, conditional likelihood ratio test. Comparison of median lengths between 95\% confidence intervals with strong instruments and $\bar{s} = 5$.}
\end{table}

In Table~\ref{tab:strong20_U10}, we set $\bar{s} = 10$, the instruments are strong, and $s^{*}$ varies from $0$ to $9$; this is the case where the investigator is very conservative about the number of valid instruments and sets $\bar{s}$ at its maximum value $L$. Note that pretesting methods cannot be applied in this extreme case because Theorem \ref{theorem2} requires $c(B^C) \geq 2$; in this case, $c(B^C) = 1$. Table~\ref{tab:strong20_U10} shows that similar to Table~\ref{tab:strong20_U5}, our method without pretesting and the oracles become similar as the number of invalid instruments $s^{*}$ grows. 
\begin{table}[H]
	\centering
	\resizebox{0.9\textwidth}{!}{\begin{tabular}{rlllllllllll}
			\hline
			Method & Test Statistic & $s^{*}$=0 & $s^{*}$=1 & $s^{*}$=2 & $s^{*}$=3 & $s^{*}$=4 & $s^{*}$=5 & $s^{*}$=6 & $s^{*}$=7 & $s^{*}$=8 & $s^{*}$=9 \\ 
			\hline
			Naive & TSLS & 94.3 & 0.0 & 0.0 & 0.0 & 0.0 & 0.0 & 0.0 & 0.0 & 0.0 & 0.0 \\
			& AR & 93.0 & 0.0 & 0.0 & 0.0 & 0.0 & 0.0 & 0.0 & 0.0 & 0.0 & 0.0 \\ 
			&  CLR & 94.7 & 0.0 & 0.0 & 0.0 & 0.0 & 0.0 & 0.0 & 0.0 & 0.0 & 0.0 \\ 
			Union & TSLS & 100.0 & 100.0 & 100.0 & 100.0 & 100.0 & 100.0 & 100.0 & 99.9 & 100.0 & 95.6 \\ 
			& AR & 100.0 & 100.0 & 100.0 & 100.0 & 100.0 & 100.0 & 100.0 & 100.0 & 99.8 & 95.4 \\ 
			& CLR & 100.0 & 100.0 & 100.0 & 100.0 & 100.0 & 100.0 & 100.0 & 100.0 & 99.8 & 95.4 \\ 
			Oracle & TSLS & 94.3 & 94.4 & 93.7 & 94.0 & 94.2 & 93.9 & 94.1 & 93.6 & 94.5 & 95.6 \\ 
			& AR & 93.0 & 94.5 & 93.0 & 94.3 & 95.0 & 94.3 & 95.1 & 95.3 & 95.3 & 95.4 \\ 
			& CLR & 94.7 & 94.8 & 95.2 & 94.5 & 94.5 & 95.2 & 94.2 & 94.3 & 93.9 & 95.4 \\ 
			\hline
	\end{tabular}}
	\caption{\label{tab:strong20_U10} TSLS, two-stage least squares; AR, Anderson-Rubin test; CLR, conditional likelihood ratio test. Comparison of coverage between 95\% confidence intervals with strong instruments and $\bar{s}=10$.}
\end{table}
The simulation results suggest that when there are strong instruments, the Anderson-Rubin test and the pretesting method with two-stage least squares or conditional ratio test perform well with respect to power and length, with the Anderson-Rubin test being the simpler alternative since it doesn't use a pretest. Between the Anderson-Rubin test and the conditional likelihood ratio test, the Anderson-Rubin test dominates the conditional likelihood ratio test for $s^{*} > 0$. This finding differs from the advice in the weak instruments literature where the conditional likelihood ratio generally dominates the Anderson-Rubin test~\citep{andrews_optimal_2006, mikusheva_robust_2010}.

The supplementary materials present median lengths of our proposed confidence interval when we set $\bar{s}=10$ and the instruments are weak. In brief, in the worst case where the instruments are weak and there are many invalid instruments (i.e. all instrumental variables assumptions are violated), some of the test statistics used in our procedure lead to infinite-length, but valid confidence intervals. The presence of an infinite interval can be disappointing at first, but we believe it is informative in the sense that it alerts the investigator that the observed data is insufficient to draw any meaningful and robust conclusion about the treatment effect $\beta^{*}$.

\subsection{Power Comparison Between Methods}

\label{sec:power}
In this simulation study, we compare the statistical power between the union method, the collider bias test, and the combined method. Similar to the previous section, we consider both strong and weak instruments and vary the true number of invalid instruments $s^{*}$. For the union method, we set the upper bound on $s^*$ to be $\bar{s} = s^{*} + 1$. We use the Anderson-Rubin test for strong and weak instruments and the conditional likelihood ratio test for weak instruments.
\begin{figure}[!ht]
	\centering
	\begin{subfigure}[b]{0.4\textwidth}
		\includegraphics[width=\textwidth]{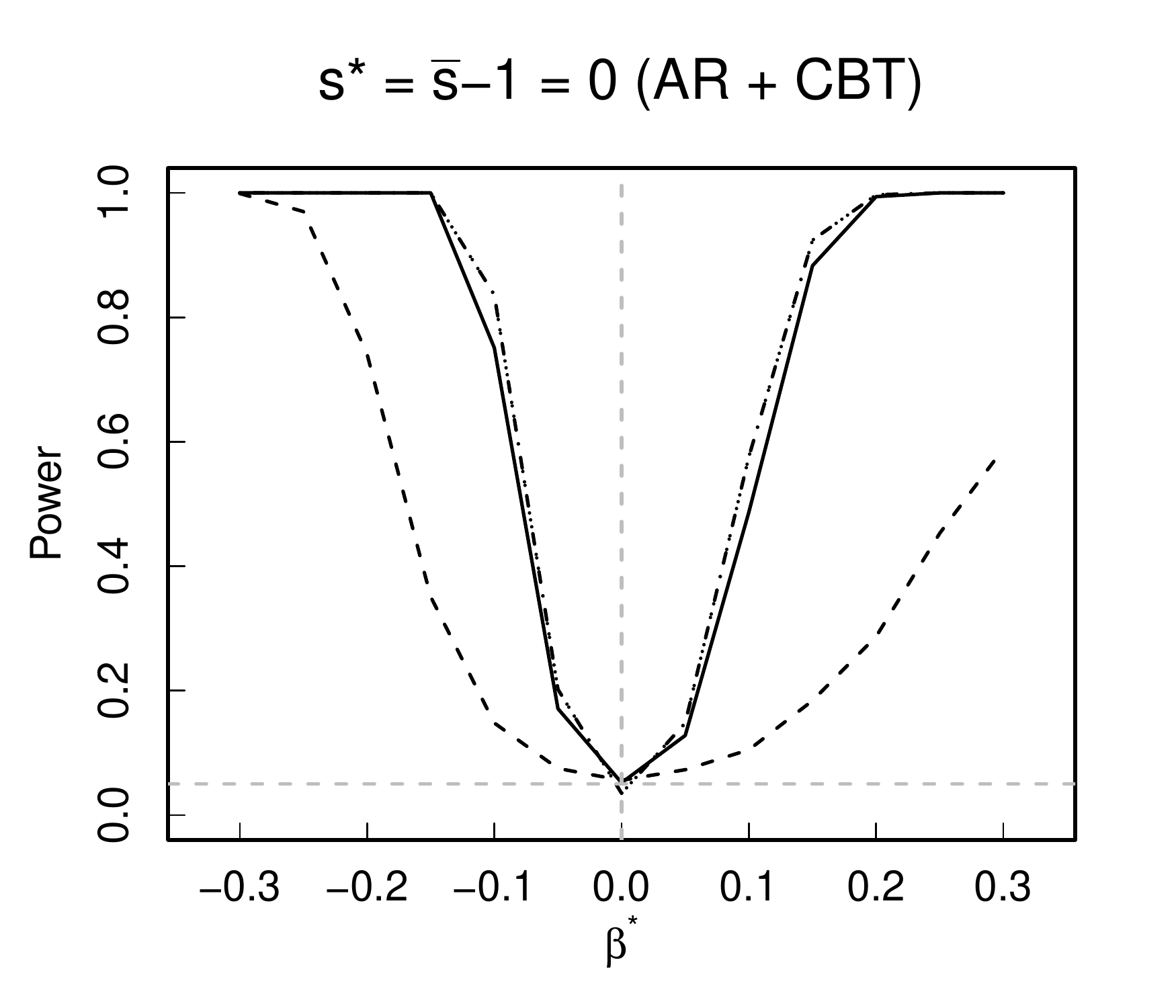}
	\end{subfigure}
	\begin{subfigure}[b]{0.4\textwidth}
		\includegraphics[width=\textwidth]{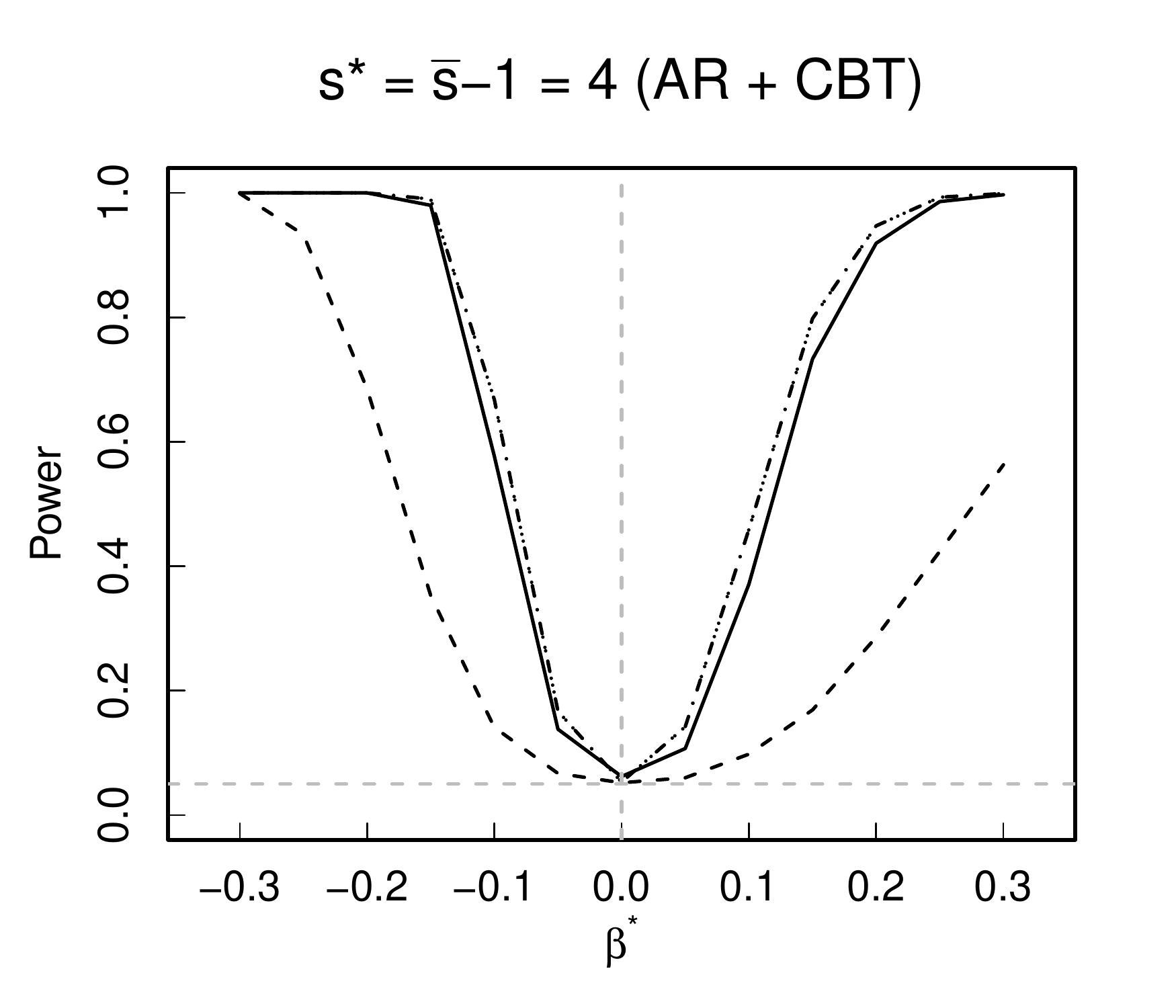}
	\end{subfigure}
	\begin{subfigure}[b]{0.4\textwidth}
		\includegraphics[width=\textwidth]{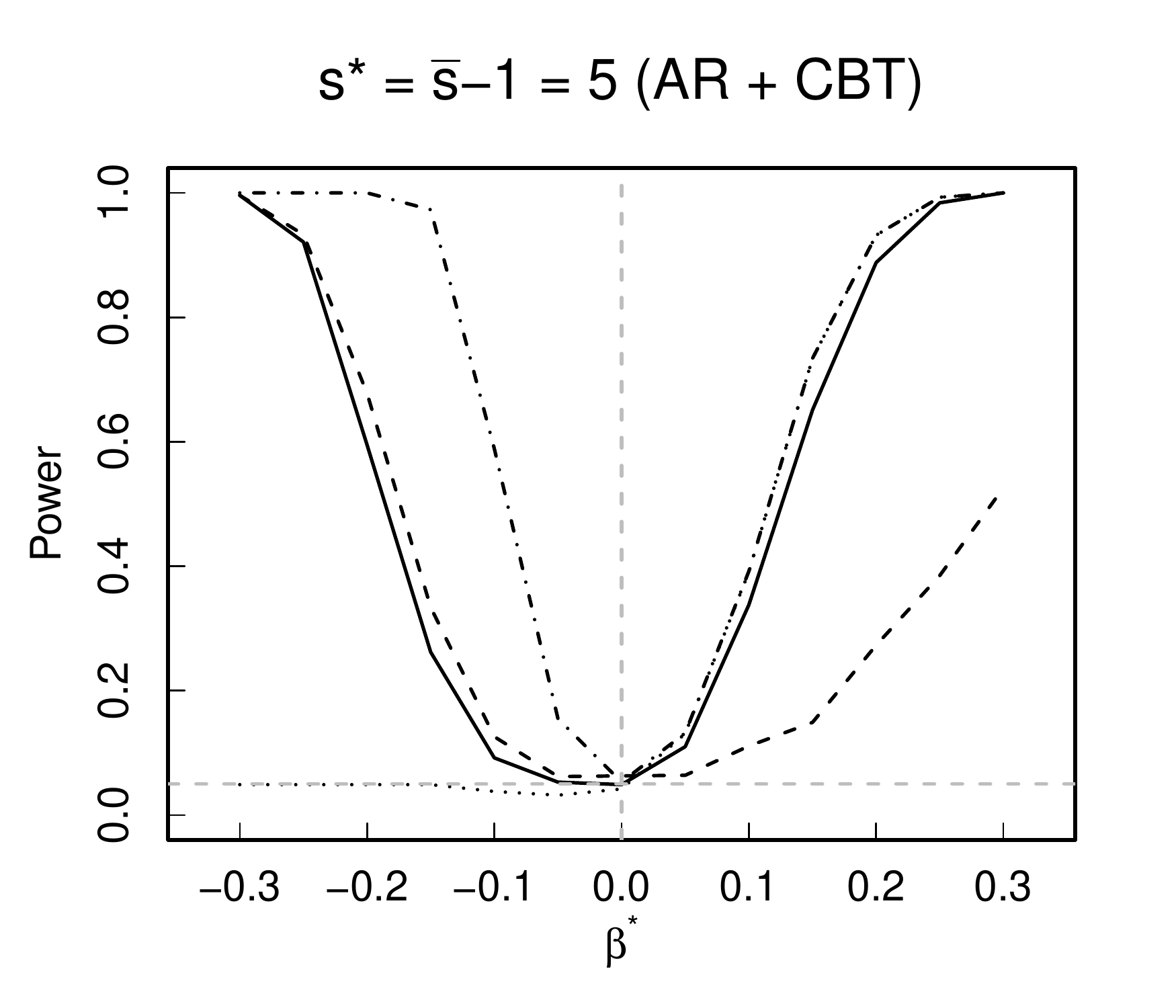}
	\end{subfigure}	
	\begin{subfigure}[b]{0.4\textwidth}
		\includegraphics[width=\textwidth]{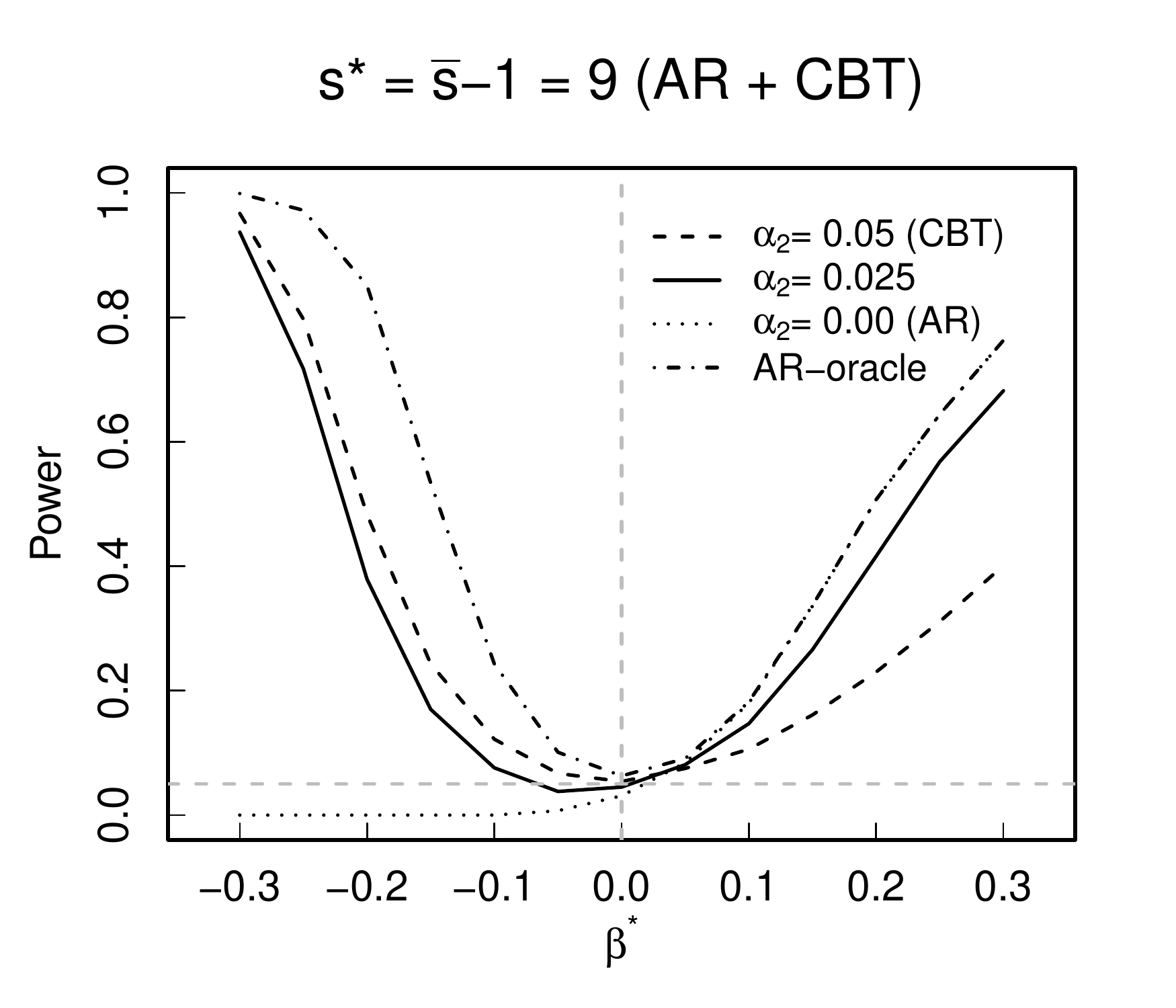}
	\end{subfigure}
	\caption{\label{fig:strong_either} AR: Anderson-Rubin test; CBT: collider bias test. Power of different methods under strong instruments with different numbers of invalid instruments. We fix $\alpha = 0.05$ and vary $\alpha_{2}$ to be $0.0, 0.025$, or $0.05$. When $\alpha_{2} = 0.05$, the combined test is equivalent to the collider bias test. When $\alpha_{2} = 0.0$, the combined test is equivalent to the union method. The oracle AR test is the AR test that knows exactly which instruments are valid.}
\end{figure}

Figure~\ref{fig:strong_either} presents the collider bias test ($\alpha_{2} = 0.05$), the union method using the Anderson-Rubin test $(\alpha_{2} = 0.00)$, the combined test $(\alpha_{1} = \alpha_{2} = 0.025)$, and the oracle Anderson-Rubin test that knows which instruments are valid when the instruments are strong.
When $s^{*}=0$ to $s^{*}=4$, the union method using the Anderson-Rubin test has similar power as the oracle method. But, when $s^{*}$ is greater than or equal to $5$, i.e.~when 50\% or more instruments are invalid, the Anderson-Rubin test only has power if the treatment effect $\beta^{*}$ is positive; it has no power when $\beta^{*}$ is negative. This asymmetric power of the Anderson-Rubin test may arise from the inflection points of the likelihood function \citep{kleibergen2007generalizing}.
The collider bias test has less power than the union method when $s^{*} \leq 4$, but has more power than the union method when $s^{*} \geq 5$ and $\beta^*$ is negative. The combined test achieves the best of both worlds where it has non-trivial power when $s^{*} \geq 5$ and $\beta^* < 0$ and has nearly similar performance as the union method used alone when  $s^{*} \leq 4$. Overall, the combined test shows the best performance among tests that do not assume any knowledge about which instruments are valid.
\begin{figure}[!ht]
	\centering
	\begin{subfigure}[b]{0.48\textwidth}
		\includegraphics[width=\textwidth]{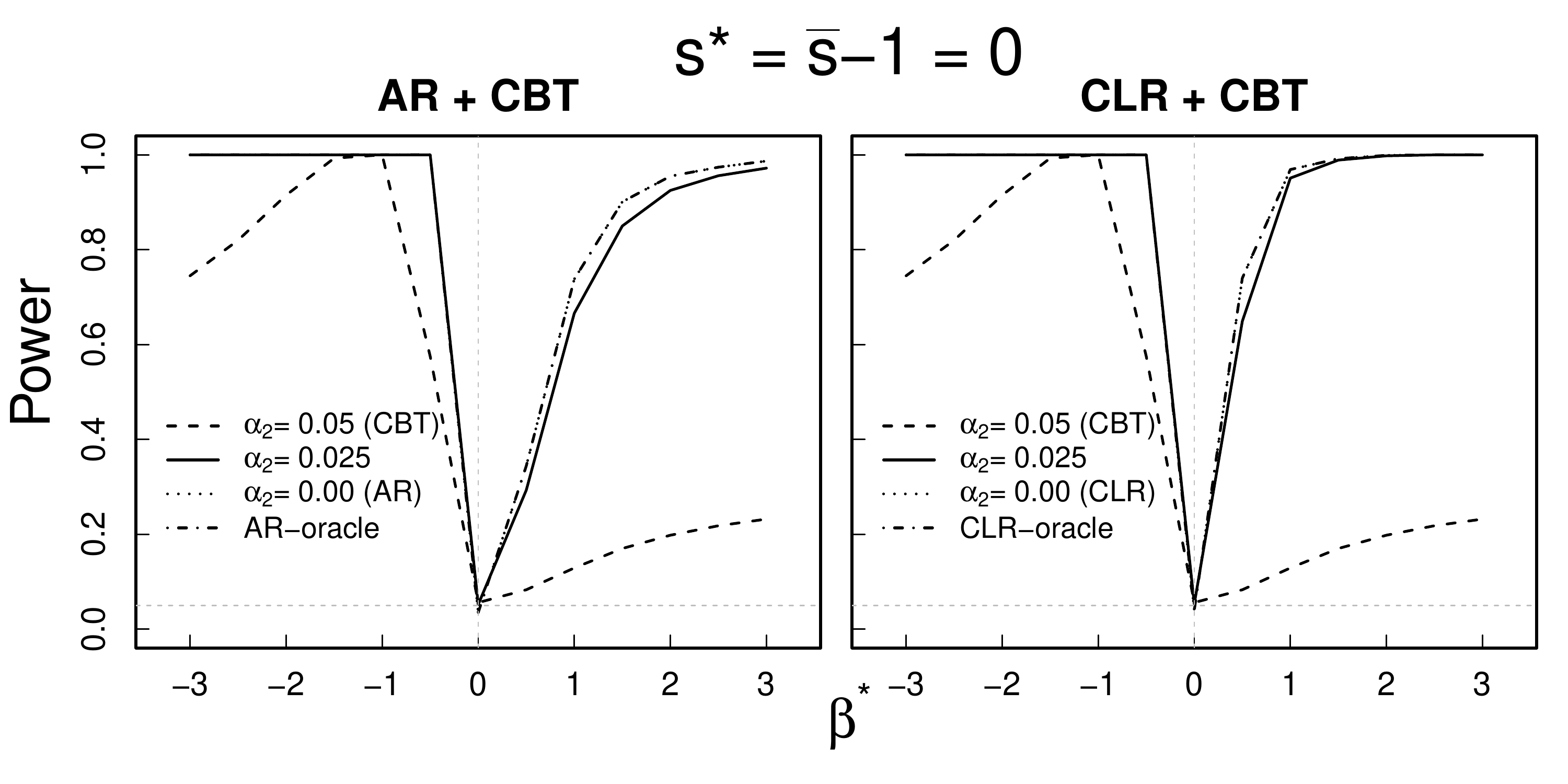}
	\end{subfigure}	{}
	\begin{subfigure}[b]{0.48\textwidth}
		\includegraphics[width=\textwidth]{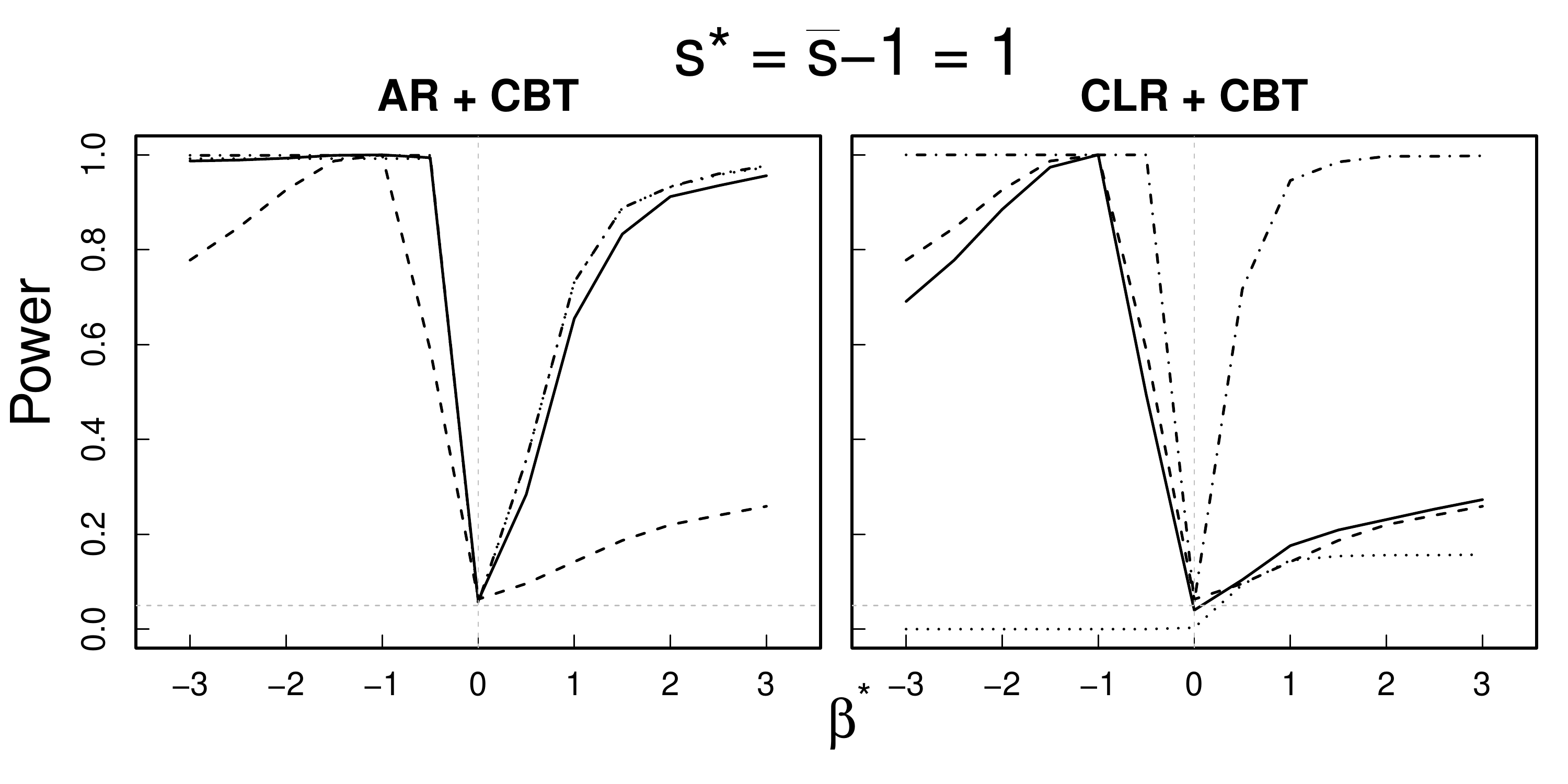}
	\end{subfigure} {}
	\begin{subfigure}[b]{0.48\textwidth}
		\includegraphics[width=\textwidth]{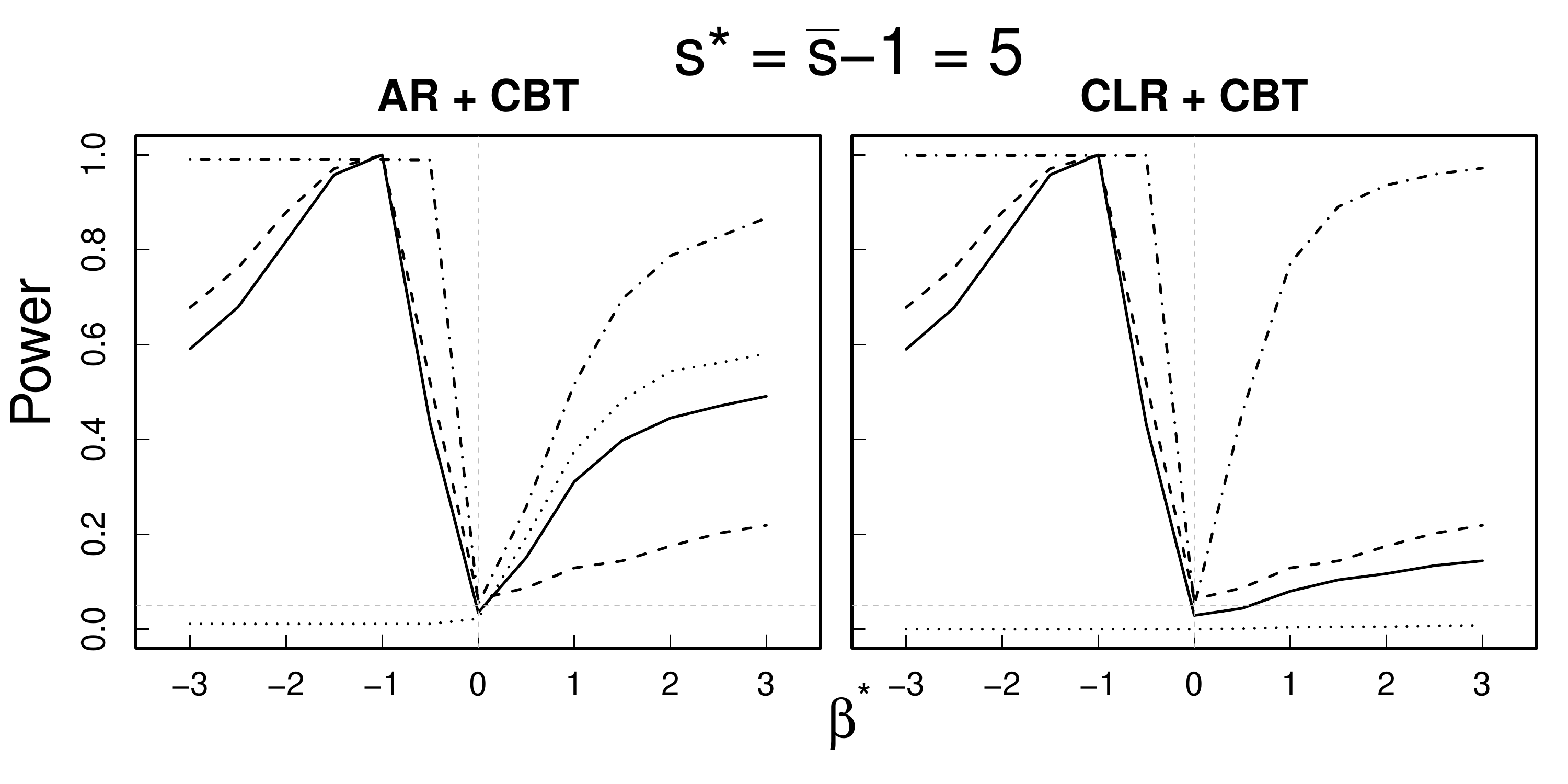}
	\end{subfigure}
	\begin{subfigure}[b]{0.48\textwidth}
		\includegraphics[width=\textwidth]{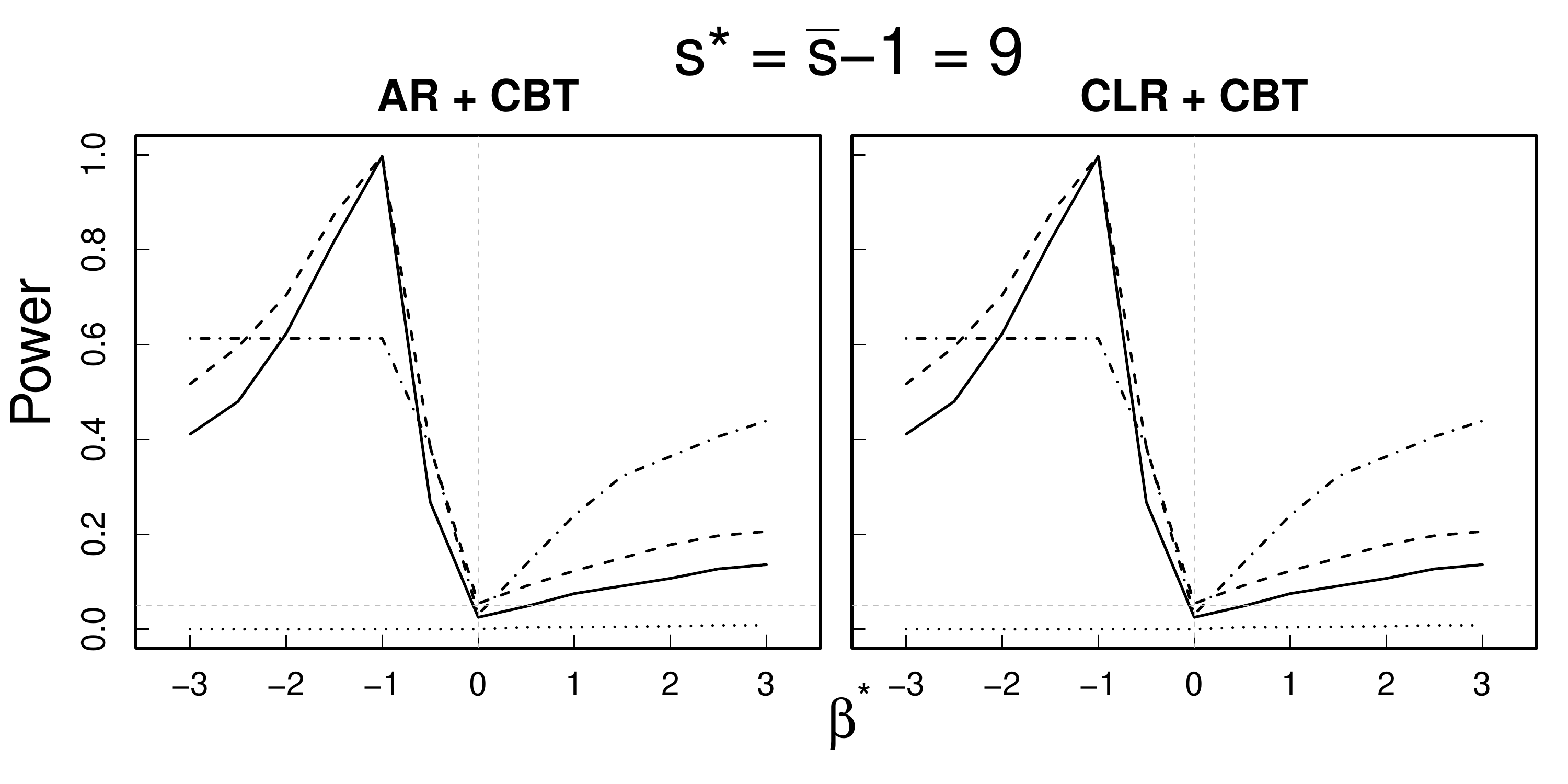}	
	\end{subfigure}	
	\caption{\label{fig:weak_either} AR: Anderson-Rubin test; CBT: collider bias test; CLR: conditional likelihood ratio test. Power of different methods under weak instruments with different numbers of invalid instruments. We fix $\alpha = 0.05$ and vary $\alpha_{2}$ to be $0.0, 0.025$, or $0.05$. When $\alpha_{2} = 0.05$, the combined test is equivalent to the collider bias test. When $\alpha_{2} = 0.0$, the combined test is equivalent to the union method. The oracle tests are tests that know exactly which instruments are valid. 
}
\end{figure}


Figure~\ref{fig:weak_either} presents the power of different methods under weak instruments. Similar to Figure~\ref{fig:strong_either}, we see that the collider bias test has better power than the union method when the treatment effect is negative, i.e. when $\beta^* < 0$, and $s^{*} > 4$; the Anderson-Rubin test and the conditional likelihood ratio test have zero power in this region.
Also, the combined test using the Anderson-Rubin test at $\alpha_{1} = \alpha_{2} = 0.025$ generally has higher power than the combined test using the conditional likelihood ratio test when $s^{*} > 0$. Finally, while not shown in the graph, we note that the power of the collider bias test decreases as $\beta^{*} < -1$. There are multitude of reasons for this, but the most likely explanation is the opposite signs of $\pi^{*}$ and $\beta^{*}$ can attenuate the collider bias and thus decrease power. 

Overall, the simulation studies suggest that there is no uniformly dominant test for the treatment effect across all scenarios concerning invalid instruments. The performance depends both on the number of invalid instruments and instrument strength. Nevertheless, nearly all proposed tests achieve near-oracle power when $s^{*}$ is close to $\bar{s}$ and the combined test has substantially better power overall than if each method is used alone. 

\section{Data Analysis: Mendelian Randomization in the Framingham Heart Study}
\label{sec:FHS}

We use our two methods to study the effect of low-density lipoprotein (LDL-C) on the incidence of cardiovascular disease (CVD) among individuals in the Framingham Heart Study (FHS) Offspring Cohort. Over several decades, the FHS has been one of the most popular epidemiologic cohort studies to identify risk factors for CVD, and recently, Mendelian randomization has been used to uncover the causal relationships in the presence of unmeasured confounding~\citep{smith2014association, mendelson2017association}. Traditional Mendelian randomization requires every instrument to be valid in order to test for a treatment effect, a tall order for many studies. Our two proposed methods relax this requirement and allow some of the instruments to be invalid.

For the main analysis, we selected ten SNPs that are known to be significantly associated with LDL-C measured in mg/dL ~\citep{kathiresan2007genome, ma2010genome, smith2014association} and are located in different chromosomes or in linkage equilibrium; the latter ensures that candidate instruments are statistically uncorrelated each other and reasonably satisfy the mutual independence requirement of instruments.
Our outcome $Y$ is binary indicating an incidence of CVD occurring at Exam 1 or later. Among $n=2,982$ subjects who had their LDL-C measured during Exam, $1,532$ ($17.5\%$) subjects had CVD afterwards. We also use subjects' age and sex as covariates $\mathbf{X}$.  The supplementary materials have additional details about the data.

Table~\ref{tab:tenmat} shows the confidence intervals based on method 1 as we vary $\bar{s}$ varies from $1$ to $3$. Since LDL-C is known to have a positive causal effect with CVD incidence ~\citep{kathiresan2008polymorphisms, voight2012plasma}, we estimated one-sided confidence intervals instead of two-sided confidence intervals. Only when there are no invalid instruments (i.e $\bar{s} = 1$) and we use the Anderson-Rubin test for the union method do we reject the null hypothesis of no effect; for all other tests and values of $\bar{s}$, we don't have the power to reject the null of no effect if we allow some instruments to be invalid. We also observed that as $\bar{s}$ increases, confidence intervals become wider and the null becomes harder to reject.

\begin{table}[ht]
	\centering
	\resizebox{0.7\textwidth}{!}{\begin{tabular}{rrrrrr}
			\hline
			& AR & CLR & TSLS & SarganTSLS & SarganCLR \\ 
			\hline
			$\mathbf{\alpha_{1} = 0.05}$ \\ 
			$\bar{s}$=1  & \cellcolor{Gray}+0.000 & -0.001 & -0.001 & -0.001 & -0.001 \\ 
			$\bar{s}$=2  &  -0.002 & -0.002 & -0.002 & -0.002 & -0.002 \\ 
			$\bar{s}$=3  & -0.002 & -0.002 & -0.002 & -0.003 & -0.003 \\ 
			\hline
			$\mathbf{\alpha_{1} = 0.025}$ \\ 
			$\bar{s}$=1  & -0.001 & -0.001 & -0.001 & -0.001 & -0.002 \\  
			$\bar{s}$=2 & -0.002 & -0.002 & -0.002 & -0.002 & -0.002 \\
			$\bar{s}$=3  & -0.003 & -0.003 & -0.003 & -0.003 & -0.003 \\  
			\hline
	\end{tabular}}
	\caption{\label{tab:tenmat} TSLS, two-stage least squares; AR, Anderson-Rubin test; CLR, conditional likelihood ratio test; Sargan--, Sargan test used as a pretest. Lower bound of an one-sided confidence interval for $\beta^{*}$ using different test statistics. We use $\alpha_1$
to be $0.05$ and $0.025$ and vary $\bar{s}$. There are $L=10$ candidate SNPs. Grey cells represent confidence intervals that rejected the null hypothesis of no effect.}
\end{table}

Table \ref{tab:tenstat} shows the result using the collider bias test by calculating the critical values $\{ \underline{\chi}^{2}_{L, v, 1-\alpha_2}: v = L - \bar{s} + 1 \}$ as a function of $\bar{s}$.
We vary the size of the test $\alpha_{2}$ between $0.05$ and $0.025$. The observed value of the collider bias test statistic is $\lambda_n = 11.019$. For $\alpha_2 = 0.05$, we can reject the null of no effect when less than $\bar{s}=7$ out of $L=10$ instruments are invalid. For $\alpha_2 = 0.025$, we can reject the null when less than $\bar{s} = 6$ out of $L=10$ instruments are valid. Compared to the union method in Table \ref{tab:tenmat}, the collider bias test suggests more evidence against the null hypothesis for this data set.

\begin{table}[H]
	\centering
	\resizebox{0.9\textwidth}{!}{\begin{tabular}{rrrrrrrrrrr}
			\hline
			$\bar{s}$ & 10 & 9 & 8 & 7 & 6 & 5 & 4 & 3 & 2 & 1 \\ 
			\hline
			$\alpha_{2} = 0.05$ & 18.227 & 13.463 & 11.316 & \cellcolor{Gray}10.087 & \cellcolor{Gray}9.275 & \cellcolor{Gray}8.679 & \cellcolor{Gray}8.148 & \cellcolor{Gray}7.891 & \cellcolor{Gray}7.584 & \cellcolor{Gray}7.366 \\ 
			$\alpha_{2} = 0.025$ &    20.172 &  14.800 &  12.253 & 11.057 & \cellcolor{Gray}10.137 & \cellcolor{Gray} 9.486 & \cellcolor{Gray} 8.973 & \cellcolor{Gray} 8.536 & \cellcolor{Gray} 8.246 & \cellcolor{Gray} 7.972 \\ 
			\hline
	\end{tabular}}
	\caption{\label{tab:tenstat} The critical value of the collider bias test $\lambda_{n}$ under the null hypothesis of no effect. Each row represents the size $\alpha_2$ and each column represents the number of valid instruments allowed. Grey cells represent settings where the observed value of the collider bias test exceeds the critical value and therefore, we can reject the null hypothesis of no effect.}
\end{table}

Table~\ref{tab:combined_result} shows the combined testing procedure in Section \ref{sec:comb_test} to test the null hypothesis of no effect. Here, $\alpha_1$ corresponds to the size of $C_{1-\alpha_1}(\mathbf{Y}_n, \mathbf{D}_{n}, \mathbf{Z}_n)$ and $\alpha_2$ corresponds to the size for the collider bias test; both sum to $0.05$. We see that when we evenly split the size into $\alpha_{1} = \alpha_{2} = 0.025$, we reject the null with less than $\bar{s} = 6$ invalid instruments. This holds across different test statistics used in $C_{1-\alpha_1}(\mathbf{Y}_n, \mathbf{D}_{n}, \mathbf{Z}_n)$.

\begin{table}[ht]
	\centering
	\resizebox{\textwidth}{!}{\begin{tabular}{r|rrrrrrrrrrr}
			\hline
			$\alpha_{1}$ & 0 & 0.005 & 0.01 & 0.015 & 0.02 & \cellcolor{Gray} 0.025 & 0.03 & 0.035 & 0.04 & 0.045 & 0.05 \\ 
			$\alpha_{2}$ & 0.05 & 0.045 & 0.04 & 0.035 & 0.03 & \cellcolor{Gray} 0.025 &  0.02 & 0.015 & 0.01 & 0.005 & 0 \\
			\hline
			AR & 7 & 7 & 7 & 7 & 7 & \cellcolor{Gray} 6 & 6 & 6 & 5 & 4 & 1 \\
			CLR $\cdot$ TSLS $\cdot$ SarganTSLS $\cdot$ SarganCLR  & 7 & 7 & 7 & 7 & 7 & \cellcolor{Gray} 6 & 6 & 6 & 5 & 4 & NR \\
			\hline
	\end{tabular}}
	\caption{\label{tab:combined_result} NR: No rejection of the null hypothesis. Smallest upper bound $\bar{s}$ needed to reject the null hypothesis of no effect using the combined testing procedure. $\alpha_1$ represents the size of $C_{1-\alpha_1}(\mathbf{Y}_n, \mathbf{D}_{n}, \mathbf{Z}_n)$ and $\alpha_2$ represents the size of the collider bias test. Grey cells represent regions where the the null hypothesis of no effect was rejected.}
\end{table}

Overall, the data analysis reaffirms the presence of a causal effect between LDL-C and risk of CVD and this conclusion is robust so long as there are at least 5 valid instruments among the 10 candidate instruments we used in our analysis. The supplementary materials conduct additional analysis where we change the candidate instruments. We also rerun the analysis when we use at most one subject from each family in the Offspring Cohort. The latter analysis reduces confounding due to population structure and cryptic relatedness by not including genetically and/or socially related subjects in the study~\citep{astle2009population, sul2018population}.  Overall, we find that if we use a different set of candidate instruments, (i) $C_{1-\alpha_1}(\mathbf{Y}_n, \mathbf{D}_{n}, \mathbf{Z}_n)$ has more evidence against the null and (ii) the combined test has better performance than either methods if used alone.



\section{Discussion}

This paper proposes two methods to conduct valid inference for the treatment effect
when instruments are possibly invalid. The first method is a simple modification of pre-existing methods in instrumental variables to construct robust confidence intervals for causal effect. 
The second method is a novel test that leverages the presence of a collider bias and test the null hypothesis of no causal effect; the second method produces valid inference so long as there is at least one valid instrument.
We also propose a combined test that generally has better power than either method used alone.
We show through numerical experiments and data analysis how our proposed methods can be used to arrive at more robust conclusions about the presence of a causal effect when invalid instruments are present.
\section*{Acknowledgements}
The Framingham Heart Study is conducted and supported by the National Heart,
Lung, and Blood Institute (NHLBI) in collaboration with Boston University (Contract No.N01-HC-25195 and HHSN268201500001I). This manuscript was not prepared in
collaboration with investigators of the Framingham Heart Study and does not
necessarily reflect the opinions or views of the Framingham Heart Study, Boston
University, or NHLBI.
Funding for CARe genotyping was provided by NHLBI Contract N01-HC-65226. The research of Hyunseung Kang was supported in part by NSF Grant DMS-1811414.

\bibliography{example}

\appendix
\section{Appendix}
\begin{proof}[Proof of Theorem~\ref{theorem1}]
By $s^{*} = c(B^*) < \bar{s}$, there is a subset $\tilde{B}$ where $c(\tilde{B}) = \bar{s}-1$ and $B^* \subseteq \tilde{B}$. Also, its complement $\tilde{B}^C$ only contains valid instruments and thus, $\text{pr}\{\beta^* \in C_{1 - \alpha}(\mathbf{Y}_{n}, \mathbf{D}_{n},\mathbf{Z}_{n},\tilde{B}) \} \geq 1 - \alpha$. Hence, we have 
\[
\text{pr}\{\beta^* \in C_{1 - \alpha}(\mathbf{Y}_{n}, \mathbf{D}_{n}, \mathbf{Z}_{n})\} \geq \text{pr}\{\beta^* \in C_{1- \alpha}(\mathbf{Y}_{n},\mathbf{D}_{n},\mathbf{Z}_{n},\tilde{B})\} \geq 1- \alpha
\]
for all values of $\beta^*$.
\end{proof}

\begin{proof}[Proof of Theorem~\ref{theorem2}] Similar to the proof for Theorem \ref{theorem1}, $\tilde{B}$, which is a superset containing all invalid instruments, has to exist. Also, $\tilde{B}$ has the property $\text{pr}\{S(\tilde{B}) \geq q_{1 - \alpha_s}\} \leq \alpha_s$. Then, we can use Bonferroni's inequality to obtain
\begin{align*}
\text{pr}\{\beta^* \in C_{1 - \alpha}'(\mathbf{Y}_{n}, \mathbf{D}_{n},\mathbf{Z}_{n})\} &\geq \text{pr}\{\beta^* \in C_{1- \alpha_{t}}(\mathbf{Y}_{n}, \mathbf{D}_{n},\mathbf{Z}_{n},\tilde{B}) \cap S(\tilde{B}) \leq q_{1 - \alpha_{s}} \} \\
&\geq 1 - \text{pr}\{\beta^* \notin C_{1- \alpha_{t}}(\mathbf{Y}_{n}, \mathbf{D}_{n},\mathbf{Z}_{n},\tilde{B})\} - \text{pr}\{S(\tilde{B}) \geq q_{1 - \alpha_{s}} \} \\
&\geq 1 - \alpha_{s} - \alpha_{t} = 1 - \alpha
\end{align*}
thereby guaranteeing at least $1-\alpha$ coverage.
\end{proof}

\begin{proof}[Proof of Theorem~\ref{thm:hyp_imp}]
First suppose that $H_{0} : \beta^{*} = 0$ holds under (A4), and we have at least one valid instrument among $L$. Then there exists at least one $j \in \{1,2,\ldots, L\}$ that satisfies $\{ Z_{j} \bigCI Z_{k} \mbox{ and } Z_{j} \bigCI Y; k = 1,2,\ldots, L,~k \neq j \}$, so $\prod\limits_{j=1,2,\ldots, L} \sigma_{j, L+1} = 0$. 

Next suppose that $H_{0} : \beta^{*} = 0$ does not hold. Then since $Y$ is a collider between the $L$ instruments, for all $j = 1,2,\ldots,L$, $\{Z_{j} \bigCI Z_{k} | Y;~k=1,2,\ldots,L,~k \neq j\}$ does not hold. Therefore, $H_{0j}$ does not hold and $\{Z_{k} \nbigCI Y:~k\neq j\}$, either of which results to $\prod\limits_{j=1,2,\ldots, L} \sigma_{j, L+1} \neq 0$.
\end{proof}

\begin{proof}[Proof of Theorem~\ref{thm:CT}]

Consider $L$ candidate instruments. Let $A$ denote a set of indices for the first $(L-1)$ instruments $\{ 1,2, \ldots, L \} \setminus \{L\}$ and the outcome, and $A^{-j}$ denote a set of indices for $(L-2)$ instruments $\{ 1,2, \ldots, L \} \setminus \{j, L\}$ and the outcome $(j=1,2,\ldots, L)$.
For any $K \times K~(K \in \mathbb{N})$ matrix $\mathbf{M} = (m_{kl})$,  $\mathbf{M}_{-j, -j}$ denotes a submatrix of $\mathbf{M}$ with $j^{\text{th}}$ column and row removed; 
and $m_{(jj.A)} := m_{jj} - \mathbf{M}_{ \{j\} \times A} \mathbf{M}^{-1}_{A \times A} \mathbf{M}_{A \times \{j\}}$, following the same notations introduced in \cite{drton2006algebraic}. 

We prove Theorem~\ref{thm:CT} first by showing that the likelihood ratio test statistic for $H_{0j} : \sigma_{jk} = 0;~k=1,2,\ldots, L+1,~k\neq j$:
\begin{eqnarray*}
\lambda_{n,j} := n\log\left( \frac{s_{jj} \mbox{det}(\mathbf{S}^{(L)}_{-j,-j})}{\mbox{det}(\mathbf{S}^{(L)})} \right) \overset{d}{\longrightarrow} \sum\limits_{k=1}^{L} W_{jk}
\end{eqnarray*}
by induction.
First consider $L = 2$ instruments and $\mathbf{S}^{(2)} = (s^{(2)}_{kl})$ is a sample covariance matrix of $(Z_{1}, Z_{2}, Y)$. Based on Proposition 4.2 in~\cite{drton2006algebraic}, we can claim that for a valid instrument $Z_{j}$, $j=1,2$ that satisfies $H_{0j}$:
\begin{eqnarray*}
n\log\left( \frac{s^{(2)}_{jj} \text{det}(\mathbf{S}^{(2)}_{-j,-j})}{\text{det} (\mathbf{S}^{(2)}) } \right) \overset{d}{\longrightarrow} W_{jj} + W_{jk}~\mbox{ for } k \neq j~,k=1,2.
\end{eqnarray*}
Next consider $L-1$ candidate instruments. Let $\mathbf{S}^{(L-1)} = (s^{(L-1)}_{kl})$ denote a $L \times L$ covariance matrix of $(Z_{1}, Z_{2}, \ldots, Z_{L-1}, Y)$ and $\mathbf{S}^{(L-1)}_{-j, -j}$ denote a submatrix of $\mathbf{S}^{(L-1)}$ with $j^{\text{th}}$ column and row removed from $\mathbf{S}^{(L)}$ ($j=1,2,\ldots,L-1$). \textit{Suppose} that the following equation holds under the null of $H_{0,j}$ for $j=1,2,\ldots, L-1$.
 \begin{eqnarray}
 \label{eq:induction}
n\log \left( \frac{s^{(L-1)}_{jj} \text{det} ( \mathbf{S}^{(L-1)}_{-j, -j} )}{\text{det}(\mathbf{S}^{(L-1)}) }    \right)  \overset{d}{\longrightarrow}   \sum\limits_{k=1}^{L-1} W_{jk}.
\end{eqnarray}
Lastly consider we have $L$ valid instruments with a sample covariance matrix $\mathbf{S}^{(L)} = (s_{kl})  \in \mathbb{R}^{(L+1) \times (L+1)}$ and a covariance matrix $\mathbf{\Sigma}^{(L)} = (\sigma_{kl}) \in \mathbb{R}^{(L+1) \times (L+1)}$. Then we can show that for any $j=1,2,\ldots, L$, the null of $H_{0, j} : \sigma_{jk} = 0$ for $k=1,2,\ldots, L+1,~k \neq j$ leads to the following decomposition under (A4), each term of which converges to the $\chi^{2}$ distribution.
\begin{eqnarray}
\label{eq:laststep}
\lambda_{n,j}=n\log \left( \frac{s_{jj} \text{det} ( \mathbf{S}^{(L)}_{-j, -j} )}{\text{det}(\mathbf{S}^{(L)}) }    \right) & =  & n\log \left( \frac{s^{(L-1)}_{jj} \text{det} ( \mathbf{S}^{(L-1)}_{-j, -j} )}{\text{det}(\mathbf{S}^{(L-1)}) }    \right) + n \log \left( \frac{ \text{det}(s^{(L)}_{LL.A^{-j}}) }{ \text{det}( s^{(L)}_{LL.A} ) }  \right) \nonumber \\  
& \overset{d}{\longrightarrow} & \sum\limits_{k=1}^{L-1} W_{jk} + W_{j L} = \sum\limits_{k=1}^{L} W_{jk}. 
\end{eqnarray}
The first term in~\eqref{eq:laststep} holds by our assumption with $(L-1)$ instruments \eqref{eq:induction}; and the second term can be proven by the next Lemma~\ref{lemma:chisq}. From \eqref{eq:laststep}, when the null for an instrument $Z_{j}$, i.e. $H_{0j}$, holds, $\lambda_{n,j}$ follows the distribution of $\sum\limits_{k=1}^{L} W_{jk} \sim \chi^{2}_{L}$. 

\begin{lemma}
\label{lemma:chisq} 
Let $\mathbf{\Sigma}^{(L)} = (\sigma_{kl})$ and $\mathbf{S}^{(L)} = (s_{kl})$ denote the $(L+1) \times (L+1)$ covariance matrix and sample covariance matrix of $(\mathbf{Z}_{n}, \mathbf{Y}_{n})$, respectively. 
Then under the null of $H_{0, j}: \sigma_{j, L+1} = 0$ and with mutually independent instruments:
\begin{eqnarray}
\label{eq:CT_onlyone}
n \log \left( \frac{ \text{det}(\mathbf{S}^{(L)}_{LL.A^{-j}}) }{ \text{det}( \mathbf{S}^{(L)}_{LL.A} ) }  \right)  \overset{d}{\longrightarrow} W_{Lj} \overset{d}{=} \chi^{2}_{1},
\end{eqnarray}
where $W_{Lj} \overset{d}{=} W_{jL}$ and $W_{Lj}$ is mutually independent with $\{ W_{kl}:~k,l=1,2,\ldots, L : W_{kl} \neq W_{jL}, W_{Lj}\}$.
\end{lemma}

Now under $H_{0} = \prod\limits_{j=1,2,\ldots, L} \sigma_{j, L+1} = 0$ and (A4), suppose that we know $v^{*} (\geq 1)$ out of $L$ nulls $\{H_{0,j} : j =1,2,\ldots, L \}$'s hold. Using the idea of Proposition 4.3 in~\cite{drton2008lectures}, we use the minimum of $\{ \lambda_{n,j} \}_{j=1}^{L}$'s as a test statistic that converges to the minimum of $\{ \sum\limits_{j=1,2,\ldots, L} W_{jk} \}_{j=1}^{L}$, and apply \cite{drton2008lectures}'s proposal that $\sum_{k=1}^{L} W_{jk}$ for an invalid instrument $Z_{j}$ would not contribute to the minimum; while each of other $v$, valid instruments'  $\sum\limits_{k=1}^{L} W_{jk}$ follows the asymptotic distribution of $\chi^{2}_{L}$. Therefore, the minimum of $L$ statistics in the parenthesis~\eqref{eq:CT} turns into the minimum of $v^{*}$, $\chi^{2}_{L}$-distributed variables and these variables are dependent each other through $W_{jk}$'s in $j,k \in V^{*}$. The following equivalences show that the distribution of $\min\limits_{j =  1,2,\ldots,n}( \sum\limits_{k=1}^{L} W_{jk})$ is invariant to the actual set of $V^{*}$ but only depends on the number of valid instruments $v^{*}$.

\begin{eqnarray*}
\lambda_{n} &=& \min\limits_{j=1,2,\ldots, L} \left( n\log \left( \frac{s_{jj} \text{det} ( \mathbf{S}^{(L)}_{-j, -j} )}{\text{det}(\mathbf{S}^{(L)}) }    \right) \right)  \overset{d}{\longrightarrow} \min \limits_{j=1,2,\ldots, L} \left( \sum\limits_{k=1}^{L} W_{jk} \right) \\
&=& \min\limits_{j \in V^{*}} \left( \sum\limits_{k=1}^{L} W_{jk}  \right)
=  \min\limits_{j \in V^{*}}\left( W_{jj} +   \sum\limits_{k \in V^{*} \setminus \{j\}} W_{jk} + \sum\limits_{k \notin V^{*}} W_{jk} \right) 
 := \underline{\chi}^{2}_{L,v^{*}}.
\end{eqnarray*}

\end{proof}

\begin{proof}[Proof of Lemma~\ref{lemma:chisq}]

First consider the asymptotic distribution of the log of two determinants from a sample covariance matrix $\mathbf{S}^{(L)}$ without scaled by $n$. We denote it as $\phi(\mathbf{S}^{(L)})$. Then under $H_{0j} : \sigma_{j, L+1} = 0$ and (A4), two determinants can be simplified using the Schur complement. 
\begin{eqnarray*}
\phi\left( \mathbf{S}^{(L)} \right)  := \log \left( \frac{ \text{det}(\mathbf{S}^{(L)}_{LL.A^{-j}}) }{\text{det}(\mathbf{S}^{(L)}_{LL.A}) }\right) &=& \log \left( \frac{s_{LL} - \sum\limits_{k,l \neq j,L}^{L+1} s_{Lk} s_{Ll} \tilde{s}^{-1}_{kl}}{ s_{LL} - \sum\limits_{k,l \neq L}^{L+1} s_{Lk} s_{Ll} \dot{s}^{-1}_{kl} } \right),
\end{eqnarray*}
where $\tilde{\mathbf{S}} = (\tilde{s}_{kl})$ is a $(L+1) \times (L+1)$ matrix with the same entries as those in the covariance matrix of $(\mathbf{Z}\setminus \{Z_{j}, Z_{L}\}, Y )$ at the row and column index in $A^{-j}$; and with a zero vector of length $(L+1)$ at $j^{\text{th}}$ and $L^{\text{th}}$ rows and columns. Similarly define a $(L+1) \times (L+1)$ matrix $\dot{\mathbf{S}} = (\dot{s}_{kl})$ with the same entries as the covariance matrix of $\{Z_{1}, Z_{2}, \cdots, Z_{L}, Y  \} \setminus \{Z_{L}\}$ at the rows and columns with index in $A$ and with a zero vector at the $L^{\text{th}}$ row and column. 
Let $\widetilde{\mathbf{S}}^{-1} = (\tilde{s}^{-1}_{kl})$ and $\dot{\mathbf{S}}^{-1} = (\dot{s}^{-1}_{kl})$ denote an inverse of each matrix. Define the covariance matrices of $\widetilde{\mathbf{\Sigma}} = (\widetilde{\sigma}_{kl})$ and $\dot{\mathbf{\Sigma}} = (\dot{\sigma}_{kl})$ from $\mathbf{\Sigma}$ matrix in a similar manner.

Then using the delta-method, we approximate the distribution of $\phi\left(  \mathbf{S}^{(L)} \right)$:
\begin{eqnarray*}
\phi(\mathbf{S}^{(L)}) &=& \phi(\mathbf{\Sigma}^{(L)}) + \sum\limits_{k,l=1;~k \leq l}^{L+1} \frac{\partial \phi(\mathbf{\Sigma}^{(L)}) }{ \partial \sigma_{kl}} (s_{kl} - \sigma_{kl} )  \\ && \quad + 
\frac{1}{2} \sum\limits_{k,l,g,h=1;~k \leq l, g \leq h}^{L+1} \frac{\partial^{2} \phi(\mathbf{\Sigma}^{(L)})}{\partial \sigma_{kl} \sigma_{gh}} (s_{kl} - \sigma_{kl})(s_{gh} - \sigma_{gh}) + \Delta_{n}.
\end{eqnarray*}
Under the null $H_{0j}: \sigma_{j, L+1} = 0$ and (A4), we have $\phi(\mathbf{\Sigma}^{(L)}) = \phi^{\prime}(\mathbf{\Sigma}^{(L)}) = 0$;
$\partial^{2} \phi / \partial \sigma_{kl} \sigma_{gh} = 0$ for $(k,l),(g,h) \neq (j,L)$; $\tilde{\sigma}_{kl} = \dot{\sigma}_{kl}$;
$\tilde{\sigma}^{-1}_{jL} = 0$ and $\bar{\sigma}^{-1}_{jL} = 0$; therefore, $\phi(\mathbf{S}^{(L)})$ ends up only depending on the one of second-order derivatives. 
\begin{eqnarray*}
\phi(\mathbf{S}^{(L)}) &=&  \frac{1}{2} \left( \frac{\partial^{2} \phi(\mathbf{\Sigma}^{(L)})}{\partial \sigma_{jL} \sigma_{jL}} + \frac{\partial^{2} \phi(\mathbf{\Sigma}^{(L)})}{\partial \sigma_{jL} \sigma_{jL}} \right) + \Delta_{n}, \\
\frac{\partial^{2} \phi(\mathbf{S}^{(L)})}{\partial \sigma_{jL} \sigma_{jL}} & = &  \frac{\partial^{2}}{\partial \sigma_{jL} \sigma_{jL}} \log \left( \frac{\sigma_{LL} - \sum\limits_{k,l \neq j,L}^{L+1} \sigma_{Lk} \sigma_{Ll} \tilde{\sigma}^{-1}_{kl}}{ \sigma_{LL} - \sum\limits_{k,l \neq L}^{L+1} \sigma_{Lk} \sigma_{Ll} \bar{\sigma}^{-1}_{kl} } \right) \\
&= & - \frac{\partial^{2}}{\partial \sigma_{jL} \sigma_{jL}} \log \left( \sigma_{LL} - \sum\limits_{k,l \neq j, L}^{L+1} \sigma_{Lk} \sigma_{Ll} \bar{\sigma}^{-1}_{kl} - \sigma^2_{Lj} \bar{\sigma}^{-1}_{jj} - 2 \sum\limits_{k=1, k \neq j, L}^{L+1} \sigma_{Lj}\sigma_{Lk} \bar{\sigma}^{-1}_{jk} \right) \bigg\rvert_{\sigma_{Lk} = 0, k \neq L}  \\ 
&= & (\sigma_{LL} \sigma_{jj})^{-1}.
\end{eqnarray*}
Since the terms beyond the second order, $\Delta_{n}$, converge to zero as the sample size increases, we have $\phi(\mathbf{S}^{(L)}) \approx (\sigma_{LL} \sigma_{jj})^{-1} (s_{jL} - \sigma_{jL})^{2}$. We finally use an Isserlis matrix~\citep{roverato1998isserlis, drton2006algebraic} to derive the distribution of each component of the sample covariance matrix $\mathbf{S}^{(L)} = (s_{kl})$ of which upper triangular components are known to follow: 
\begin{eqnarray}
\label{eq:Isserlis}
n^{-1/2}\left\{ (s_{11}, s_{12}, \ldots, s_{1,L+1}, s_{22}, \ldots, s_{L+1, L+1})^{T} -  (\sigma_{11}, \sigma_{12}, \ldots, \sigma_{1,L+1}, \sigma_{22}, \ldots, \sigma_{L+1, L+1})^{T}\right\} \nonumber \\ \overset{d}{\longrightarrow} \mathcal{N}\left(0, (\sigma_{ik}\sigma_{jm} + \sigma_{im}\sigma_{jk} )_{ij, km} \right).
\end{eqnarray}
Hence, we have $\sqrt{n}(s_{jL} - \sigma_{jL}) \overset{d}{\longrightarrow} \mathcal{N}(0, \sigma_{LL}\sigma_{jj} + \sigma_{jL}^2)$. Because $\sigma_{kL} = 0$ for $k \neq L, L+1$ under independent instruments (A4), we have:
\begin{eqnarray*}
n \phi(\mathbf{S}^{(L)}) \approx n(\sigma_{LL}\sigma_{jj})^{-1}(s_{Lj} - \sigma_{Lj})^2 := W_{jL} \overset{d}{\longrightarrow}  \chi^{2}_{1}.
\end{eqnarray*}
Because an Isserlie matrix~\eqref{eq:Isserlis} is diagonal, $W_{jL}$ is independent with $\{W_{kl}, k,l=1,2,\ldots, L: W_{kl} \neq W_{jL}, W_{Lj}\}$, and by the definition of the notation, $W_{kl} = W_{lk}$.
\end{proof}
\clearpage
\setcounter{equation}{0}
	\setcounter{figure}{0}
	\setcounter{table}{0}
	\setcounter{page}{1}
	\setcounter{section}{0}
	\renewcommand{\theequation}{S\arabic{equation}}
	\renewcommand{\thefigure}{S\arabic{figure}}
	\renewcommand{\thesection}{S\arabic{section}}
	\renewcommand{\thetable}{S\arabic{table}}
\renewcommand{\thefigure}{S\arabic{figure}}
\begin{center}
{\huge\bf Supplementary Material}
\end{center}

\section{Additional Discussion about the Model}
We make a few remarks about model \eqref{eq:model1} and \eqref{eq:model2} in the main paper. First, if the magnitudes of $\phi^*$  and $\psi^*$ are equal, but with opposite signs, $\pi^* = 0$ and all of our instruments are valid under Definition \ref{def:validIV}. But, if we consider $\phi^*$ and $\psi^*$ individually, these are violations of the IV assumptions (A2) and (A3). While such scenario will probably not occur in practice, it does raise the limitations of the linear constant effects modeling assumption in model \eqref{eq:model1}. 

Second, our framework assumes an additive, linear, constant effects model between $Y$, $D$, and $Z$, which may not be met in practice. However, similar to the development in weak instrument literature where homogeneity assumption is used as a basis to study properties of IV estimators under near-violations of IV assumption (A1) \citep{staiger_instrumental_1997,stock_survey_2002}, we also find the homogeneity assumption as a good starting point to tackle the question of invalid instruments and violations of (A2) and (A3). Also, our setup is the most popular IV model that's typically introduced in standard econometrics \citep{wooldridge_econometrics_2010}. Nevertheless, we leave it as a future research topic to investigate the effect of heterogeneity and non-linearity in the presence of invalid instruments.

\section{Test Statistics for Theorem 1}\label{sec:choiceTest}
Let $\mathbf{P}_{\mathbf{Z}_A} = \mathbf{Z}_A (\mathbf{Z}_A^T \mathbf{Z}_A)^{-1}\mathbf{Z}_{A}^T$ be the $n$ by $n$ orthogonal projection matrix onto the column space of $\mathbf{Z}_{A}$. Let $\mathbf{R}_{\mathbf{Z}_A}$ be the $n$ by $n$ residual projection matrix so that $\mathbf{R}_{\mathbf{Z}_A} + \mathbf{P}_{\mathbf{Z}_A} = \mathbf{I}_{n}$ and $\mathbf{I}_{n}$ is an $n$ by $n$ identity matrix. 

In the instrumental variables literature, there are many tests of causal effects $T(\beta_0,B)$ that satisfy Theorem 1. We start with a discussion of the most popular test statistic in instrumental variables, the t-test based on two stage least squares. For a given subset $B \subset \{1,\ldots,L\}$, consider the following optimization problem
\begin{equation} \label{eq:tsls_est}
\hat{\bm{\pi}}_{B}^{TSLS(B)}, \hat{\beta}^{TSLS(B)} = {argmin}_{\bm{\pi}_B,\beta} \|\mathbf{P}_{\mathbf{Z}_n} (\mathbf{Y}_n - \mathbf{Z}_{B} \bm{\pi}_B - \mathbf{D}_n \beta)\|_2^2 
\end{equation}
The estimates from \eqref{eq:tsls_est} are known as two-stage least squares estimates of $\bm{\pi}^*$ and $\beta^*$ where for $\beta^*$; note that $\hat{\bm{\pi}}_{B^C}^{TSLS(B)} = 0$ since the estimation assumes $B$ contains all the invalid instruments. Let $\hat{\bm{\epsilon}}^{TSLS(B)}$ be the residuals from the optimization problem, $\hat{\bm{\epsilon}}^{TSLS(B)} = \mathbf{Y}_n - \mathbf{Z}_B \hat{\bm{\pi}}_{B}^{TSLS(B)} - \mathbf{D}_n \hat{\beta}_B^{TSLS(B)}$. Then, the t-test based on the two-stage least squares estimator of $\beta^*$ in \eqref{eq:tsls_est} is
\begin{equation} \label{eq:tsls}
\text{\TSLS}(\beta_0,B) = \sqrt{n - c(B) - 1} \left\{ \frac{\hat{\beta}^{TSLS(B)} - \beta_0}{\sqrt{\|\hat{\bm{\epsilon}}^{TSLS(B)}\|_2^2 / \| \mathbf{R}_{\mathbf{Z}_B}\mathbf{P}_{\mathbf{Z}_n}\mathbf{D}_n \|_2^2}} \right\}
\end{equation}
Under $H_0: \beta^* = \beta_0$ and if $B^* \subseteq B$, \eqref{eq:tsls} is asymptotically standard Normal and consequently, the null is rejected for the alternative $H_a: \beta^* \neq \beta_0$ when $|\text{\TSLS}(\beta_0,B)| \geq z_{1-\alpha/2}$ where $z_{1-\alpha/2}$ is the $1- \alpha/2$ quantile of the standard Normal. Unfortunately, in practice, instruments can be weak and the nominal size of the two-stage least squares based on asymptotic Normal  can be misleading \citep{staiger_instrumental_1997}.

The Anderson-Rubin (AR) test \citep{anderson_estimation_1949} is another simple and popular test in instrumental variables based on the partial F-test of the regression coefficients between $\mathbf{Y}_n - \mathbf{D}_n\beta_0$ versus $\mathbf{Z}_{B^C}$, i.e.
\begin{equation} \label{eq:AR}
\text{\AR}(\beta_0,B) = \frac{ (\mathbf{Y}_n - \mathbf{D}_n \beta_0)^T (\mathbf{P}_{\mathbf{Z}_n} - \mathbf{P}_{\mathbf{Z}_B})(\mathbf{Y}_n  - \mathbf{D}_n \beta_0) / L - c(B)}{ (\mathbf{Y}_n - \mathbf{D}_n \beta_0)^T \mathbf{R}_{\mathbf{Z}_n} (\mathbf{Y}_n - \mathbf{D}_n \beta_0) / (n - L)}
\end{equation}
The Anderson-Rubin test has some attractive properties, including robustness to weak instruments (i.e. violation of (A1)) \citep{staiger_instrumental_1997}, robustness to modeling assumptions on $D_i$, and an exact null distribution under Normality, to name a few; see \citet{dufour_identification_2003} for a full list. A caveat to the Anderson-Rubin test is its lackluster power, especially compared to the conditional likelihood ratio test \citep{moreira_conditional_2003,andrews_optimal_2006, mikusheva_robust_2010}. However, the Anderson-Rubin test can be used as a pretest to check whether the candidate subset of instruments $B$ contains all the invalid instruments \citep{kleibergen_generalizing_2007}. This feature is particularly useful for our problem where we have possibly invalid instruments and we want to exclude subsets $B$ that contain invalid instruments.

Finally, \citet{moreira_conditional_2003} proposed the conditional likelihood ratio test which also satisfies the condition for Theorem 1 and, more importantly, is robust to weak instruments. Specifically, for a given $B$, let $\mathbf{W}_n$ be an $n$ by $2$ matrix where the first column contains $\mathbf{Y}_n$ and the second column contains $\mathbf{D}_n$. Let $\mathbf{a}_0 = (\beta_0, 1)$ and $\mathbf{b}_0 = (1, -\beta_0)$ to be two-dimensional vectors and $\hat{\Sigma} = \mathbf{W}_n^T \mathbf{R}_{\mathbf{Z}_n} \mathbf{W}_n / (n - L)$.  Consider the 2 by 2 matrix $Q(\beta_0,B)$
\begin{align*}
Q(\beta_0,B) &= 
\begin{pmatrix}
Q_{11}(\beta_0,B) & Q_{12}(\beta_0,B) \\
Q_{21}(\beta_0,B) & Q_{22}(\beta_0,B)
\end{pmatrix} \\
&= 
\begin{pmatrix}
\frac{\mathbf{b}_0^T \mathbf{W}_n^T (\mathbf{P}_{\mathbf{Z}_n} - \mathbf{P}_{\mathbf{Z}_B}) \mathbf{W}_n \mathbf{b}_0}{\mathbf{b}_0^T \hat{\Sigma} \mathbf{b}_0} & \frac{ \mathbf{b}_0^T \mathbf{W}_n^T (\mathbf{P}_{\mathbf{Z}_n} - \mathbf{P}_{\mathbf{Z}_B}) \mathbf{W}_n \hat{\Sigma}^{-1} \mathbf{a}_0}{\sqrt{\mathbf{a}_0^T \hat{\Sigma}^{-1} \mathbf{a}_0} \sqrt{\mathbf{b}_0^T \hat{\Sigma}^{-1} \mathbf{b}_0}} \\
\frac{ \mathbf{b}_0^T \mathbf{W}_n^T (\mathbf{P}_{\mathbf{Z}_n} - \mathbf{P}_{\mathbf{Z}_B}) \mathbf{W}_n \hat{\Sigma}^{-1} \mathbf{a}_0}{\sqrt{\mathbf{a}_0^T \hat{\Sigma}^{-1} \mathbf{a}_0} \sqrt{\mathbf{b}_0^T \hat{\Sigma}^{-1} \mathbf{b}_0}}&\frac{\mathbf{a}_0^T \hat{\Sigma}^{-1} \mathbf{W}_n^T (\mathbf{P}_{\mathbf{Z}_n} - \mathbf{P}_{\mathbf{Z}_B}) \mathbf{W}_n \hat{\Sigma}^{-1} \mathbf{a}_0 }{\mathbf{a}_0^T \hat{\Sigma}^{-1} \mathbf{a}_0}
\end{pmatrix}
\end{align*}
Then, the conditional likelihood ratio test \citep{moreira_conditional_2003} is 
\begin{align}
\text{\CLR}(\beta_0,B) &= \frac{1}{2} \left\{Q_{11}(\beta_0,B) - Q_{22}(\beta_0,B) \right\} \label{eq:CLR} \\
&+ \frac{1}{2}\sqrt{\{Q_{11}(\beta_0,B) + Q_{22}(\beta_0,B)\}^2 - 4\{Q_{11}(\beta_0,B) Q_{22}(\beta_0,B) - Q_{12}^2(\beta_0,B)\}} \nonumber
\end{align}
The null $H_0: \beta^* = \beta_0$ for the conditional likelihood ratio test is rejected at level $\alpha$ when $\text{\CLR}(\beta_0,B) \geq q_{1-\alpha}^{CLR}$ where $q_{1-\alpha}^{CLR}$ is the $1 -\alpha$ quantile of a conditional null distribution dependent on $Q_{22}(\beta_0,B)$ (see \citet{andrews_optimal_2006} for computing the exact quantile). 

The Anderson-Rubin test and the conditional likelihood ratio test, both of which are robust to weak instruments, have characteristics that are unique to each test. \citep{staiger_instrumental_1997, stock_survey_2002, moreira_conditional_2003, dufour_identification_2003, andrews_optimal_2006}. There is no uniformly most powerful test among the two tests, but \citet{andrews_optimal_2006} and \citet{mikusheva_robust_2010} suggest using \eqref{eq:CLR} due to its generally favorable power compared to the Anderson-Rubin test in most cases when weak instruments are present. However, the Anderson-Rubin test has the unique feature that it can be used as a pretest to check whether the candidate subset of instruments $B$ contains all the invalid instruments. Also, between the two tests, the Anderson-Rubin test is the simplest in that it can be written as a standard F-test in regression. In addition, the conditional likelihood ratio test requires an assumption that the exposure, $D_i$, is linearly related to the instruments $\mathbf{Z}_{i\cdot}$; the Anderson-Rubin test does not require this linearity assumption \citep{dufour_identification_2003}.

\section{Test Statistic for Theorem 2: The Sargan Test}
There are pretests in the instrumental variables literature that satisfy the conditions for Theorem 2. The most well-known pretest is the Sargan test for overidentification \citep{sargan_estimation_1958}, which tests, among other things, whether the instruments $B^C$ contain only valid instruments, $\bm{\pi}_{B^C}^* = 0$. The Sargan test is
\begin{equation} \label{eq:SarganTest}
\text{\SAR}(B) = \frac{\| (\mathbf{P}_{\mathbf{Z}_n} - \mathbf{P}_{\mathbf{Z}_B}) \hat{\bm{\epsilon}}^{TSLS(B)}\|_2^2}{\|\hat{\bm{\epsilon}}^{TSLS(B)}\|_2^2/n}
\end{equation}
The null hypothesis is rejected at level $\alpha_s$ when $\text{\SAR}(B)$ exceeds the $1-\alpha_s$ quantile of a Chi-square distribution with $c(B)-1$ degrees of freedom. Thus, if we use the Sargan test as a pretest for $C_{1-\alpha}'(\mathbf{Y}_n,\mathbf{D}_n,\mathbf{Z}_n)$, then $q_{1 - \alpha_s}$ in our pretest confidence interval would be the $1 - \alpha_s$ quantile of a Chi-square distribution with $c(B) - 1$ degrees of freedom and we would only proceed to construct a confidence interval with the test statistic $T(\beta_0,B)$ at $1 - \alpha_t$ if the null hypothesis is retained. Unfortunately, the null of the Sargan test can be misleading when weak instruments are present \citep{staiger_instrumental_1997} and to the best of our knowledge, pretests that are robust to weak instruments do not exist.

\section{Power Under Invalid Instruments} \label{sec:powerTheory}
\subsection{Anderson-Rubin Test} \label{sec:powerSetup}
To study power under invalid instruments, consider the following model 
\begin{subequations} \label{eq:powerSetupSupp}
	\begin{align}
	Y_i &= \mathbf{Z}_{i\cdot}^T \boldsymbol{\pi}^* + D_i \beta^* + \epsilon_i, \quad{} E(\epsilon_i, \xi_i | \mathbf{Z}_{i\cdot}) =0 \\
	D_i &= \mathbf{Z}_{i\cdot}^T \boldsymbol{\gamma}^* + \xi_i \\
	\begin{pmatrix} \epsilon_i \\ \xi_i \end{pmatrix} &\sim N\left[\begin{pmatrix} 0 \\ 0 \end{pmatrix},\begin{pmatrix} \sigma_2^2  & \rho \sigma_1 \sigma_2 \\ \rho \sigma_1 \sigma_2  & \sigma_1^2 \end{pmatrix}\right]
	\end{align}
\end{subequations}
The setup above is a special case of model \eqref{eq:model2} with two additional assumptions. First, the treatment variable $D_i$ is linearly associated to $\mathbf{Z}_{i\cdot}$. Second, the error terms are bivariate i.i.d. Normal with an unknown covariance matrix. Both the linearity of $D_i$ and the Normality assumption are common in the IV literature to theoretically study the property of tests \citep{staiger_instrumental_1997,andrews_performance_2007}. 

Under the model in \eqref{eq:powerSetupSupp}, we study whether a particular statistical test has the power to detect the alternative $H_a: \beta^* \neq \beta_0; \boldsymbol{\pi}_{B^C}^* \neq 0$ under the null $H_0: \beta^* = \beta_0; \boldsymbol{\pi}_{B^C}^* = \mathbf{0}$ for a given set $B$. The first alternative $\beta^* \neq \beta_0$ measures whether the treatment is away from the null value $\beta_0$. The second alternative $\boldsymbol{\pi}_{B^C}^* \neq \mathbf{0}$  measures whether a wrong subset $B$ is in the union of $C_{1-\alpha}(\mathbf{Y}_n,\mathbf{D}_n,\mathbf{Z}_n)$. A wrong subset $B$ is where $B$ does not contain all the invalid instruments so that $\boldsymbol{\pi}_{B^C}^* \neq \mathbf{0}$. If a test has good power against the second alternative, we would be less likely to take unions over wrong $B$s in $C_{1-\alpha}(\mathbf{Y}_n,\mathbf{D}_n,\mathbf{Z}_n)$ and our union confidence interval will tend to be shorter.

Under the framework introduced in Section \ref{sec:powerSetup}, Theorem \ref{thm:powerAR} shows the power of the Anderson-Rubin test under invalid instruments.
\begin{theorem} \label{thm:powerAR} Consider any set $B \subset \{1,\ldots,L\}$ with $c(B) = \bar{s} - 1$ and model \eqref{eq:powerSetupSupp}. Suppose we are testing the null hypothesis $H_0: \beta^* = \beta_0, \boldsymbol{\pi}_{B^C}^* = \mathbf{0}$ against the alternative $H_a: \beta^* \neq \beta_0$ or $\boldsymbol{\pi}_{B^C}^* \neq \mathbf{0}$. The exact power of $\text{AR}(\beta_0,B)$ in \eqref{eq:AR} is 
	\begin{equation} \label{eq:powerAR}
	\text{pr}\{\AR(\beta_0,B) \geq q_{1-\alpha}^{F_{L - c(B),n-L,0}} \} = 1 - F_{L - c(B), n - L,\eta(B)}(q_{1-\alpha}^{F_{L - c(B),n-L,0}})
	\end{equation}
	where $q_{1-\alpha}^{F_{L - c(B),n-L,\eta(B)}}$ is the $1 -\alpha$ quantile of the non-central F distribution with degrees of freedom $L - c(B)$, $n - L$ and non-centrality parameter $\eta(B) = ||\mathbf{R}_{\mathbf{Z}_B} \mathbf{Z}_{B^C}(\boldsymbol{\pi}_{B^C}^* + \boldsymbol{\gamma}_{B^C}^* (\beta^* - \beta_0)) ||_2^2$.
\end{theorem}
Theorem \ref{thm:powerAR} generalizes the power of the Anderson-Rubin test when some instruments are invalid. Specifically, suppose $B$ contains all the invalid instruments so that $B^* \subseteq B $ and $\boldsymbol{\pi}_{B^C}^* = 0$. Then, the non-centrality parameter $\eta(B)$ in Theorem \ref{thm:powerAR} would only consist of instruments' strength, specifically $||\mathbf{R}_{\mathbf{Z}_B} \mathbf{Z}_{B^C} \boldsymbol{\gamma}_{B^C}^* (\beta^* - \beta_0)) ||_2^2$, and we would return to the usual power of the Anderson-Rubin test with all valid instruments. On the other hand, if $B$ does not contain invalid instruments so that $B^* \not\subseteq B$ and $\boldsymbol{\pi}_{B^C}^* \neq \mathbf{0}$, the Anderson-Rubin test will still have power to reject $H_0$, even if $\beta^* = \beta_0$. In other words, the Anderson-Rubin test will reject $H_0$ and will generally have shorter intervals when $B$ does not contain all the invalid instruments. Also, Theorem \ref{thm:powerAR} shows that the Anderson-Rubin has no power when $\boldsymbol{\pi}_{B^C}* +  \boldsymbol{\gamma}_{B^C}^* (\beta^* - \beta_0) = 0$; a similar result was shown in \citet{kadane_comment_1977} and \citet{small_sensitivity_2007} when studying the power of overidentifying restrictions tests. Finally, we note that our power formula is exact and does not invoke asymptotics. 

\begin{proof}[Proof of Theorem~\ref{thm:powerAR}] 
By Cochran's theorem, (i) the numerator and the denominator of \eqref{eq:AR} are independent, (ii) the denominator, scaled by $\tilde{\sigma}^2 = \sigma_2^2 + (\beta^* - \beta_0)^2 \sigma_1^2 + 2(\beta^* - \beta_0)\rho \sigma_1 \sigma_2$, is a central chi-square with $n-L$ degrees of freedom, i.e.
\[
\frac{(\mathbf{Y}_n - \mathbf{D}_n \beta_0)^T \mathbf{R}_{\mathbf{Z}_n} (\mathbf{Y}_n - \mathbf{D}_n \beta_0)}{\tilde{\sigma}^2 (n - L)} = \frac{\{(\beta^* - \beta_0) \bm{\xi} + \bm{\epsilon}\}^T \mathbf{R}_{\mathbf{Z}_n} \{(\beta^* - \beta_0) \bm{\xi} + \bm{\epsilon}\}}{\tilde{\sigma}^2 (n - L)} \sim \chi_{n - L,0}^2
\]
and (iii), the numerator, scaled by $\tilde{\sigma}^2$, is a non-central chi-square distribution with non-centrality $\eta(B)$
\[
 \frac{ (\mathbf{Y}_n - \mathbf{D}_n \beta_0)^T (\mathbf{P}_{\mathbf{Z}_n} - \mathbf{P}_{\mathbf{Z}_B})(\mathbf{Y}_n  - \mathbf{D}_n \beta_0)}{\tilde{\sigma}^2 \{L - c(B)\}} \sim \chi_{L - c(B),\eta(B)}^2, \eta(B) = \| (\mathbf{P}_{\mathbf{Z}_n} - \mathbf{P}_{\mathbf{Z}_B}) \mathbf{Z}_n\{ \boldsymbol{\pi}^* + \boldsymbol{\gamma}^*(\beta^* - \beta_0)\} \|_2^2
\]
Since $(\mathbf{P}_{\mathbf{Z}_n} - \mathbf{P}_{\mathbf{Z}_B})\mathbf{Z}_n$ can be rewritten as the residual projection of $\mathbf{Z}_n$ onto $\mathbf{Z}_{B}$, i.e.
\[
(\mathbf{P}_{\mathbf{Z}_n}  - \mathbf{P}_{\mathbf{Z}_B})\mathbf{Z}_n = \mathbf{Z}_n - [\mathbf{Z}_B : \mathbf{P}_{\mathbf{Z}_B} \mathbf{Z}_{B^C}] = [0 : \mathbf{Z}_{B^C} - \mathbf{P}_{\mathbf{Z}_B} \mathbf{Z}_{B^C}] =[0 : \mathbf{R}_{\mathbf{Z}_B} \mathbf{Z}_{B^C}] = \mathbf{R}_{\mathbf{Z}_B} \mathbf{Z}_n
\]
the $\text{\AR}(\beta_0,B)$ is a non-central F distribution with degrees of freedom. 
\end{proof}

\subsection{Two-stage Least Squares} \label{sec:asymptoticPower}
This section derives power  of the two-stage least squares test statistic under invalid instruments by using local asymptotics. We also evaluate the accuracy of the asymptotic approximation to power in Section \ref{sec:accuracy}. 

As before, consider the setup in equation \eqref{eq:powerSetupSupp}. Given a subset $B \subset \{1,\ldots,L\}$, we consider the the null hypothesis $H_0: \beta^* = \beta_0; \bm{\pi}_{B^C}^* = \mathbf{0}$ versus the local alternative $H_a: \beta^*  = \beta_0 + \Delta_1/\sqrt{n}; \bm{\pi}_{B^C}^* = \bm{\Delta}_2 /\sqrt{n}$ where $\Delta_{1} \neq 0$, $\bm{\Delta}_2 \neq \mathbf{0}$. This differs from the alternative in Section \ref{sec:powerSetup} because it is  $\sqrt{n}$ ``local'' to the null hypothesis; see Section \ref{sec:accuracy} and \citet{lehmann_elements_2004} for details on local alternative hypothesis. We also note that this type of asymptotics, specifically the alternative $\bm{\pi}_{B^C}^* = \bm{\Delta}_2 / \sqrt{n}$ was use in \citet{staiger_instrumental_1997} to characterize the behavior of the Sargan test under weak instruments. 

Theorem \ref{thm:powerTSLS} states the asymptotic power of the two-stage least squares test under invalid instruments.
\begin{theorem} \label{thm:powerTSLS} Consider any set $B \subset \{1,\ldots,L\}$ with $c(B) = \bar{s} - 1$ and model \eqref{eq:powerSetupSupp}. Suppose we are testing the null hypothesis $H_0: \beta^* = \beta_0, \bm{\pi}_{B^C}^* = \mathbf{0}$ against the local alternative $H_a: \beta^*  = \beta_0 + \Delta_1/\sqrt{n}$, $\bm{\pi}_{B^C}^* = \bm{\Delta}_2 /\sqrt{n}$ where $\Delta_{1} \neq 0$, $\bm{\Delta}_2 \neq \mathbf{0}$. Let $\mu(B) = \lim_{n \to \infty} {\bm{\gamma}_{B^C}^*}^T \mathbf{Z}_{B^C}^T \mathbf{R}_{\mathbf{Z}_B} \mathbf{Z}_{B^C} \bm{\gamma}_{B^C}^*/n \neq 0$ and $\kappa(B) = \lim_{n \to \infty} {\bm{\gamma}_{B^C}^*}^T $ $\mathbf{Z}_{B^C}^T \mathbf{R}_{\mathbf{Z}_B} \mathbf{Z}_{B^C} \bm{\Delta}_{2} /n$. Then, as $n\to \infty$, the asymptotic power of $\text{\TSLS}(\beta_0,B)$ in \eqref{eq:tsls} is
\begin{align*}
\text{pr}(|TSLS(\beta_0,B)|\geq z_{1-\alpha/2}) \to& 1 - \Phi\left\{z_{1-\alpha/2} - \left(\frac{\Delta_1}{\sqrt{\sigma_{2}^2/\mu(B)}} + \frac{\kappa(B)}{\sqrt{\sigma_{2}^2 \mu(B)}}\right)\right\} \\
&+ \Phi\left\{-z_{1-\alpha/2} - \left(\frac{\Delta_1}{\sqrt{\sigma_{2}^2/\mu(B)}} + \frac{\kappa(B)}{\sqrt{\sigma_{2}^2\mu(B) }}\right)\right\}
\end{align*}
where $\Phi(\cdot)$ is the cumulative distribution function of the Normal distribution.
\end{theorem}


In Theorem \ref{thm:powerTSLS}, the term $\mu(B)$ in the power formula represents the concentration parameter of the instruments in set $B^C$ (up to a scaling by the variance) \citep{stock_survey_2002}. The term $\kappa(B)$ represents the interaction between instruments' strength (via $\boldsymbol{\gamma}_{B^C}^*$) and instruments' invalidity (via $\bm{\Delta}_2$) for the instruments in set $B^C$. Note that if all the instruments in $B^C$ are actually valid, $\kappa(B) = 0$ and we would end up with the usual power formula for two-stage least squares. The presence of invalid instruments essentially shifts the usual power curve by a factor $\kappa(B) / \sqrt{\sigma_{2}^2  \mu(B)}$. But, if the instruments are stronger than the effects of the invalid instruments so that $\kappa(B) / \sqrt{\mu(B)} \approx 0$, then you would have the usual power formula for the two-stage least squares that assumes all the instruments are valid. Hence, for two-stage least squares, having strong instruments can mitigate the effect of invalid instruments.

Between the power of the Anderson-Rubin test and the two-stage least squares under invalid instruments, the Anderson-Rubin test does not rely on the variance $\sigma_2^2$. Also, geometrically speaking, the power of the Anderson-Rubin test depends ``spherically''  (i.e. in $\ell_2$ ball) via the term $\| \mathbf{R}_{\mathbf{Z}_B} \mathbf{Z}_{B^C} \bm{\pi}_{B^C}^*\|_2^2$ while the power of the two-stage least squares depends linearly by $\kappa(B)/\sqrt{\mu(B)}$. Intuitively, this implies that the Anderson-Rubin should be able to detect invalid instruments in $B^C$ across all directions of $\bm{\pi}_{B^C}^*$ while the two-stage least squares will only be able to detect invalid instruments in $B^C$ only along certain directions. 

\begin{proof}[Proof of Theorem~\ref{thm:powerTSLS}]
We analyze the test statistic $\text{\TSLS}(\beta_0,B)$ by looking at the numerator and the denominator separately. First, the two-stage least squares estimator, which makes up the numerator of $\text{\TSLS}(\beta_0,B)$, can be written as follows.
\begin{align*}
&\hat{\beta}^{TSLS(B)} \\
=& \frac{\mathbf{D}_n^T \mathbf{R}_{\mathbf{Z}_B} \mathbf{Z}_{B^C} (\mathbf{Z}_{B^C}^T \mathbf{R}_{\mathbf{Z}_B} \mathbf{Z}_{B^C})^{-1} \mathbf{Z}_{B^C}^T \mathbf{R}_{\mathbf{Z}_B} \mathbf{Y}_n}{\mathbf{D}_n^T \mathbf{R}_{\mathbf{Z}_B} \mathbf{Z}_{B^C} (\mathbf{Z}_{B^C}^T \mathbf{R}_{\mathbf{Z}_B} \mathbf{Z}_{B^C})^{-1} \mathbf{Z}_{B^C}^T \mathbf{R}_{\mathbf{Z}_B} \mathbf{D}_n}   \\
=& \beta^* + \frac{\mathbf{D}_n^T \mathbf{R}_{\mathbf{Z}_B} \mathbf{Z}_{B^C} \bm{\pi}_{B^C}^*}{\mathbf{D}_n^T \mathbf{R}_{\mathbf{Z}_B} \mathbf{Z}_{B^C} (\mathbf{Z}_{B^C}^T \mathbf{R}_{\mathbf{Z}_B} \mathbf{Z}_{B^C})^{-1} \mathbf{Z}_{B^C}^T \mathbf{R}_{\mathbf{Z}_B} \mathbf{D}_n} + \frac{\mathbf{D}_n^T \mathbf{R}_{\mathbf{Z}_B} \mathbf{Z}_{B^C} (\mathbf{Z}_{B^C}^T \mathbf{R}_{\mathbf{Z}_B} \mathbf{Z}_{B^C})^{-1} \mathbf{Z}_{B^C}^T \mathbf{R}_{\mathbf{Z}_B} \bm{\epsilon}}{\mathbf{D}_n^T \mathbf{R}_{\mathbf{Z}_B} \mathbf{Z}_{B^C} (\mathbf{Z}_{B^C}^T \mathbf{R}_{\mathbf{Z}_B} \mathbf{Z}_{B^C})^{-1} \mathbf{Z}_{B^C}^T \mathbf{R}_{\mathbf{Z}_B} \mathbf{D}_n}
\end{align*}
Some algebra can show that the terms in the denominator can be reduced to 
\[
\mathbf{D}_n^T \mathbf{R}_{\mathbf{Z}_B} \mathbf{Z}_{B^C} (\mathbf{Z}_{B^C}^T \mathbf{R}_{\mathbf{Z}_B} \mathbf{Z}_{B^C})^{-1} \mathbf{Z}_{B^C}^T \mathbf{R}_{\mathbf{Z}_B} \mathbf{D}_n = \| \mathbf{R}_{\mathbf{Z}_B} \mathbf{P}_\mathbf{Z} \mathbf{D}_n \|_2^2
\]
Also, using standard asymptotics, we can arrive at the following limiting quantities
\begin{align*}
&\frac{1}{n} \mathbf{D}_n^T \mathbf{R}_{\mathbf{Z}_B} \mathbf{Z}_{B^C} (\mathbf{Z}_{B^C}^T \mathbf{R}_{\mathbf{Z}_B} \mathbf{Z}_{B^C})^{-1} \mathbf{Z}_{B^C}^T \mathbf{R}_{\mathbf{Z}_B} \mathbf{D}_n &&\to 
\mu(B) \neq 0 \\
&\frac{1}{\sqrt{n}} \mathbf{D}_n^T \mathbf{R}_{\mathbf{Z}_B} \mathbf{Z}_{B^C} \frac{\bm{\Delta}_2}{\sqrt{n}} &&\to
\kappa(B) \\
&\frac{1}{\sqrt{n}} \mathbf{D}_n^T \mathbf{R}_{\mathbf{Z}_B} \mathbf{Z}_{B^C} (\mathbf{Z}_{B^C}^T \mathbf{R}_{\mathbf{Z}_B} \mathbf{Z}_{B^C})^{-1} \mathbf{Z}_{B^C}^T \mathbf{R}_{Z_B}\bm{\epsilon} &&\to N(0,\sigma_{2}^2  \mu(B)) 
\end{align*}
where the first two relies on the law of large numbers and the last limit result uses a central limit theorem. For the denominator of $\text{\TSLS}(\beta_0,B)$, we can rewrite $\|\hat{\bm{\epsilon}}^{TSLS(B)} \|_2^2 / (n - c(B) - 1)$ as follows.
\[
\frac{\|\hat{\bm{\epsilon}}^{TSLS(B)} \|_2^2}{n - c(B) - 1} = \frac{\| \mathbf{R}_{\mathbf{Z}_B}(\mathbf{Y}_n - \mathbf{D}_n \hat{\beta}^{TSLS(B)}) \|_2^2}{n - c(B) - 1}  = \frac{ \|\mathbf{R}_{\mathbf{Z}_B} \mathbf{Z}_{B^C} \bm{\pi}_{B^C}^* + \mathbf{R}_{\mathbf{Z}_B} \mathbf{D}_n (\hat{\beta}^{TSLS(B)} - \beta^*) + \mathbf{R}_{\mathbf{Z}_B} \bm{\epsilon}\|_2^2 }{n - c(B) - 1}
\]
Under the null, $\bm{\pi}_{B^C}^* = 0$ and $\hat{\beta}^{TSLS(B)} - \beta^* \to 0$, leading to $\|\hat{\bm{\epsilon}}^{TSLS(B)}\|_2^2/(n - c(B) - 1) \to \sigma_{2}^2$. Under the alternative, $ \|\mathbf{R}_{\mathbf{Z}_B} \mathbf{Z}_{B^C} \bm{\pi}_{B^C}^* \|_2^2 / (n - c(B) - 1) = \|\mathbf{R}_{\mathbf{Z}_B} \mathbf{Z}_{B^C} \bm{\Delta}_{2} \|_2^2/(n (n - c(B) - 1)) \to 0$ and $\beta^{TSLS(B)} - \beta^* \to 0$, again leading to $\|\hat{\bm{\epsilon}}^{TSLS(B)}\|_2^2/(n - c(B) - 1) \to \sigma_{2}^2$. Then, under both the null and the alternative, we have
\begin{align*}
&\sqrt{n - c(B) - 1} \frac{\hat{\beta}^{TSLS(B)} - \beta^*}{\sqrt{\|\hat{\bm{\epsilon}}^{TSLS(B)}\|_2^2 /  \| \mathbf{R}_{\mathbf{Z}_B}\mathbf{P}_{\mathbf{Z}_n}\mathbf{D}_n \|_2^2}} \\
=& \frac{\sqrt{n}\left\{
 \frac{\mathbf{D}_n^T \mathbf{R}_{\mathbf{Z}_B} \mathbf{Z}_{B^C} \bm{\pi}_{B^C}^*/n}{\mathbf{D}_n^T \mathbf{R}_{\mathbf{Z}_B} \mathbf{Z}_{B^C} (\mathbf{Z}_{B^C}^T \mathbf{R}_{\mathbf{Z}_B} \mathbf{Z}_{B^C})^{-1} \mathbf{Z}_{B^C}^T \mathbf{R}_{\mathbf{Z}_B} \mathbf{D}_n/n} + \frac{\mathbf{D}_n^T \mathbf{R}_{\mathbf{Z}_B} \mathbf{Z}_{B^C} (\mathbf{Z}_{B^C}^T \mathbf{R}_{\mathbf{Z}_B} \mathbf{Z}_{B^C})^{-1} \mathbf{Z}_{B^C}^T \mathbf{R}_{\mathbf{Z}_B} \bm{\epsilon}/n}{\mathbf{D}_n^T \mathbf{R}_{\mathbf{Z}_B} \mathbf{Z}_{B^C} (\mathbf{Z}_{B^C}^T \mathbf{R}_{\mathbf{Z}_B} \mathbf{Z}_{B^C})^{-1} \mathbf{Z}_{B^C}^T \mathbf{R}_{\mathbf{Z}_B} \mathbf{D}_n/n}
\right\}}{\sqrt{ \{\|\hat{\bm{\epsilon}}^{TSLS(B)}\|_2^2 / (n - c(B) - 1)\}  / ( \| \mathbf{R}_{\mathbf{Z}_B}\mathbf{P}_{\mathbf{Z}_n}\mathbf{D}_n \|_2^2/n)}}\\
\to& N\left(\frac{\kappa(B)}{\sqrt{\mu(B) \sigma_2^2}},1\right) 
\end{align*}
Combining all of it together, the local power of $\text{\TSLS}(\beta_0,B)$ is
\begin{align*}
&\text{pr}\left\{ \sqrt{n - c(B) - 1} \left| \frac{\hat{\beta}^{TSLS(B)} - \beta_0}{\sqrt{\|\hat{\bm{\epsilon}}^{TSLS(B)}\|_2^2 / \| \mathbf{R}_{\mathbf{Z}_B} \mathbf{P}_{\mathbf{Z}_n}\mathbf{D}_n \|_2^2}} \right| \geq z_{1-\alpha/2} \right\} \\
=& \text{pr}\left\{ \sqrt{n - c(B) - 1} \frac{\hat{\beta}^{TSLS(B)} - \beta^* + \Delta_{1}/\sqrt{n}}{\sqrt{\|\hat{\bm{\epsilon}}^{TSLS(B)}\|_2^2 / \| \mathbf{R}_{\mathbf{Z}_B} \mathbf{P}_{\mathbf{Z}_n}\mathbf{D}_n \|_2^2}} \geq z_{1-\alpha/2} \right\} \\
&\quad{}+ \text{pr}\left\{ \sqrt{n - c(B) - 1} \frac{\hat{\beta}^{TSLS(B)} - \beta^* + \Delta_{1}/\sqrt{n}}{\sqrt{\|\hat{\bm{\epsilon}}^{TSLS(B)}\|_2^2 / \| \mathbf{R}_{\mathbf{Z}_B} \mathbf{P}_{\mathbf{Z}_n}\mathbf{D}_n \|_2^2}} \leq -z_{1-\alpha/2} \right\} \\
\to& 1 - \Phi\left\{z_{1-\alpha} - \left(\frac{\Delta_1}{\sqrt{\sigma_{2}^2/\mu(B)}} + \frac{\kappa(B)}{\sqrt{\sigma_{2}^2 \mu(B)}}\right)\right\} + \Phi\left\{-z_{1-\alpha} - \left(\frac{\Delta_1}{\sqrt{\sigma_{2}^2/\mu(B)}} + \frac{\kappa(B)}{\sqrt{\sigma_{2}^2\mu(B) }}\right)\right\}
\end{align*}
where $\Phi(\cdot)$ is the cumulative distribution function of the Normal distribution.
\end{proof}


\subsection{Accuracy of Asymptotic Power Formula Under Invalid Instruments}  \label{sec:accuracy}
In this section, we assess how accurate our asymptotic power framework in Theorem \ref{thm:powerTSLS} is in approximating finite-sample power of two-stage least squares under invalid instruments. Specifically, for each fixed sample size and a selected set of invalid instruments $B$, we will empirically estimate the power curve of the two-stage least squares under different alternatives of $\beta^*$ and $\bm{\pi}^*$. We will also compute the asymptotic power expected from theoretical calculations in Theorem \ref{thm:powerTSLS} under different alternatives of $\beta^*$ and $\bm{\pi}^*$. Afterwards, we will contrast the two power curves to see if the asymptotic power provides a decent approximation to the empirically generated power curve at each sample size $n$. We expect that as $n$ grows, the theoretical asymptotic power will match very closely with the empirically generated power curve for various alternatives and different selections of $B$. As we will see, our asymptotic framework provides a good approximation of power under invalid instruments even at $n = 250$.

The setup for our power curve comparisons are as follows. For the data generating model, we follow the model in Section \ref{sec:powerTheory}. For the parameters in the data generating model, we let $\epsilon_i$ and $\xi_i$ have mean zero, variance one, and covariance $0.8$. We assume $L = 10$ instruments are generated from a multivariate Normal with zero mean and identity covariance matrix. For instrument strength, we set $\gamma_j^* = 1$ for all $j = 1,\ldots,L$, which indicates strong instruments; note that since two-stage least squares behave well only under strong instruments, we will only compare power under this scenario. Among the 10 instruments, there are $s^{*} = 3$ invalid instruments. We vary $\boldsymbol{\pi}^*$ and $\beta^*$, which are the alternatives in our testing framework, $H_0: \beta^*=  \beta_0, \boldsymbol{\pi}_{B}^* = \mathbf{0}$ versus the alternative $H_a: \beta^* \neq \beta_0$, $\boldsymbol{\pi}_{B}^* \neq \mathbf{0}$. We compute power under $n = 50, 100, 250, 500, 1000$, and $5000$. We repeat the simulation $5000$ times.

The first simulation in Figure \ref{fig:powerApprox1} demonstrates power when $B = B^*$ and $\beta^* - \beta_0$ varies. Under $B = B^*$, one knows exactly which instruments are valid so that $\boldsymbol{\pi}_{B}^* = 0$ under the null and the alternative and the power of the two-stage least squares reduces to the usual power curve for two-stage least squares. Figure \ref{fig:powerApprox1} shows that the dotted lines, which are the asymptotic power curves from theory, are very close to the solid lines, which are the finite sample power curves, as the sample size increases. In fact, after $n \geq 100$, the theoretical power curve and the empirically generated power curve are nearly indistinguishable. The result in Figure \ref{fig:powerApprox1} suggests that the theory-based asymptotic power curve is a good approximation of two-stage least square's power even for relatively small sample size.

\begin{figure}[H]
\includegraphics[width=\textwidth]{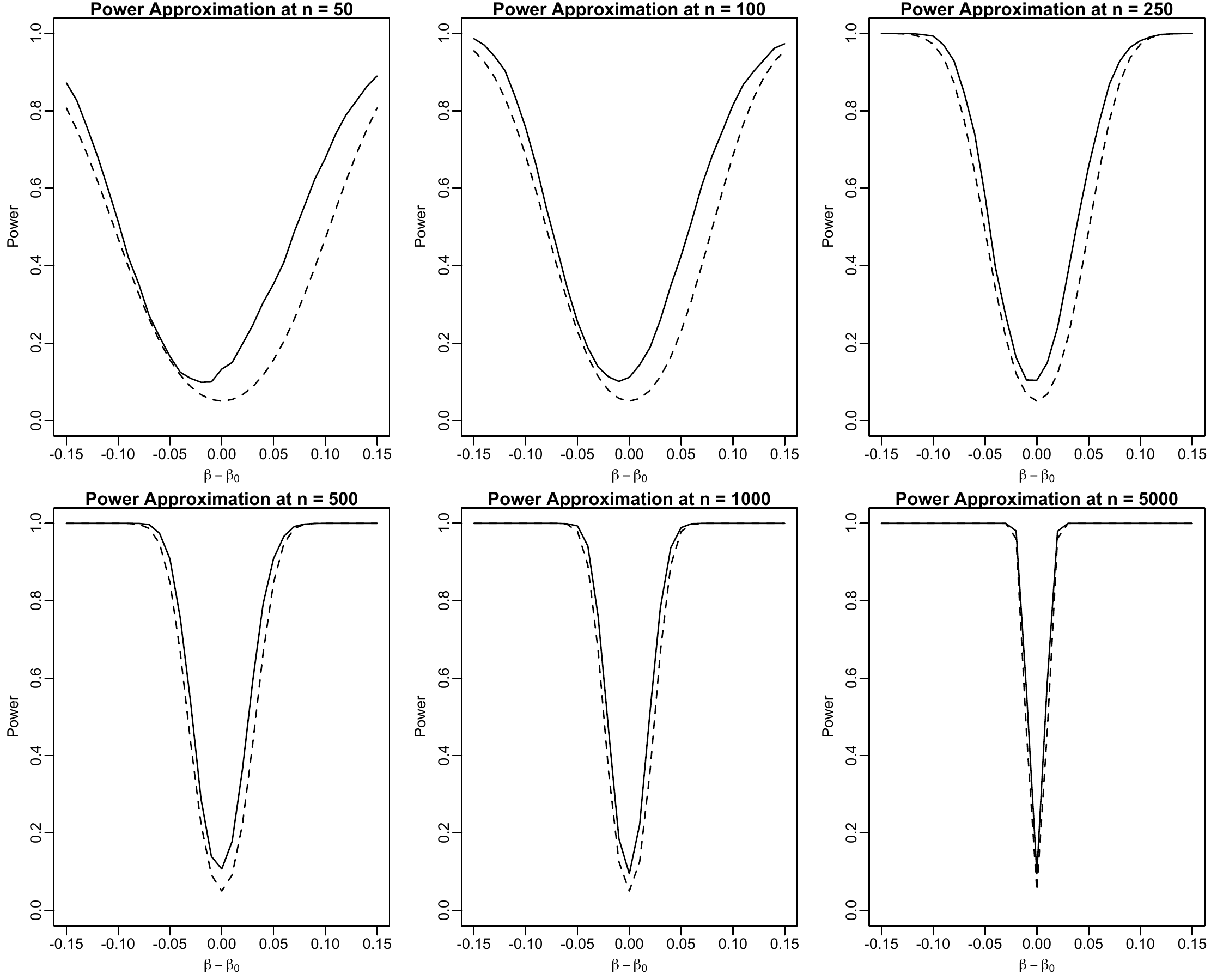}
\caption{Power curves for two-stage least squares under invalid instruments. The dotted line represents the asymptotic power curve and the solid line represents the simulated power. We set $B = \{1,2,3\}$ and $B^* = \{1,2,3\}$; $B$ contains all the invalid instruments.}
\label{fig:powerApprox1}
\end{figure}

The next simulation in Figure \ref{fig:powerApprox2} demonstrates power when $B$ doesn't contain all the invalid instruments. Specifically, we set $B = \{1,2\}$ and $\pi^* = (1,1,1,0,\ldots,0)$ so that $B^* = \{1,2,3\}$ and carry out the simulations. Figure \ref{fig:powerApprox1} shows, again, the asymptotic power curve matches the empirical power curve as sample size grows. After $n \geq 100$, the asymptotic power approximates the finite-sample behavior of the two-stage least squares under invalid instruments. We also observe similar behavior when we set $\boldsymbol{\pi}^* = (2,2,2,0,\ldots,0)$, $\boldsymbol{\pi}^* =(0.5,0.5,0.5,0,\ldots,0)$, $\boldsymbol{\pi}^* = (0.5,0.5,2,0,\ldots,0)$, and $\boldsymbol{\pi}^* = (2,2,0.5,0,\ldots,0)$.

\begin{figure}[H]
\includegraphics[width=\textwidth]{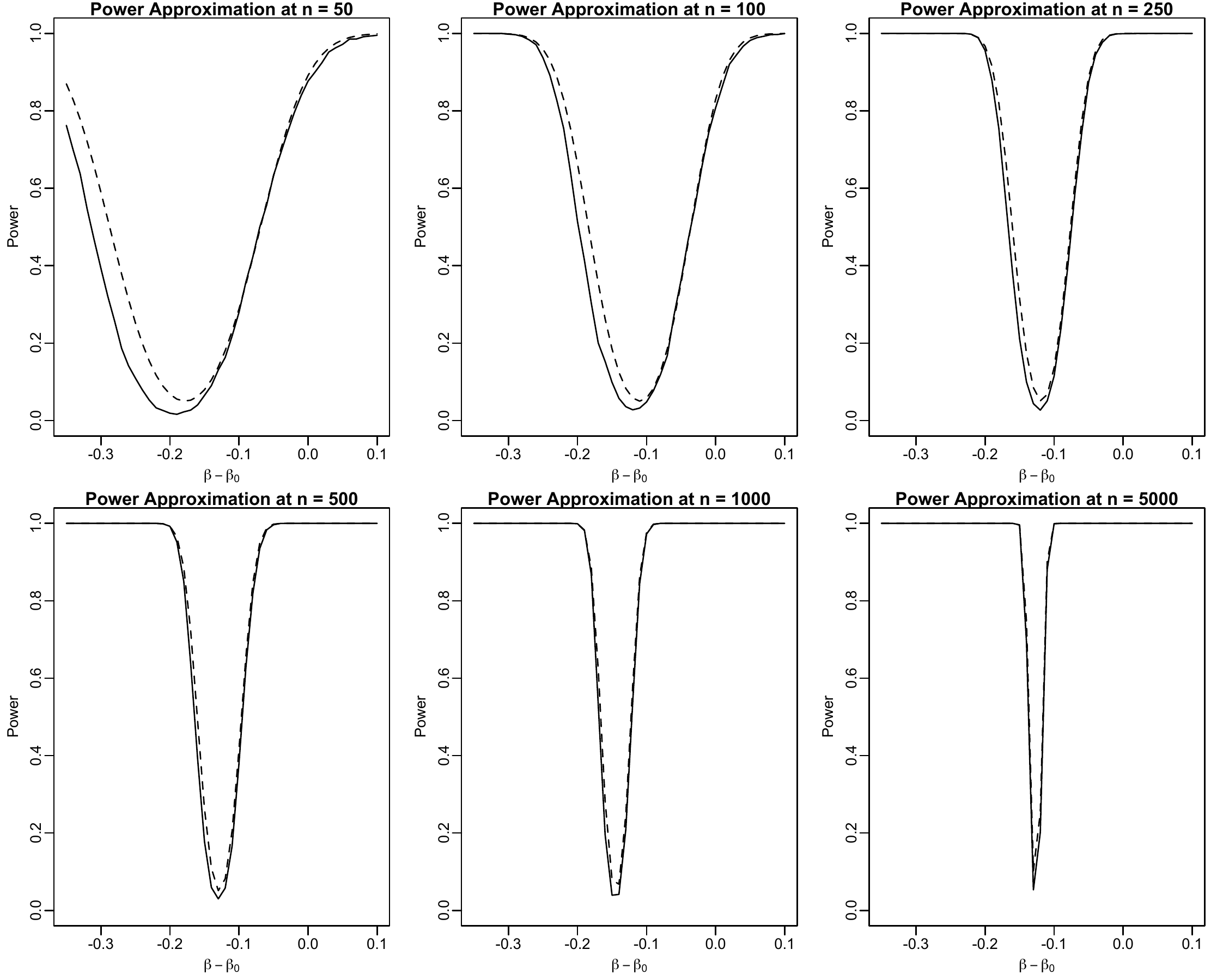}
\caption{Power curves for two-stage least squares under invalid instruments. The dotted line represents the asymptotic power curve and the solid line represents the simulated power. We set $B = \{1,2\}$ and $B^* = \{1,2,3\}$; $B$ does not contain all the invalid instruments.}
\label{fig:powerApprox2}
\end{figure}

The final simulation in Figure \ref{fig:powerApprox4} is the same as Figure \ref{fig:powerApprox2} except $B = \{1\}$ and $B^* = \{1,2,3\}$. Figure \ref{fig:powerApprox1} shows, again, the asymptotic power curve matches the empirical power curve as sample size grows. If we set $\boldsymbol{\pi}^* = (2,2,2,0,\ldots,0)$ or $\boldsymbol{\pi}^* =(0.5,0.5,0.5,0,\ldots,0)$, we see similar results.

\begin{figure}[H]
\includegraphics[width=\textwidth]{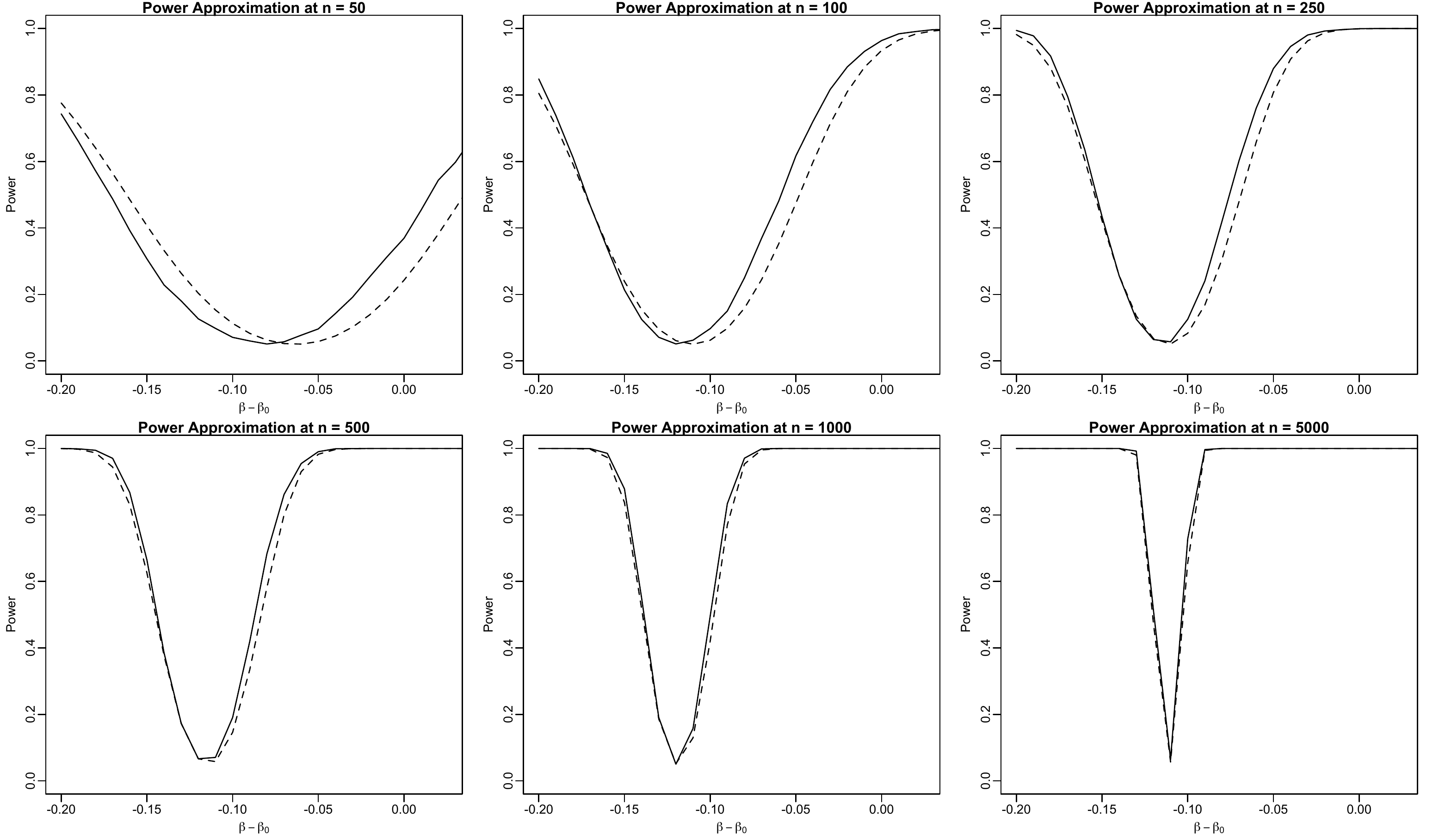}
\caption{Power curves for two-stage least squares under invalid instruments. The dotted line represents the asymptotic power curve and the solid line represents the simulated power. We set $B = \{1\}$ and $B^* = \{1,2,3\}$; $B$ does not contain all the invalid instruments.}
\label{fig:powerApprox4}
\end{figure}

In summary, our simulation results show promise that the asymptotic power under invalid instruments for two-stage least squares is a good approximation for the finite sample power of two-stage least squares under invalid instruments.

\section{Additional Simulations Results}
\label{sec:addition}
\subsection{Choice of Test Statistics for Method 1}
In this section, we present the additional simulations results with the same settings of the union method.  First, Table~\ref{tab:strong20_fixU10_length} examines the median length of 95\% confidence intervals under strong instruments when we vary the number of invalid instruments from $s^{*} = 0$ to $s^{*} = 9$ and use a fixed upper bound of $\bar{s}=10$. The Anderson-Rubin test and the conditional likelihood ratio test produce confidence intervals with almost equivalent lengths when used in the union  method. In contrast, the oracle confidence intervals based on the conditional likelihood ratio test and two-stage least squares are shorter than those based on the Anderson-Rubin test. 

\begin{table}[H]
\centering
\resizebox{\textwidth}{!}{\begin{tabular}{rlllllllllll}
 \hline
Method & Test Statistic & $s^{*}$=0 & $s^{*}$=1 & $s^{*}$=2 & $s^{*}$=3 & $s^{*}$=4 & $s^{*}$=5 & $s^{*}$=6 & $s^{*}$=7 & $s^{*}$=8 & $s^{*}$=9 \\ 
  \hline
  Union & TSLS & 0.628 & 1.573 & 1.824 & 1.946 & 2.027 & 2.071 & 2.122 & 2.113 & 2.101 & 2.039 \\ 
 &AR & 0.676 & 1.659 & 1.946 & 2.087 & 2.173 & 2.218 & 2.281 & 2.273 & 2.256 & 2.190 \\ 
  & CLR & 0.676 & 1.659 & 1.946 & 2.087 & 2.173 & 2.218 & 2.281 & 2.273 & 2.256 & 2.190 \\ 
 Oracle &  TSLS &  0.105 & 0.111 & 0.117 & 0.126 & 0.136 & 0.148 & 0.167 & 0.193 & 0.235 & 0.338 \\ 
 & AR & 0.168 & 0.176 & 0.181 & 0.190 & 0.202 & 0.211 & 0.228 & 0.251 & 0.283 & 0.350 \\ 
 & CLR & 0.106 & 0.113 & 0.119 & 0.128 & 0.138 & 0.151 & 0.170 & 0.197 & 0.241 & 0.350 \\ 
   \hline
\end{tabular}}
\caption{\label{tab:strong20_fixU10_length} TSLS, two-stage least squares; AR, Anderson-Rubin test; CLR, conditional likelihood ratio test. The median lengths of the 95\% confidence intervals when instruments are strong and $\bar{s} = 10$.}
\end{table}

Tables~\ref{tab:weak20_U5} and~\ref{tab:weak20_length} present the coverage proportion and the median lengths of 95\% confidence intervals when instruments are weak and $\bar{s} = 5$. Pretesting methods were not considered because the Sargan test is known to perform poorly with weak instruments. The coverage proportion presented in Table~\ref{tab:weak20_U5} suggests that the two-stage least squares method performs poorly with weak instruments, producing less than 95\% coverage rate even in the oracle setting. Like the strong instruments case, as the number of invalid instruments, $s^{*}$, increases, the Anderson-Rubin test and the conditional likelihood ratio test under the union method get closer to 95\% coverage.

\begin{table}[ht]
\centering
\resizebox{0.6\textwidth}{!}{\begin{tabular}{rllllll}
  \hline
Method & Test Statistic & $s^{*}$=0 & $s^{*}$=1 & $s^{*}$=2 & $s^{*}$=3 & $s^{*}$=4 \\
  \hline
Naive & TSLS & 66.9 & 2.4 & 0.4 & 0.0 & 0.0 \\
  & AR & 93.0 & 0.0 & 0.0 & 0.0 & 0.0 \\
 & CLR & 94.6 & 0.0 & 0.0 & 0.0 & 0.0 \\
 Union & TSLS & 100.0 & 99.7 & 99.1 & 95.3 & 79.2 \\
 & AR & 100.0 & 100.0 & 100.0 & 99.5 & 95.0 \\
  & CLR & 100.0 & 100.0 & 99.9 & 99.1 & 94.6 \\
 Oracle & TSLS & 66.9 & 72.4 & 70.9 & 73.8 & 77.6 \\
  & AR & 93.0 & 94.5 & 93.0 & 94.3 & 95.0 \\
  & CLR & 94.6 & 95.1 & 95.7 & 94.4 & 94.6 \\
   \hline
\end{tabular}}
\caption{\label{tab:weak20_U5} TSLS, two-stage least squares; AR, Anderson-Rubin test; CLR, conditional likelihood ratio test. Comparison of coverage rates of 95\% confidence intervals under weak instruments and $\bar{s} = 5$.}
\end{table}

\begin{table}[ht]
\centering
\resizebox{0.6\textwidth}{!}{\begin{tabular}{rllllll}
  \hline
Case & Test & $s^{*}$=0 & $s^{*}$=1 &  $s^{*}$=2 & $s^{*}$=3 & $s^{*}$=4 \\
  \hline
  Our method 
 & AR &   5.737 & 3.900 & 8.523 & 19.352 & 28.609 \\ 
 & CLR & 1.821 &  $\infty$ &  $\infty$ & $\infty$ &  $\infty$ \\ 
 Oracle 
  & AR & 0.876 & 0.944 & 0.964 & 1.051 & 1.173 \\
 & CLR &  0.521 & 0.557 & 0.598 & 0.652 & 0.717 \\ 
   \hline
\end{tabular}}
\caption{\label{tab:weak20_length} TSLS, two-stage least squares; AR, Anderson-Rubin test; CLR, conditional likelihood ratio test. Comparison of median lengths between 95\% confidence intervals under weak instruments and $\bar{s} = 5$. }
\end{table}
Table~\ref{tab:weak20_length} shows the median lengths of the 95\% confidence intervals when instruments are weak and $\bar{s}= 5$. The conditional likelihood ratio test under the union method produces infinite intervals. These infinite lengths suggest that weak instruments can greatly amplify the bias caused by invalid instruments, thereby forcing our method to produce infinite intervals to retain honest coverage; see \citet{small_war_2008} for a similar observation. The Anderson-Rubin test under our method also produces very wide confidence intervals. In contrast, the oracle intervals produce finite intervals since instrumental validity is not an issue; note that if the instrument is arbitrary weak, infinite confidence intervals are necessary for honest coverage ~\citep{dufour_impossibility_1997}. Finally, the oracle conditional likelihood ratio intervals is shorter than the oracle Anderson-Rubin intervals, as expected, but the relationship is reversed in the union method where we do not know which instruments are invalid.

\subsection{Comparison of Methods 1 and 2}
\begin{figure}[H]
	\centering
	\begin{subfigure}[b]{0.4\textwidth}
		\includegraphics[width=\textwidth]{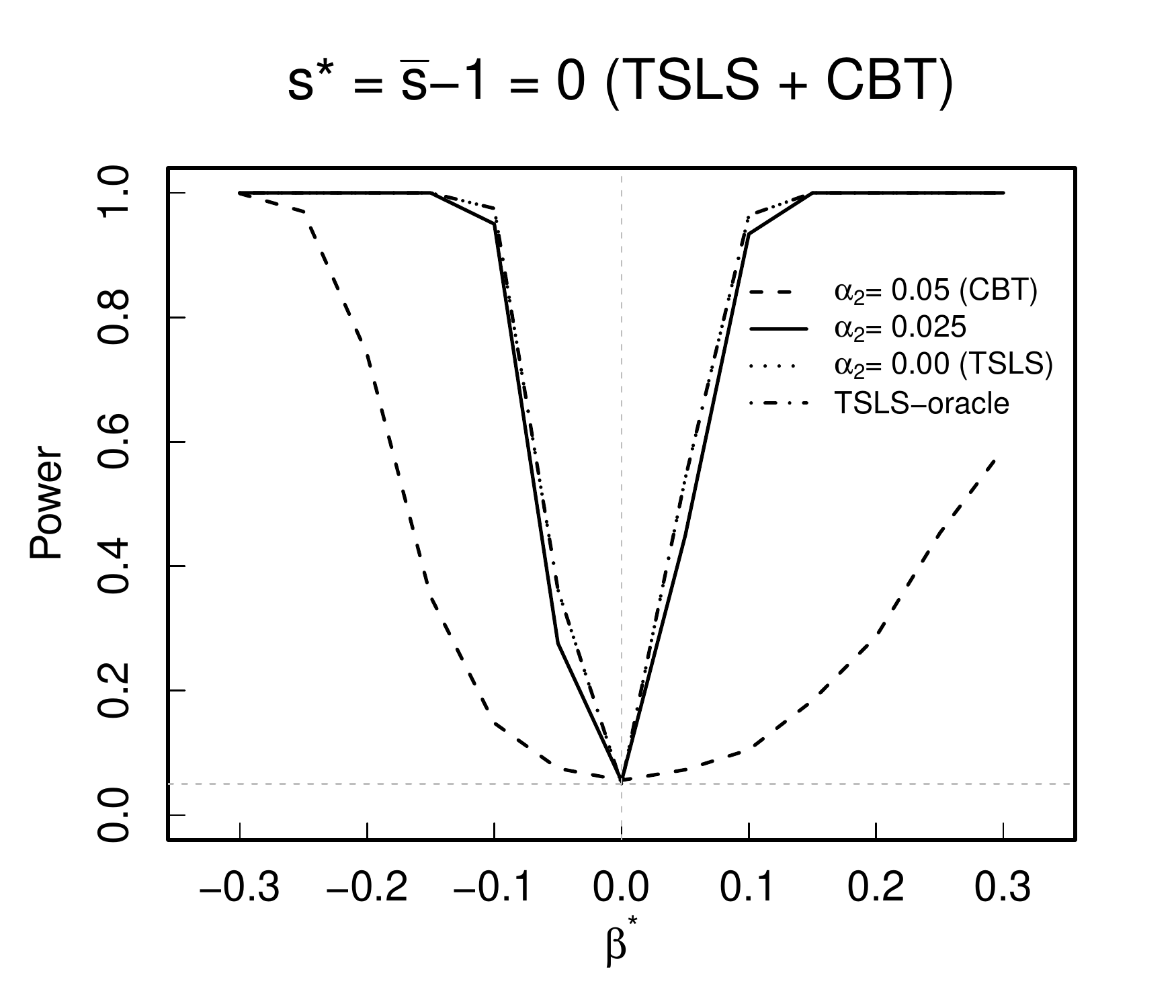}
	\end{subfigure}
	\begin{subfigure}[b]{0.4\textwidth}
		\includegraphics[width=\textwidth]{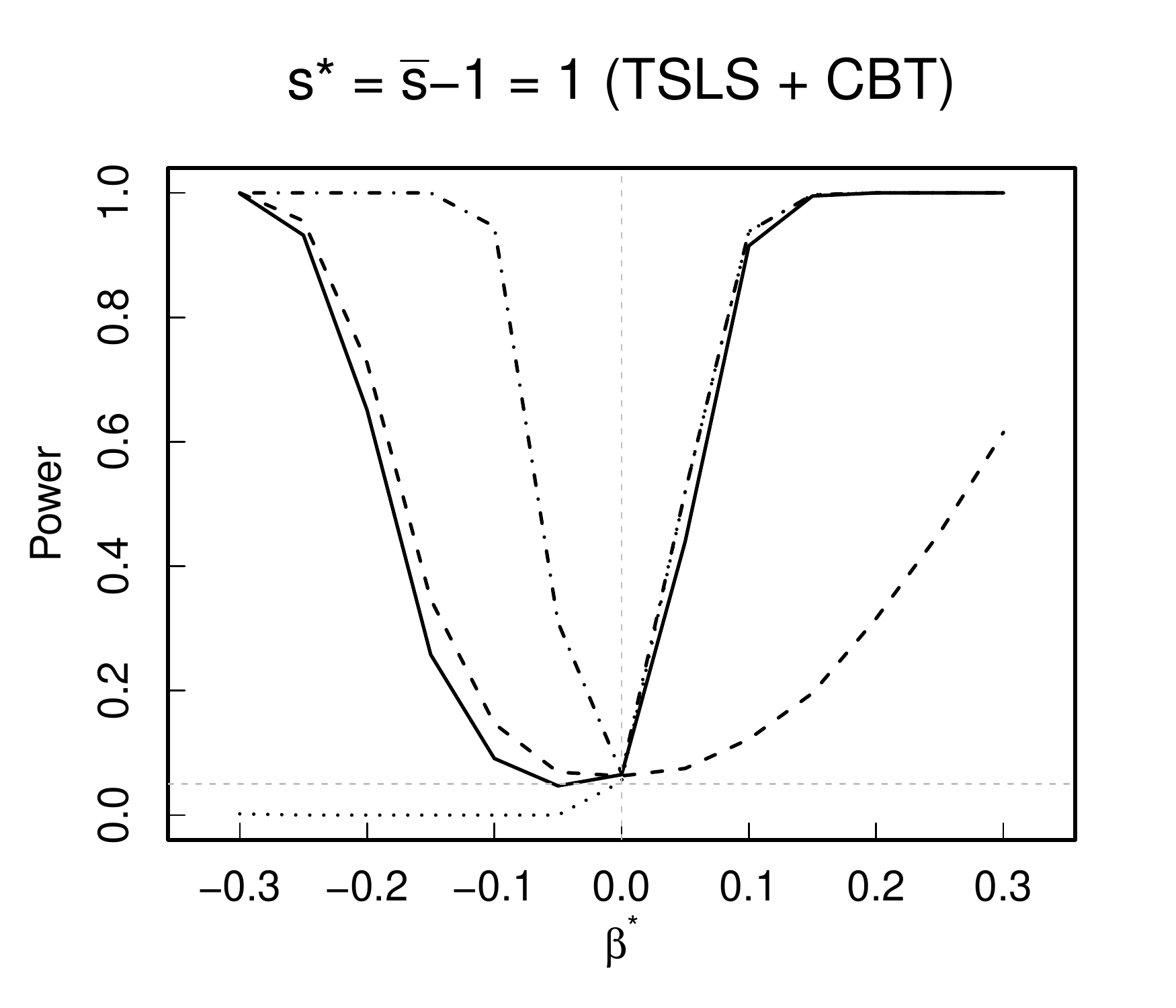}
	\end{subfigure}
	\begin{subfigure}[b]{0.4\textwidth}
		\includegraphics[width=\textwidth]{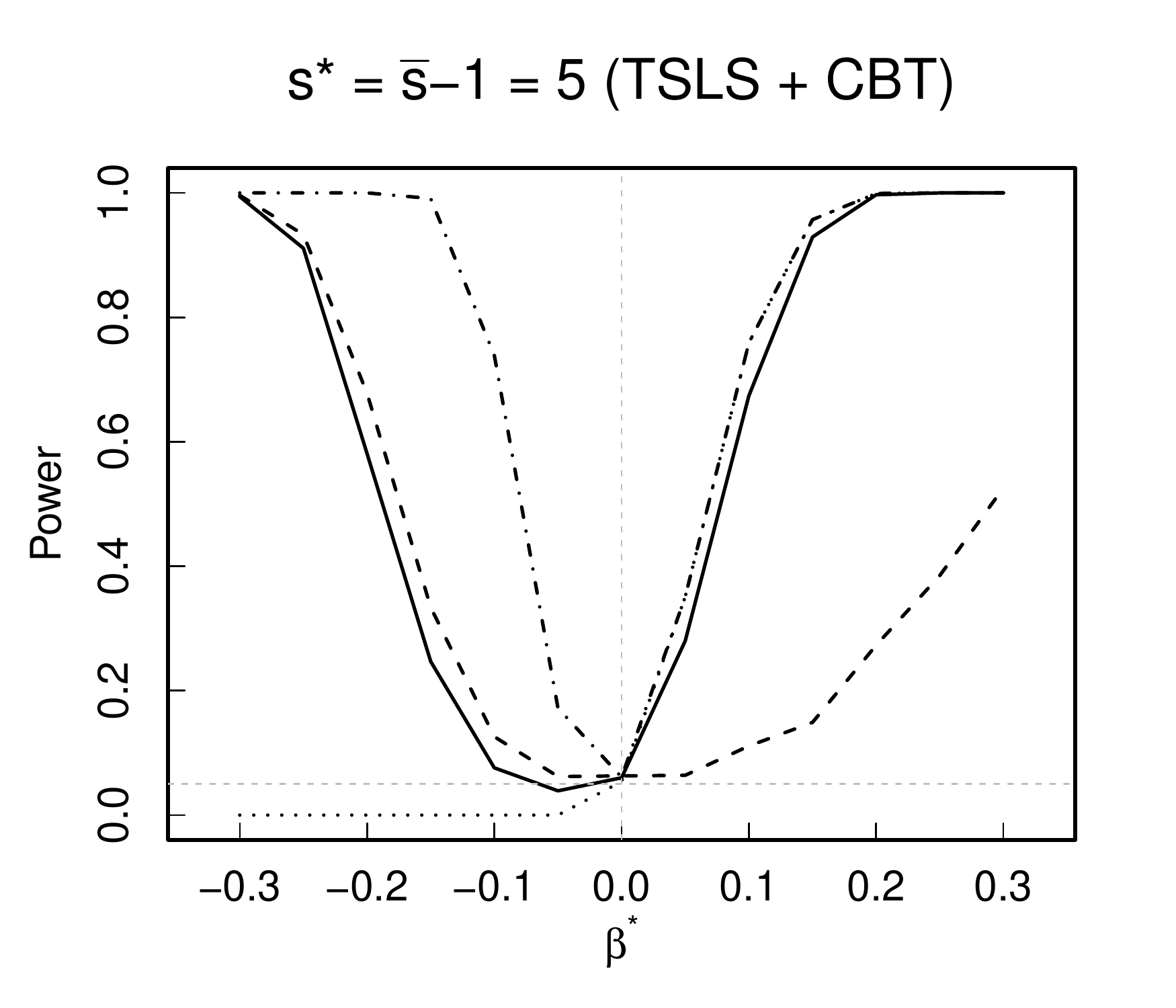}
	\end{subfigure}	
	\begin{subfigure}[b]{0.4\textwidth}
		\includegraphics[width=\textwidth]{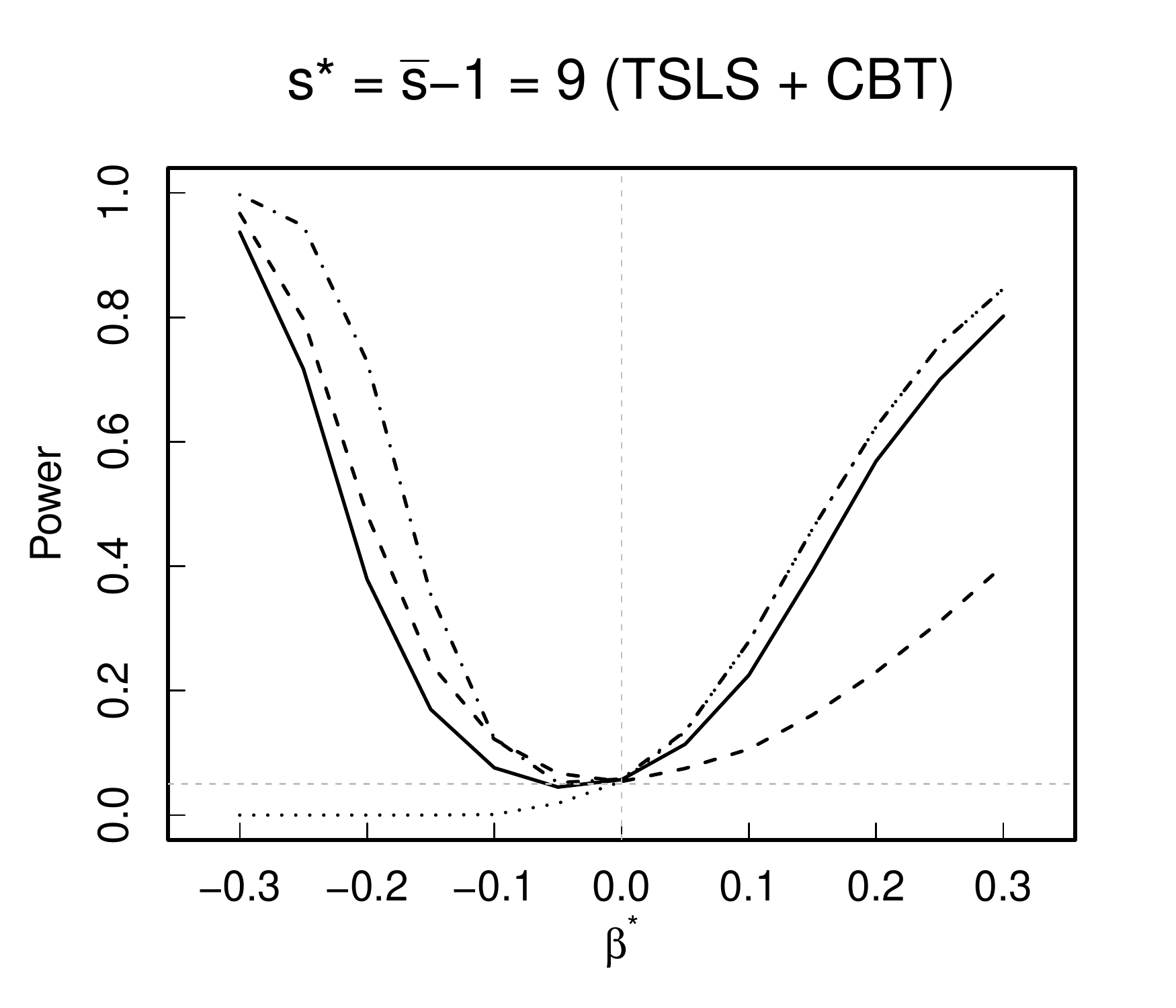}
	\end{subfigure}
	\caption{\label{fig:strong_tsls} TSLS: two-stage least squares; the CBT: collider bias test. Power of different methods under strong instruments with different numbers of invalid instruments. We fix $\alpha = 0.05$ and vary $\alpha_{2}$ to be $0.0, 0.025$, or $0.05$. When $\alpha_{2} = 0.05$, the combined test is equivalent to the collider bias test. When $\alpha_{2} = 0$, the combined test is equivalent to the union method. The oracle TSLS is a test based on TSLS that knows exactly which instruments are valid.}
\end{figure}
Figure \ref{fig:strong_tsls} compares the power between the union method using two-stage least squares, the collider bias test, and the combined method; the setup is identical to the main text. When the two-stage least squares method is used in the union procedure, Figure~\ref{fig:strong_tsls} demonstrates the asymmetric patterns in power (dotted lines) starting from $s^{*}=1$, and the collider bias test has better power when $\beta^{*} < 0$ and $s^{*} > 0$. The combined test (solid lines) essentially has the best of both worlds, where it achieves good power across most values of the alternative.

\subsection{Binary Instruments and Outcome}
We consider a simulation study when the instruments or outcomes are binary and assess the sensitivity of the methods' assumptions to distributional assumptions. To create binary instruments, we replace the model for Normal model in the main text with a Bernouilli model.
\[ 
Z_{ij} \overset{i.i.d.}{\sim} B(0.3)~i=1,2,\ldots,n;~j=1,2,\ldots, L=10.
\]
Figure~\ref{fig:binaryIV} shows the results of the simulation setting in Figure~\ref{fig:strong_either}, but with binary instruments. We see that the power curves across all methods are nearly identical to each other, suggesting that the Normality assumption on $\mathbf{Z}$ for the collider bias test can be relaxed.
\begin{figure}[!ht]
	\centering
	\begin{subfigure}[b]{0.4\textwidth}
		\includegraphics[width=\textwidth]{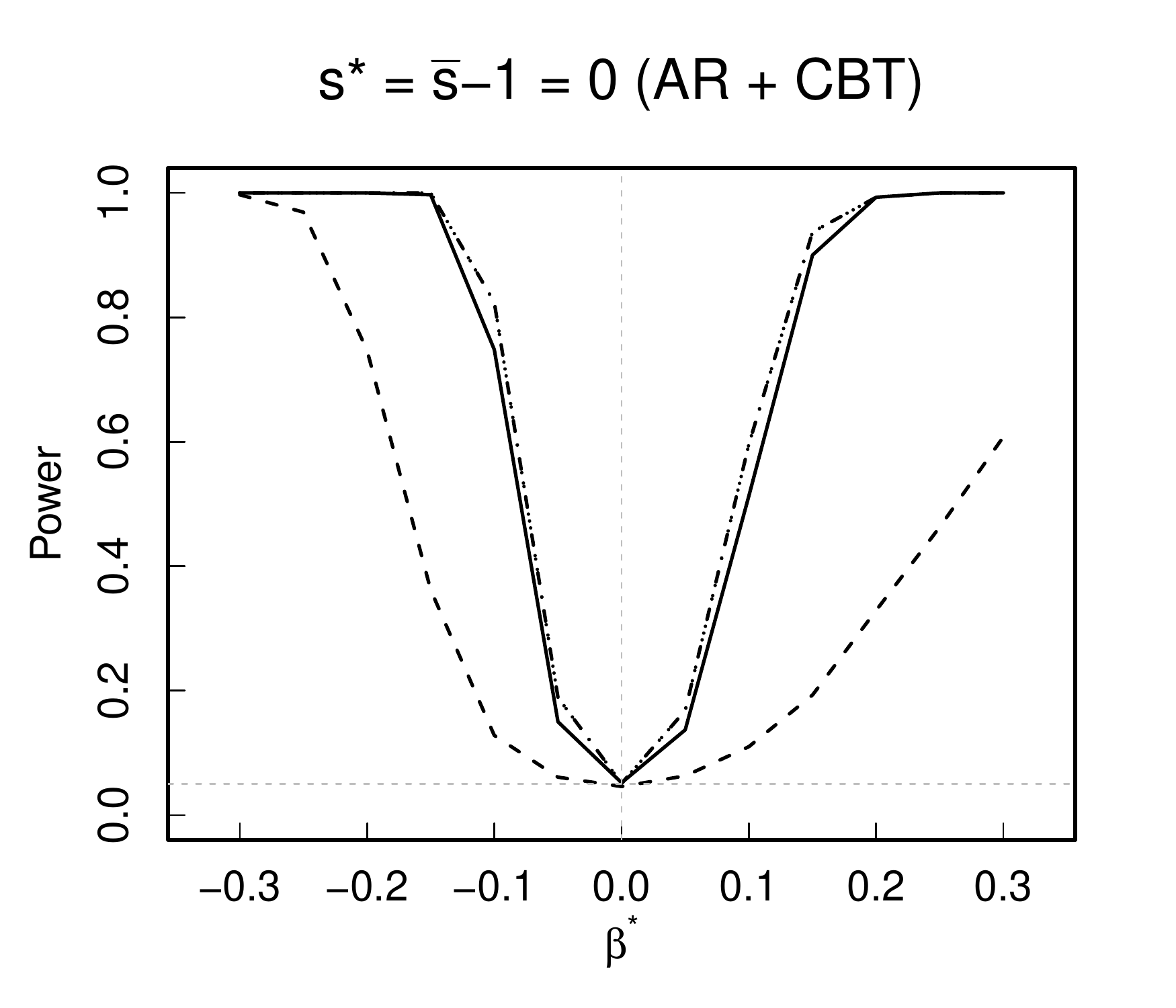}
	\end{subfigure}
	\begin{subfigure}[b]{0.4\textwidth}
		\includegraphics[width=\textwidth]{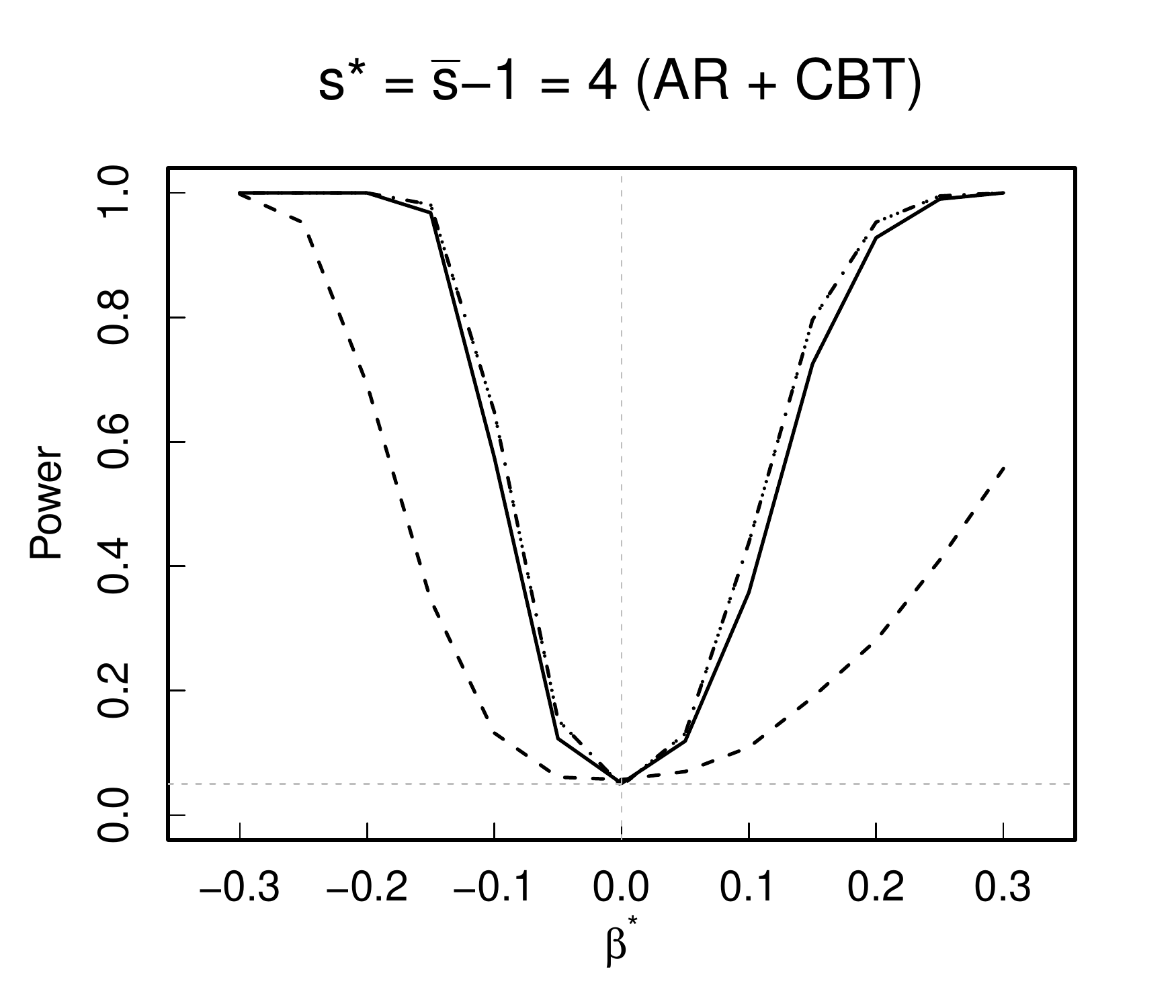}
	\end{subfigure}
	\begin{subfigure}[b]{0.4\textwidth}
		\includegraphics[width=\textwidth]{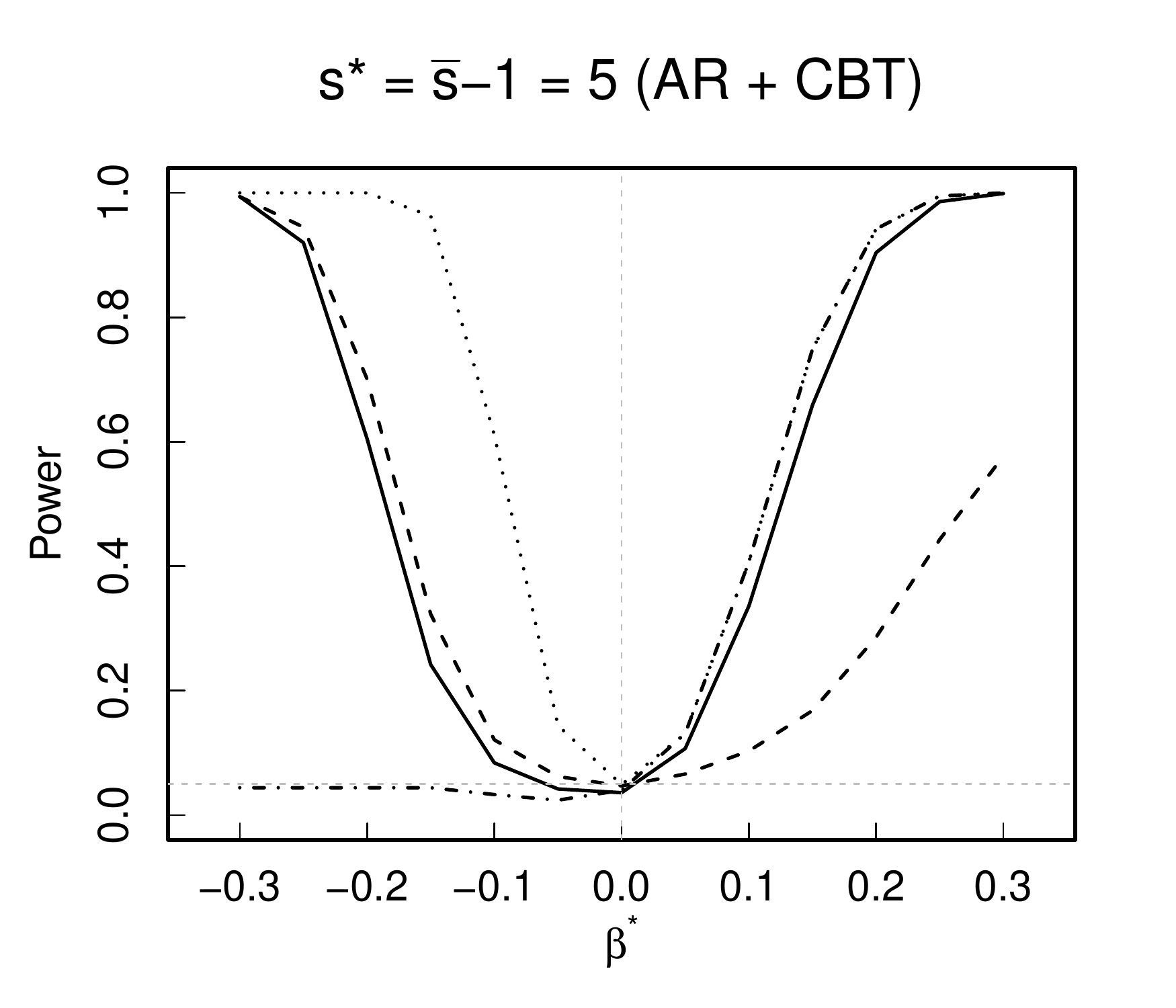}
	\end{subfigure}	
	\begin{subfigure}[b]{0.4\textwidth}
		\includegraphics[width=\textwidth]{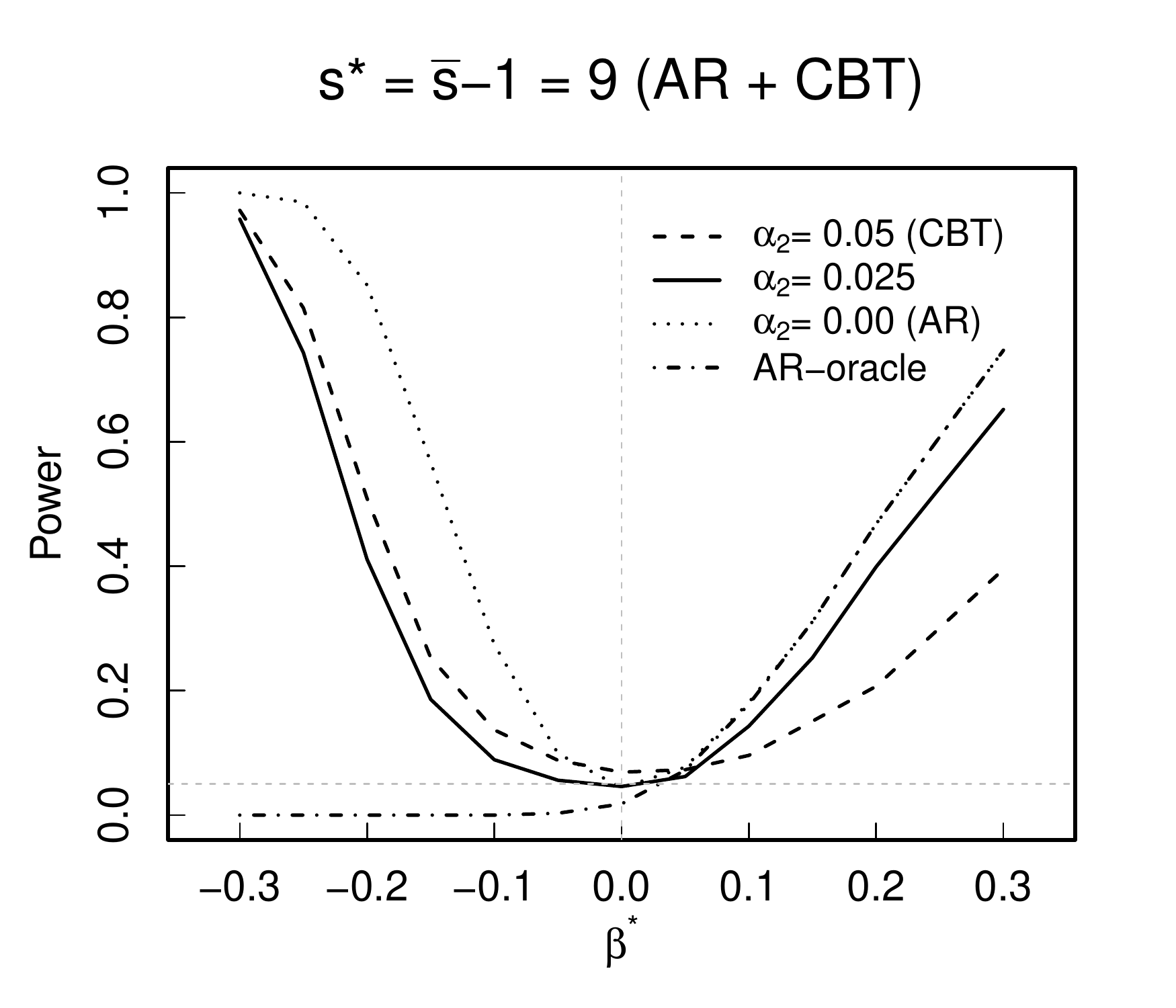}
	\end{subfigure}
	\caption{\label{fig:binaryIV} AR: Anderson-Rubin test; CBT: collider bias test. Power of different methods under strong, binary instruments with different numbers of invalid instruments. We fix $\alpha = 0.05$ and vary $\alpha_{2}$ to be $0.0, 0.025$, or $0.05$. When $\alpha_{2} = 0.05$, the combined test is equivalent to the collider bias test. When $\alpha_{2} = 0$, the combined test is equivalent to the union method. The oracle AR test is the AR test that knows exactly which instruments are valid.
}
\end{figure}
Next, we replace the linear model for the outcome in the main text with a logistic regression model 
\[
Y_{i} \overset{ind}{\sim} B\left( \mbox{logistic} ( \mathbf{Z}_{i\cdot}^T \boldsymbol{\tilde{\pi}}^* + D_i \tilde{\beta}^* + \epsilon_i ) \right),
\]
where $\mbox{logistic}(x) = (1 + e^{-x})^{-1}$. Then, $\tilde{\beta}^{*}$ is the expected change in log odds for a one-unit increase in $D_{i}$ when $\mathbf{Z}_{i \cdot}$ is held constant. The null hypothesis of no effect in the linear model $H_0: \beta^* = 0$ implies $\tilde{\beta}^* = 0$. Additionally, if instruments are independent of each other, $\pi_j^*= 0$ implies $\tilde{\pi}_j^* = 0$. Thus, the number of invalid instruments $s^*$ and its upper bound $\bar{s}$ match with the the number of zeros in $\tilde{\pi}^*$ and its upper bound in a binary model. Figure~\ref{fig:binaryY} replicates the simulation setting in Figure~\ref{fig:strong_either}, but with binary outcomes. The overall shape and trend of the power curves look very similar to Figure~\ref{fig:strong_either} across different methods. 

\begin{figure}[H]
	\centering
	\begin{subfigure}[b]{0.4\textwidth}
		\includegraphics[width=\textwidth]{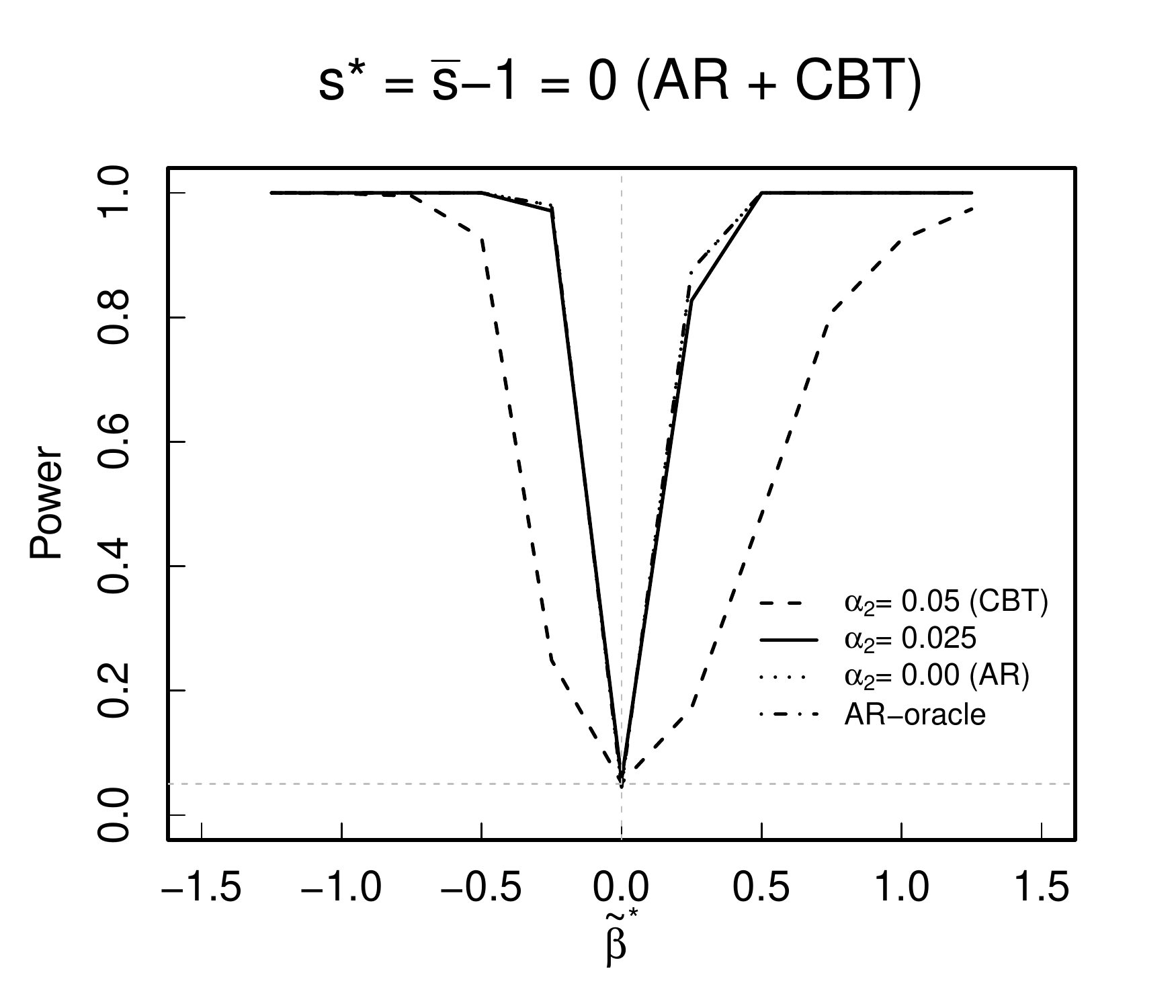}
	\end{subfigure}
	\begin{subfigure}[b]{0.4\textwidth}
		\includegraphics[width=\textwidth]{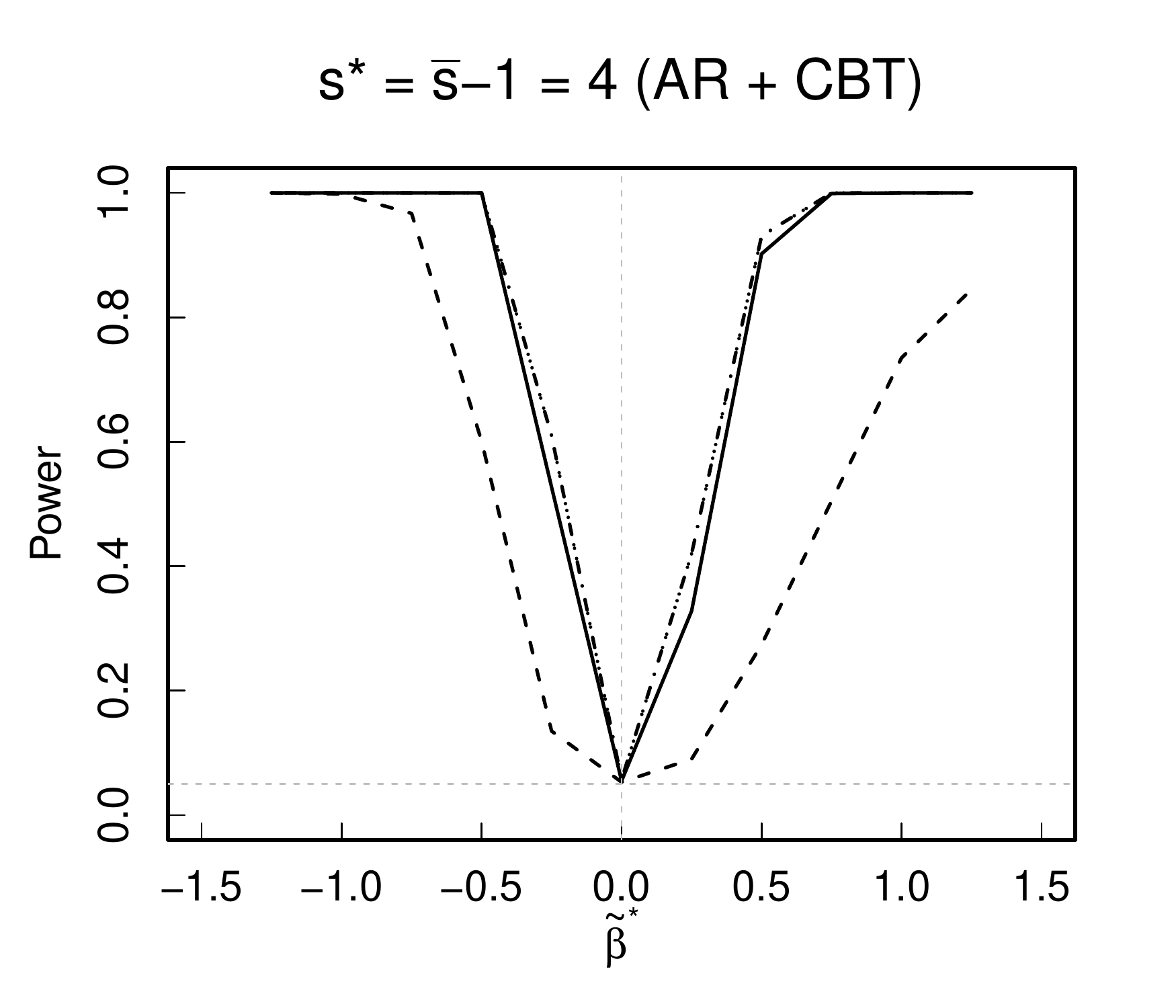}
	\end{subfigure}
	\begin{subfigure}[b]{0.4\textwidth}
		\includegraphics[width=\textwidth]{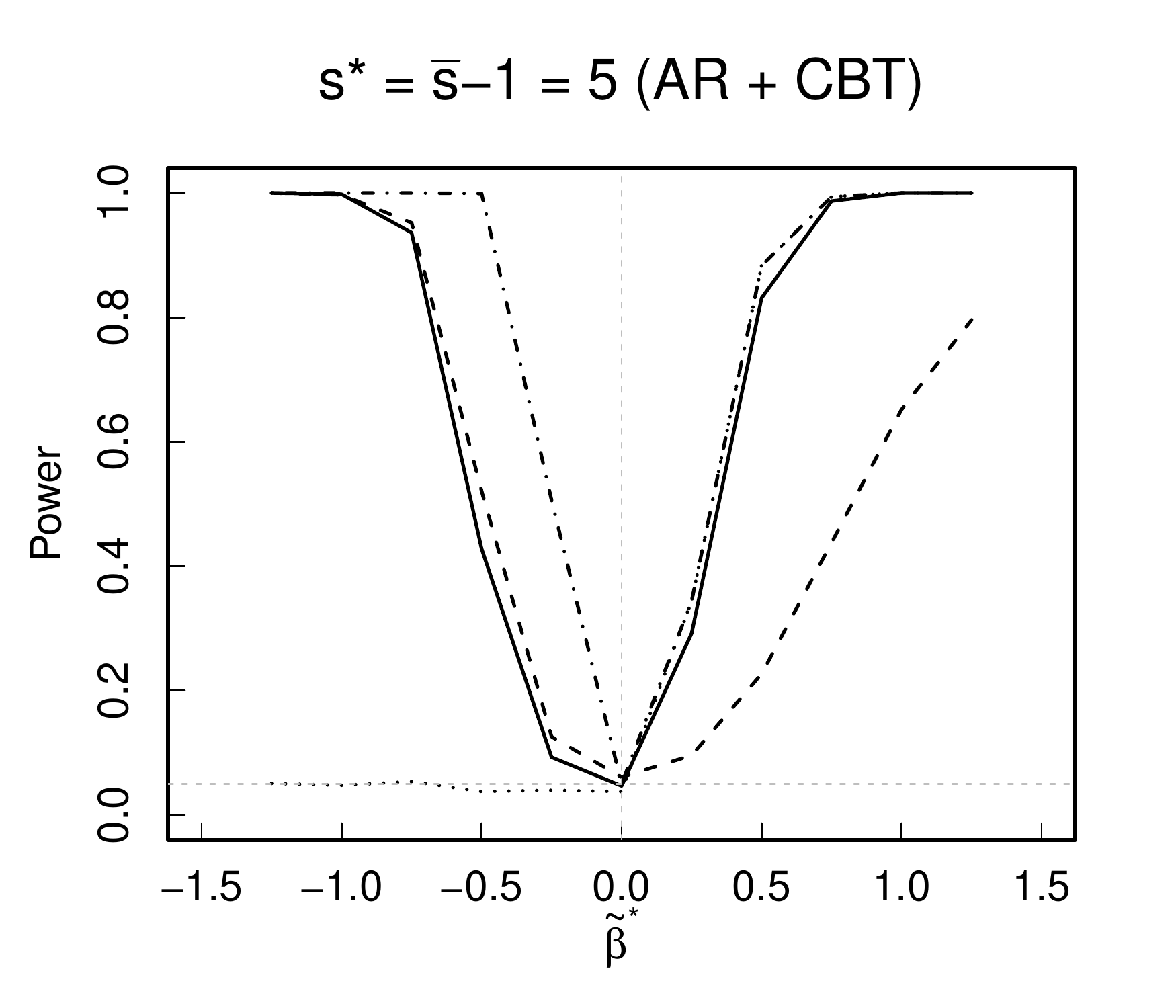}
	\end{subfigure}	
	\begin{subfigure}[b]{0.4\textwidth}
		\includegraphics[width=\textwidth]{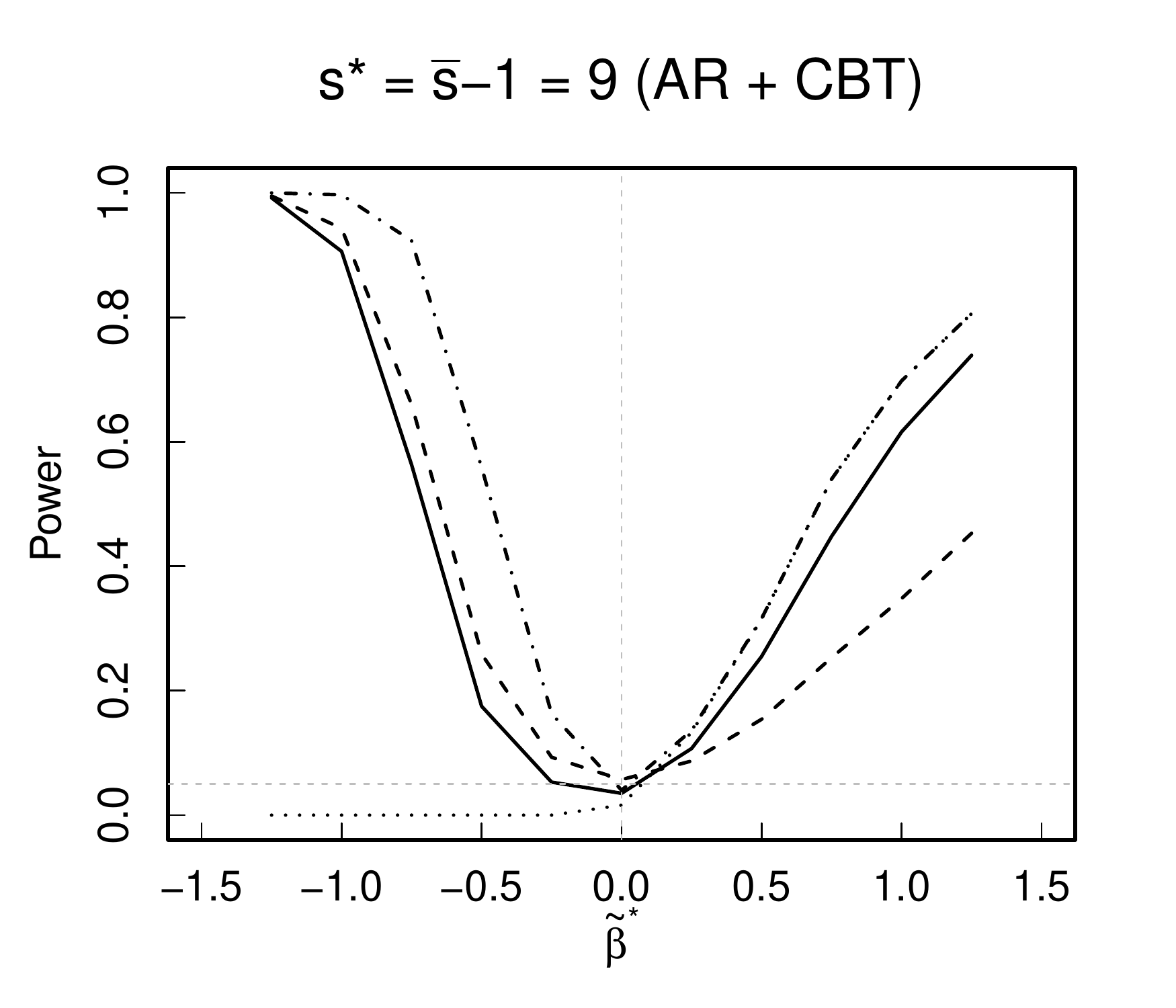}
	\end{subfigure}
	\caption{\label{fig:binaryY} AR: Anderson-Rubin test; CBT: collider bias test. Power of different methods under strong instruments with different numbers of invalid instruments and the outcome is binary. We fix $\alpha = 0.05$ and vary $\alpha_{2}$ to be $0.0, 0.025$, or $0.05$. When $\alpha_{2} = 0.05$, the combined test is equivalent to the collider bias test. When $\alpha_{2} = 0.0$, the combined test is equivalent to the union method. The oracle AR test is the AR test that knows exactly which instruments are valid.}
\end{figure}

\section{Additional Data Analysis}

Table~\ref{tab:SNP_D} presents the location of each SNP used in Section~\ref{sec:FHS}, each SNP's marginal association with LDL-C (exposure) and CVD incidence (outcome) via a linear regression and logistic regression, respectively. As expected, we see that all the instruments are strong, with t-statistics concerning the marginal association of $Z_j$ and $D$ exceeding $1.8$, which suggests that the ten instruments reasonably satisfy (A1). 
\begin{table}[H]
\centering
\resizebox{\textwidth}{!}{\begin{tabular}{rl|ll|ll}
  \hline
\multirow{2}{*}{SNP ($Z_{j}$)} & \multirow{2}{*}{Position} & \multicolumn{2}{c}{$D \sim Z_{j}$} & \multicolumn{2}{c}{$Y \sim Z_{j}$}
\\ & & Estimate (Std. Error) & t value (p-value) & Estimate (Std. Error) & t value (p-value) \\ 
  \hline
 rs11591147 & chr1:55039974 &         15.796  (3.339) &          4.731 (0.000) &         -0.053 (0.038) &         -1.378 (0.168) \\ 
  rs10455872 & chr6:160589086 &          4.012 (1.750) &          2.292 (0.022) &          0.012 (0.020) &          0.602 (0.547) \\ 
  rs646776 & chr1:109275908 &         6.637 (1.013) &         6.550 (0.000) &         0.004  (0.012) &         0.333 (0.739) \\ 
  rs693 & chr2:21009323 &         4.278 (0.824) &         5.195 (0.000) &         0.010 (0.009) &         1.030 (0.303) \\ 
  rs2228671 & chr19:11100236 &          4.998 (1.278) &          3.911 (0.000) & -0.037 (0.014) &  -2.592 (0.010) \\ 
  rs2075650 & chr19:44892362 &          5.099  (1.278) &          3.989 (0.000) &          0.015 (0.014) &          1.046 (0.295) \\ 
  rs4299376 & chr2:43845437 &          4.050 (0.883) &          4.589 (0.000) &          0.020 (0.010) & 1.981 (0.048) \\ 
  rs3764261 & chr16:56959412 &          1.936 (0.918) &          2.108 (0.035) &         -0.001 (0.010) &    -0.079 (0.937) \\ 
  rs12916 & chr5:75360714 & 1.776 (0.866) &  2.051 (0.040) &         0.009 (0.010) &   0.918 (0.359) \\ 
  rs2000999 & chr16:72074194 &  1.963 (1.048) &         1.872 (0.061) &         0.020 (0.012) &  1.696 (0.090) \\  
  \hline
\end{tabular}}
\caption{\label{tab:SNP_D} chr: chromosome. SNPs that are located in the same chromosome (e.g. rs2228671 and rs2075650) are in linkage equilibrium. Estimates and t-values are derived from a linear regression of $D$ or a logistic regression of $Y$ on $Z_{j}$. }
\end{table}
To mitigate concerns for population crypticness, we selected one subject at random from each family in the Offspring Cohort linked by sibling or marriage relationships; note that there is no parent-child relationship within Offspring Cohort. By doing so, we retain only 60\% of the subjects ($n=1,726$) used in the main study. This may reduce power of our analysis~\citep{pierce2010power}, but possibly lead to a more valid analysis of the true effect \citep{lee2019network}. All subsequent analysis will use $n=1,726$ subject.

\begin{table}[H]
\centering
\resizebox{\textwidth}{!}{\begin{tabular}{rl|ll|ll}
  \hline
\multirow{2}{*}{SNP ($Z_{j}$)} & \multirow{2}{*}{Position} & \multicolumn{2}{c}{$D \sim Z_{j}$} & \multicolumn{2}{c}{$Y \sim Z_{j}$}
\\ & & Estimate (Std. Error) & t value (p-value) & Estimate (Std. Error) & t value (p-value) \\ 
  \hline
rs562338 & chr2:21065449 & 2.656 (1.384) & 1.919 (0.055) & 0.026 (0.015) & 1.724 (0.085) \\
rs4299376 & chr2:43845437 & 4.376 (1.190) & 3.676 (0.000) & 0.017 (0.013) & 1.262 (0.207) \\
rs2000999 & chr16:72074194 & 1.609 (1.396) & 1.153 (0.249) & 0.024 (0.015) & 1.528 (0.127) \\
rs17321515 & chr8:125474167 & 1.720 (1.128) & 1.524 (0.128) & 0.018 (0.013) & 1.470 (0.142) \\
   \hline
\end{tabular}}
\caption{\label{tab:foursnps} The location of four SNPs $Z_{j}~(j=1,2,3,4)$ used in the second analysis and the inference on the regression coefficients from a linear regression of $D$ and a logistic regression of $Y$ on each SNP.}
\end{table}

For the second analysis, we consider four SNPs, rs562338, rs4299376, rs2000999, and rs17321515 (see Table~\ref{tab:foursnps}) as candidate instruments. Table~\ref{tab:fourtab1} summarizes some results from the union procedure using five different methods. The conditional likelihood ratio test at $\alpha_{1} = 0.05$ results in rejecting the null causal effect while allowing at most two invalid instruments among four; whereas the two-stage least squares method and the conditional likelihood ratio test with pretesting allow at most one invalid instrument and the Anderson-Rubin test and the two-stage least squares method with pretesting require no invalid instrument at all to be able to reject the null effect, both at $
\alpha_{1} = 0.05$ level.

\begin{table}[H]
\centering
\resizebox{0.6\textwidth}{!}{\begin{tabular}{rlllll}
  \hline
 & AR & CLR & TSLS & SarganTSLS & SarganCLR \\ 
  \hline
  $\mathbf{\alpha_{1}  = 0.05}$ \\ 
$\bar{s}$=1 & \cellcolor{Gray} 0.000 & \cellcolor{Gray} 0.002 & \cellcolor{Gray} 0.002 & \cellcolor{Gray} 0.001 & \cellcolor{Gray} 0.001 \\ 
  $\bar{s}$=2  &  -0.001 & \cellcolor{Gray} 0.001 & \cellcolor{Gray} 0.001 & -0.000 & \cellcolor{Gray} 0.000 \\ 
 $\bar{s}$=3  & -0.001 & \cellcolor{Gray} 0.000 & -0.001 & -0.003 & -0.002 \\ 
  \hline
$\mathbf{\alpha_{1}  = 0.025}$ \\ 
  $\bar{s}$=1 &  -0.001 & \cellcolor{Gray}0.001 & \cellcolor{Gray}0.001 & \cellcolor{Gray}0.000 & \cellcolor{Gray}0.001 \\ 
  $\bar{s}$=2  &  -0.001 & \cellcolor{Gray}0.000 & -0.000 & -0.001 & -0.001 \\ 
  $\bar{s}$=3 &   -0.003 & -0.002 & -0.003 & -0.004 & -0.004 \\ 
   \hline
\end{tabular}}
\caption{\label{tab:fourtab1} Lower bound of an one-sided confidence interval for $\beta^{*}$ at different sizes $\alpha_{1} (=0.05, 0.025)$ for each union procedure assuming different values of $\bar{s}$. For the Sargan test, we use $\alpha_{s} = \alpha_{t} = \alpha_{1} / 2$. There are $L=4$ candidate SNPs.
Gray cells indicate a confidence interval that does not include the null effect.}
\end{table}

As an another tool to test, we implemented conditional independence test between the instruments and the outcome and obtained a likelihood ratio test statistic of $\lambda_{n} =  2.018$ based on $n=1726$ subjects, which fails to reject the null at $\alpha_{2} = 0.05$ even we set $\bar{s} = 1$. Therefore, the conclusion from the combined test of the union procedure and the conditional independence test
at $\{ (\alpha_{1}, \alpha_{2}) :~\alpha_{1} + \alpha_{2} = 0.05 \}$ depends solely on the results from the union procedure at $\alpha_{1}$ level. See Table~\ref{tab:secondtab} for the upper bound $\bar{s}$ on $s^{*}$ plus one to reject the null at different combinations of $(\alpha_{1}, \alpha_{2})$ if we could. At $\alpha_{1}=\alpha_{2}=0.025$ (gray cells in Table~\ref{tab:secondtab}), we do not reject the null using the Anderson-Rubin method, need less than two invalid instruments using the conditional likelihood ratio method with and without pretesting and the two-stage least squares method with pretesting, and need no invalid instrument without pretesting for the two-stage least squares method.

 \begin{table}[H]
\centering
\resizebox{0.9\textwidth}{!}{\begin{tabular}{r|rrrrrrrrrrr}
  \hline
$\alpha_{1}$ & 0 & 0.005 & 0.01 & 0.015 & 0.02 & \cellcolor{Gray} 0.025 & 0.03 & 0.035 & 0.04 & 0.045 & 0.05 \\ 
$\alpha_{2}$ & 0.05 & 0.045 & 0.04 & 0.035 & 0.03 & \cellcolor{Gray} 0.025 &  0.02 & 0.015 & 0.01 & 0.005 & 0 \\
  \hline
AR & NR & NR & NR & NR & NR & \cellcolor{Gray} NR & NR & NR & NR & NR & 1 \\ 
  CLR & NR & NR & 1 & 1 & 1 & \cellcolor{Gray} 2 & 2 & 2 & 2 & 2 & 3 \\ 
  TSLS & NR & NR & 1 & 1 & 1 & \cellcolor{Gray} 1 & 2 & 2 & 2 & 2 & 2 \\ 
  SarganTSLS & NR & 1 & 1 & 1 & 2 & \cellcolor{Gray} 2 & 2 & 2 & 2 & 2 & 2 \\ 
  SarganCLR & NR & 1 & 1 & 2 & 2 & \cellcolor{Gray} 2 & 2 & 2 & 2 & 3 & 3 \\ 
   \hline
\end{tabular}}
\caption{\label{tab:secondtab} 
Each cell indicates the smallest upper bound $\bar{s}$ needed to reject the null or not to reject the null (NR) using five different union procedures at $\alpha_{1}$ and the conditional independence test at $\alpha_{2}$. The gray cells indicate the results when the error is splitted equally, i.e. $\alpha_{1} = \alpha_{2} = 0.025$.}
\end{table}

\end{document}